\documentclass{lmcs} 
\pdfoutput=1

\usepackage{lastpage}
\renewcommand{\dOi}{14(3:23)2018}
\lmcsheading{}{1--\pageref{LastPage}}{}{}%
{Feb.~23,~2017}{Sep.~24,~2018}{}

\usepackage{floatflt}
\usepackage{wrapfig}
\usepackage{pstricks}
\usepackage{rotating}
\usepackage{subfigure}
\usepackage{amsmath}
\usepackage{amssymb}
\usepackage{diagr}
\usepackage[active]{srcltx} 
\usepackage{amsmath,amssymb}
\usepackage{booktabs}
\usepackage{xcolor}

\newenvironment{enumerate-} 
{\begin{enumerate}
    
   \setlength{\parskip}{-1ex}              
   \setlength{\itemsep}{1.5ex}             
}
{
 \end{enumerate}
}

\newcommand{\K}{{\mathcal K}}
\newcommand{\T}{{\mathcal T}}
\newcommand{\pK}{\Psi_{\K}}

\newcommand{\D}{\mathcal{D}}
\newcommand{\alg}{\mathcal} 

\newcommand{\A}{\alg{A}}
\newcommand{\B}{\alg{B}}
\newcommand{\Pred}{{\sf Pred}}
\newcommand{\union}{\cup}


\usepackage{hyperref}

\begin{document}

\title[On Interpolation and Symbol Elimination in Theory Extensions]{On Interpolation and Symbol Elimination in Theory Extensions}

\author[V. Sofronie-Stokkermans]{Viorica Sofronie-Stokkermans}	
\address{University Koblenz-Landau, Koblenz, Germany}	
\email{sofronie@uni-koblenz.de}  





\begin{abstract}
In this paper we study possibilities of interpolation and symbol elimination 
in extensions of a theory $\T_0$ with additional function symbols whose 
properties are axiomatised using a set of clauses. 
We analyze situations in which we can perform such tasks in a hierarchical 
way, relying on existing mechanisms for symbol elimination in $\T_0$.
This is for instance possible if the base theory allows quantifier elimination. 
We analyze possibilities of extending such methods to situations in which 
the base theory does not allow quantifier elimination but has a model 
completion which does. We illustrate the method on various examples.
\end{abstract}

\maketitle

\section{Introduction}

\noindent Many problems in computer science (e.g.\ in program
verification) can be reduced to checking satisfiability of ground 
formulae w.r.t.\ a theory which can be a standard theory (for instance 
linear arithmetic) or a complex theory (typically the extension of 
a base theory $\T_0$ with additional function symbols axiomatized by 
a set $\K$ of formulae, or a combination 
of theories). SMT solvers are tuned for efficiently checking 
satisfiability of ground formulae in increasingly complex theories;
the output can be 
``satisfiable'',  
``unsatisfiable'', 
or ``unknown'' (if incomplete 
methods are used, or else termination 
cannot be guaranteed). 

More interesting  is to go beyond yes/no answers, i.e.\ to consider 
parametric systems and infer constraints on parameters
(which can be values or functions)
which guarantee that certain 
properties are met (e.g.\ constraints 
which guarantee the unsatisfiability of ground clauses in suitable theory 
extensions). 
In \cite{sofronie-ijcar10,sofronie-cade13} -- in a context 
specially tailored for the parametric verification of safety properties in 
increasingly more complex systems -- we showed that such 
constraints could be generated in extensions of a theory allowing 
quantifier elimination. 

In this paper, we propose a symbol elimination method in 
theory extensions and analyze its
properties. 
We also discuss possibilities of applying such methods to 
extensions of theories which do not allow quantifier elimination 
provided that they have a model completion which does.

Another problem we analyze is interpolation (widely used in program verification
\cite{McMillanCAV03,McMillanProver04,McMillanSurvey05,McMillanRelationApproximation,Kapur-et-all-06}). 
Intuitively, interpolants can be used for describing separations 
between the sets of ``good'' and ``bad'' states; they can 
help to discover relevant predicates in predicate abstraction with 
refinement and for over-approximation in model checking. 
It often is desirable to obtain ``ground'' interpolants of 
ground formulae. 
The first algorithms for interpolant generation in program 
verification required explicit constructions of proofs
\cite{Krajicek97,McMillanProver04}, which in general is a relatively difficult task. 
In \cite{Kapur-et-all-06} 
the existence of ground interpolants for 
{\em arbitrary formulae} w.r.t.\ a theory $\T$ is studied. It is
proved that this is the case if and only if $\T$ allows quantifier
elimination, which limits the applicability of the
results in \cite{Kapur-et-all-06}.  
Symbol elimination 
(e.g.\ using resolution and/or superposition) 
has been used for interpolant generation in e.g.\ \cite{KovacsVoronkov}.  
In \cite{Sofronie-lmcs} we identify classes of 
local theory extensions in which interpolants can be computed 
hierarchically, using a method of computing interpolants in the 
base theory. \cite{rybal-sofronie}
proposes an algorithm for the generation of interpolants for 
linear arithmetic with uninterpreted function 
symbols which reduces the problem to constraint solving 
in linear arithmetic. In both cases,  when considering 
theory extensions $\T_0 \subseteq \T_0 \cup \K$
we devise ways of 
``separating'' the instances of axioms in $\K$ and 
of the congruence axioms. 

There also exist results which relate ground interpolation to 
amalgamation or the injection transfer property 
\cite{Jonsson65,Bacsich75,Wronski86,Ghilardi-2014,Ghilardi-2017}.  
We use such results for obtaining criteria which allow us to 
recognize theories with ground interpolation. 
However, in general just knowing that ground interpolants 
exist is not sufficient: we 
want to construct the interpolants
fast (in a hierarchical or modular way) and characterize situations 
in which we know which (extension) terms these interpolants contain. 
For this, \cite{Wies,Wies-journal} introduce the notion
of $W$-separability and study its links to a form of hierarchical
interpolation. We here make the results in \cite{Wies, Wies-journal} more precise, 
and extend them. 

\

\noindent The main results of this paper can be summarized as follows: 
\begin{itemize}
\item We link the existence (and computation) of ground interpolants in a theory $\T$ to
  their existence (and computation) in a model completion $\T^*$ of $\T$. 
\item We study possibilities of effective symbol elimination in theory
  extensions (based on quantifier elimination in the base theory or in
  a model completion thereof) and analyze the properties of the formulae 
obtained this way. 
\item We analyze possibilities of hierarchical interpolation in
  local theory extensions. Our analysis extends
  both results in \cite{Sofronie-lmcs} and results in \cite{Wies} 
  by avoiding the restriction to convex base theories. We explicitly 
point out all conditions needed for hierarchical interpolation and
show how to check them. 
\end{itemize}
This paper is an extended version of \cite{sofronie-ijcar-2016}; it
extends and refines results described there as follows: We include 
a more comprehensive overview of prior work, as well as full proofs 
of the main results and detailed examples that explain the 
different procedures we propose. 
We expanded the considerations on the link between amalgamation 
and ground interpolation for theories which are not necessarily
universal and, when describing the symbol elimination procedure 
for theory extensions, we also explicitly consider situations in which 
instead of a single theory extension we need to consider chains of 
theory extensions. We analyze the relationship between the partial
amalgamation property proposed in \cite{Wies} and a weaker 
$W$-amalgamation property proposed in \cite{sofronie-ijcar-2016}.

\

\noindent {\em The paper is structured as follows.} In
Section \ref{sec:preliminaries} we present the main results on 
model theory needed in the paper.  In Section \ref{sec:interpolation}
we present existing results linking  
(sub-)amalgamation,  quantifier elimination
and the existence of ground interpolants, which
we then combine to obtain efficient ways of proving ground
interpolation and computing ground
interpolants. 
Section \ref{sec:local} contains the main definitions and 
results on local theory extensions; these are used in 
Section \ref{symb-elim} for 
symbol elimination and in 
Section \ref{interp-loc} for ground interpolation 
in theory extensions.


\section{Preliminaries}
\label{sec:preliminaries}

In this section we present the main results on 
model theory needed in the paper. 

\noindent 
We consider signatures of the form $\Pi = (\Sigma, {\sf Pred})$, 
where $\Sigma$ is a family of function symbols and ${\sf Pred}$
a family of predicate symbols.  
We assume known standard definitions from first-order logic  
such as $\Pi$-structures, models, homomorphisms, 
satisfiability, unsatisfiability.  We denote
``falsum'' with $\perp$.
%
%

\noindent Theories can be defined by specifying a set of axioms, or by specifying a
set of structures (the models of the theory). 
In this paper, (logical) theories are simply sets of sentences.

\begin{defi}[Entailment] 
 If $F, G$ are formulae and ${\mathcal T}$ is a theory we 
write: 
\begin{enumerate}
\item $F \models G$ to express the fact that every model of $F$ is a
  model of $G$; 
\item $F \models_{\mathcal T} G$ -- also written as ${\mathcal T} \cup F
  \models G$ -- 
to express the fact that every model of $F$ which is also a model of 
$\T$ is a model of $G$.
\end{enumerate}
\noindent If $F \models G$ we say that {\em $F$ entails $G$}. If $F \models_{\mathcal T}
 G$ we say that {\em $F$ entails $G$ w.r.t.\ ${\mathcal T}$}. 
\end{defi}
\noindent $F \models \perp$ means that $F$ is
unsatisfiable; $F \models_{\T} \perp$ means that there is no model of
$\T$ in which $F$ is true. 
If there is a model of $\T$ which is also a
model of $F$ we say
that $F$ is $\T$-consistent (or satisfiable w.r.t.\ ${\mathcal T}$). 

\medskip
\noindent 
If $\T$ is a theory over a signature $\Pi = (\Sigma, {\sf Pred})$ 
we denote by $\T_{\forall}$ (the universal theory of $\T$) the set of all
universal sentences which are entailed by $\T$. 

\medskip
\noindent 
If $\alg{A} = (A, \{ f_{\A} \}_{f \in \Sigma}, 
\{ P_\A \}_{P \in {\sf Pred}})$ is a $\Pi$-structure, in what follows we will sometimes
denote the universe $A$ of the structure $\alg{A}$ by $|\alg{A}|$. 

\begin{defi}[Embedding] 
For $\Pi$-structures ${\mathcal A}$ and ${\mathcal B}$, a map 
$\varphi : {\mathcal A} \rightarrow {\mathcal B}$ is an embedding if and only if 
it is an injective homomorphism and has the property that 
for every $P \in {\sf Pred}$ with arity $n$ and all 
$(a_1, \dots, a_n) \in |\alg{A}|^n$, $(a_1, \dots, a_n) \in P_\A$ iff
$(\varphi(a_1), \dots, \varphi(a_n)) \in P_{\mathcal B}$. 
\end{defi}
\medskip
\noindent  
In particular, an embedding preserves the truth of all literals.
\begin{defi}[Elementary Embedding]
An elementary embedding between two 
$\Pi$-structures is an 
embedding that preserves the truth of
all first-order formulae over $\Pi$. 
\end{defi} 
\begin{defi}[Elementarily Equivalent Structures]
Two $\Pi$-structures are elementarily equivalent if they satisfy 
the same first-order formulae over $\Pi$. 
\end{defi}
\begin{defi}[Diagram of a Structure]
Let $\alg{A} = (A, \{ f_{\A} \}_{f \in \Sigma}, \{ P_\A \}_{P \in {\sf
    Pred}})$ be a $\Pi$-structure. The
\emph{diagram} $\Delta(\A)$ of $\A$ is the set of all
literals true in the extension $\A^A$ 
of $\A$ where we have an additional constant for each element 
of $A$ (which we here denote with the same symbol) with the 
natural expanded interpretation
mapping the constant $a$ to the element $a$ of $|\A|$
(this is a set of  sentences over the
signature $\Pi^{\bar{a}}$ obtained by expanding 
$\Pi$ with a fresh constant $a$ 
for every element $a$ from $|\A|$).
\end{defi} 
\medskip
\noindent  
Note that if $\alg{A}$ is a $\Pi$-structure and $\T$ a theory and 
$\Delta(\alg{A})$  is $\T$-consistent then there exists a 
$\Pi$-structure $\alg{B}$ which is a model of $\T$
and into which $\alg{A}$ embeds. 

\smallskip 
\begin{defi}[Quantifier Elimination]
A theory $\T$ over a signature  
$\Pi$ {\em allows quantifier elimination} if for every formula $\phi$ over  
$\Pi$ there exists a quantifier-free formula $\phi^*$ over  
$\Pi$ which is equivalent to $\phi$ modulo $\T$. 
\end{defi} 
\medskip
\noindent  
Quantifier elimination can, in particular, be 
used for eliminating certain constants from ground 
formulae: 

\begin{thm}
Let $\T$ be a theory with 
signature $\Pi$ and $A(c_1, \dots, c_n, d_1, \dots, d_m)$ a ground 
formula over an extension $\Pi^C$ of $\Pi$ with additional constants $c_1, \dots, c_n$, $d_1, \dots, d_m$. 
If $\T$ has quantifier elimination then there exists a ground formula
$\Gamma(c_1, \dots, c_n)$ containing only constants $c_1, \dots, c_n$,
which is satisfiable w.r.t.\  $\T$ iff $A(c_1, \dots, c_n, d_1,\dots,d_m)$ 
is satisfiable w.r.t.\  $\T$.

\label{thm-qe-const}
\end{thm}

\begin{proof} Assume that $\T$ has quantifier elimination. 
Let $A(c_1, \dots, c_n, y_1, \dots, y_m)$ be the formula 
obtained from $A(c_1, \dots, c_n, d_1, \dots, d_m)$ by replacing 
every occurrence of $d_i$ with the variable $y_i$ for $i = 1, \dots,
m$. 
Let $\Gamma(c_1, \dots, c_n)$ be the formula obtained by eliminating the 
quantified variables $y_1, \dots y_m$ from the formula 
$\exists y_1, \dots, y_m A(c_1, \dots, c_n, y_1, \dots, y_m)$. 
The formula $\Gamma(c_1, \dots, c_n)$ is equivalent with $\exists y_1,
\dots, y_m A(c_1, \dots, c_n, y_1, \dots, y_m)$
w.r.t.\ $\T$, i.e.\ they are true in the same models of $\T$. 
The following are equivalent: 
\begin{itemize}
\item $A(c_1, \dots, c_n, d_1, \dots, d_m)$ is satisfiable w.r.t.\
  $\T$. 
\item $\exists y_1, \dots, y_m A(c_1, \dots, c_n, y_1, \dots, y_m)$ is satisfiable w.r.t.\
  $\T$. 
\item $\Gamma(c_1, \dots, c_n)$ is satisfiable w.r.t.\
  $\T$. \qedhere
\end{itemize}
\end{proof}

\begin{defi}[Model Complete Theory] 
A \emph{model complete} theory has the property that
all embeddings between its models are elementary.
\end{defi} 

\noindent Every theory which allows 
quantifier elimination (QE) is model complete (cf.\
\cite{hodges}, Theorem 7.3.1).
\begin{exa}
The following theories have QE 
and are therefore model complete.
\begin{enumerate-}
\item Presburger arithmetic with congruence mod. $n$ ($\equiv_n$),  
$n=2,3,...$ (\cite{enderton}, p.197).
\item Rational linear arithmetic in the signature $\{+,0,\leq\}$ (\cite{weispfenning-Q}).
\item Real closed ordered fields (\cite{hodges}, 7.4.4), e.g.,  the real numbers.
\item Algebraically closed fields (\cite{chang-keisler}, Ex.\ 3.5.2; Rem.\ p.204; \cite{hodges}, Ch.\ 7.4, Ex.\ 2). 
\item Finite fields (\cite{hodges},  Ch.\ 7.4, Example 2).
\item The theory of acyclic lists in the signature $\{ {\rm car}, {\rm cdr}, {\rm cons} \}$ (\cite{malcev,ghilardi:model-theoretic-methods}).
\end{enumerate-}
\label{examples-qe}
\end{exa}
\medskip
\noindent  
A model complete theory can sometimes be regarded as the completion of
another theory with the same universal fragment.  
Two theories $\T_1, \T_2$ are {\em companions} (or co-theories) if 
every model of $\T_1$ can be embedded (not necessarily elementarily) 
into a model of $\T_2$ and vice versa. This is the case iff $\T_1$ and
$\T_2$ 
have the same universal consequences (i.e.\ iff  ${\T_1}_{\forall} =
{\T_2}_{\forall}$). 

\begin{defi}[Model Companion]
A theory $\T^*$ is called a \emph{model companion} 
of $\T$ if 
\begin{enumerate}
\item[(i)] $\T$ and $\T^*$ are co-theories, 
\item[(ii)] $\T^*$ is model complete. 
\end{enumerate}
\end{defi}
\begin{defi}
\medskip
\noindent  
A theory $\T$ is called \emph{complete} if it has models and every two
models of $\T$ are elementarily equivalent (this is the same as saying
that for every formula $\phi$ in the language of $\T$ exactly one of
$\phi$, $\neg \phi$ is a consequence of $\T$). 
\end{defi}
\begin{defi}[Model Completion]
A theory $\T^*$ is called a \emph{model completion} of $\T$
if it is a model companion of $\T$ with the additional 
property 
\begin{enumerate-}
\item[(iii)] for every model $\alg{A}$ of $\T$, 
 $\T^* \union \Delta(\alg{A})$ is a complete theory  \\
(where $\Delta(\alg{A})$ is the diagram of $\A$).
\end{enumerate-}
\end{defi}

\noindent Thus, the model completion $\T^*$ of a theory $\T$ is model complete
(because it is a model companion of $\T$). Condition (iii) states that
every model of $\T$ is embeddable into a model of $\T^*$ ``in a unique
way''. 

A model complete theory is its own model completion. A theory that
admits quantifier elimination is the model completion of every one of
its companions. A theory $\T$ is the model completion of every one of its
companions iff it is the model completion of the weakest of them,
$\T_{\forall}$ (cf.\ e.g.\ \cite{poizat-book}). 
\begin{exa}
Below we present some examples of model completions:
\begin{enumerate-}
\item The theory of infinite sets is the model completion of the pure
theory of equality in the minimum signature containing only the
equality predicate (cf.\ e.g.\ \cite{ghilardi:model-theoretic-methods}). 
\item The theory of algebraically closed fields is the  model completion of the theory of fields.
This was the motivating example for developing the theory of model completions 
(\cite{chang-keisler}, Examples 3.5.2, 3.5.12; Remark 3.5.6 ff.; \cite{hodges}, 7.3).
\item The theory of dense total orders without endpoints is the model
  completion of the theory of total orders (cf.\ e.g.\ \cite{ghilardi:model-theoretic-methods}).
\item The theory of atomless Boolean algebras is the model completion
  of the theory of Boolean algebras (\cite{chang-keisler}, Example 3.5.12, cf. also p.196).
\item Universal Horn theories in finite signatures have a model completion if they are locally finite and have
the amalgamation property (e.g., graphs, posets) (\cite{wheeler:ec},
cf.\ also \cite{ghilardi:model-theoretic-methods}).
\end{enumerate-}
\label{example:model-completions}
\end{exa}

\noindent The following are easy consequences of the definitions: 

\begin{rem}
If $\T$ and $\T'$ are co-theories then $\T_{\forall} =
\T'_{\forall}$. 
 If $\T^*$ is a model companion (or model completion)
of $\T$ then $(\T^*)_{\forall} = \T_{\forall}$ and $\T^{**} = \T^*$. 
\label{cor:co-theories1}
\end{rem}

\begin{lem}
Let $\T_1$ and $\T_2$ be two co-theories with signature $\Pi$, and 
$A(c_1, \dots, c_n)$ be a ground clause over an extension $\Pi^C$ of $\Pi$
with additional constants $c_1, \dots, c_n$. 
Then $A$ is satisfiable w.r.t.\ $\T_1$  if and only
if it is satisfiable w.r.t.\ $\T_2$. 
\label{lem:co-theories-sat}
\end{lem}

\begin{proof}For $i = 1, 2$, 
$A(c_1, \dots, c_n)$ is unsatisfiable w.r.t.\ $\T_i$ if and only if  
the formula $\exists y_1, \dots, y_n A(y_1, \dots, y_n)$ is false in
all models of $\T_i$. This is the case if and only if  
$\T_i \models \forall y_1, \dots, y_n \neg A(y_1, \dots, y_n)$. 
As $\T_1$ and $\T_2$ are co-theories, they have the same universal
fragment.
Thus, $\T_1 \models \forall y_1, \dots, y_n \neg A(y_1, \dots, y_n)$
if and only if  $\T_2 \models \forall y_1, \dots, y_n \neg A(y_1,
\dots, y_n)$. 
It follows that $A(c_1, \dots, c_n)$ is satisfiable w.r.t.\ $\T_1$  if and only
if it is satisfiable w.r.t.\ $\T_2$. 
\end{proof}

\noindent {\bf Notation.} We denote with (indexed
versions of) $x, y, z$ variables and with (indexed versions of) $a, b,
c, d$ constants. As we will often refer to tuples of variables or
constants, we will succinctly denote them as follows: 
${\overline x}$ will stand for a sequence of variables $x_1, \dots,
x_n$, ${\overline x}^i$ for a sequence of variables $x^i_1, \dots,
x^i_n$, and ${\overline c}$ for a sequence of constants $c_1, \dots, c_n$. 

\section{Ground Interpolation}
\label{sec:interpolation}

A $\Pi$-theory ${\mathcal T}$ has interpolation if, for all
$\Pi$-formulae  $\phi$ and $\psi$,  
if $\phi \models_{\mathcal T} \psi$ then there exists a formula 
$I$ containing only symbols common\footnote{For full first-order logic, the symbols
  common to $\phi$ and $\psi$ are the function and predicate symbols 
which occur in both $\phi$ and $\psi$. 
Remark~\ref{remark} discusses which symbols are considered to
be common to $\phi$ and $\psi$ in articles in which interpolation
modulo a theory is considered.} to $\phi$ and 
$\psi$ such that  $\phi \models_{\mathcal T} I$ and $I  \models_{\mathcal T}
\psi$.
The formula $I$ is then called {\em the interpolant} of $\phi$ and
$\psi$. 

\noindent Craig proved that first order logic has interpolation \cite{Craig57} but even if 
$\phi$ and $\psi$ are e.g.\ conjunctions of ground literals  
the interpolant $I$ may still be an arbitrary formula. 
It is often important to identify situations in which  
ground clauses have ground interpolants.
\begin{defi}[Ground Interpolation]
A theory ${\mathcal T}$ has the {\em ground interpolation 
property} (for short: ${\mathcal T}$ has {\em ground interpolation}) 
if for every pair of ground 
formulae 
$A({\overline c}, {\overline a})$ (containing constants ${\overline
  c}, {\overline a}$) 
and  $B({\overline c}, {\overline b})$ (containing
  constants ${\overline c}, {\overline b}$), if 
$A({\overline c}, {\overline a}) \wedge B({\overline c}, {\overline b}) \models_{\mathcal T} \perp$ then 
there exists a ground formula $I({\overline c})$, containing only 
the constants ${\overline c}$ occurring both in $A$ and $B$, such that 
$A({\overline c}, {\overline a}) \models_{\mathcal T} I({\overline c}) 
\text{ and } 
B({\overline c}, {\overline b}) \wedge 
I({\overline c}) \models_{\mathcal T} \perp.$
\label{ground-interpolation} 
\end{defi}

\noindent Let $\T$ be a theory in a signature $\Sigma$ and $\Sigma'$ a signature
disjoint from $\Sigma$. We denote by $\T \cup {\sf UIF}_{\Sigma'}$ the
extension of $\T$ with uninterpreted symbols in $\Sigma'$. 
\begin{defi}[General Ground Interpolation \cite{Ghilardi-2014}]
We say that a theory ${\mathcal T}$ in a signature $\Sigma$ 
has the {\em general ground interpolation 
property} (or, shorter, that ${\mathcal T}$ has {\em general ground interpolation}) 
if for every signature $\Sigma'$ disjoint from $\Sigma$ and every pair
of ground $\Sigma \cup \Sigma'$-formulae $A$ and  
$B$,  if 
$A \wedge B \models_{{\T} \cup {\sf UIF}_{\Sigma'}} \perp$ then 
there exists a ground formula $I$ such that: 
\begin{enumerate}
\item[(i)] all constants, predicate and function symbols from $\Sigma'$ 
occurring in $I$ occur both in $A$ and $B$,  and 
\item[(ii)] $A \models_{{\mathcal T} \cup {\sf UIF}_{{\Sigma}'}} I \text{ and } 
B  \wedge I \models_{{\mathcal T} \cup {\sf UIF}_{{\Sigma}'}}  \perp.$
\end{enumerate}
\label{general-ground-interpolation}
\end{defi}

\begin{rem}
{\em 
When defining ground interpolation, in many papers a difference is
made between interpreted and uninterpreted function symbols or
constants: The interpolant $I$ of two (ground) formulae $A$ and $B$ is
often required to contain only constants and uninterpreted function 
symbols occurring in both $A$ and $B$; no restriction is imposed on 
the interpreted function symbols. 
We explain how these aspects are addressed in the previously given
definitions: 
\
\begin{itemize}
\item Definition~\ref{ground-interpolation} 
assumes that all function and predicate symbols in the signature of
$\T$ which are not constant are 
{\em interpreted}, i.e.\ can be contained in the interpolant of two 
formulae $A$ and $B$ also if  they are not common to the two
formulae. 

\item In Definition~\ref{general-ground-interpolation} 
(with the notation used there) the function and predicate symbols 
from the signature
$\Sigma$ of the theory ${\mathcal T}$ are considered to be {\em
  interpreted}
(thus can be contained in the interpolant of two 
formulae $A$ and $B$ also if  they are not common to the two
formulae), whereas the constants and the function and predicate
symbols from $\Sigma'$ are considered to be {\em uninterpreted} 
(thus all constants and all predicate and function symbols occurring
in the interpolant of two 
formulae $A$ and $B$ must occur in both $A$ and $B$). 
\end{itemize}
}
\label{remark} 
\end{rem}


%

\subsection{Amalgamation and Ground Interpolation}

There exist results which relate ground interpolation to 
amalgamation 
\cite{Jonsson65,Bacsich75,Wronski86,Ghilardi-2014,Ghilardi-2017} 
and thus allow us to recognize 
many theories with ground interpolation. 

\noindent For instance, Bacsich \cite{Bacsich75} shows that every 
{\em universal theory} with the amalgamation property has 
ground interpolation. The terminology is defined below.
\begin{defi}[Amalgamation Property]
A theory $\T$ has the {\em amalgamation property} iff 
whenever 
we are given models $M_1$ and $M_2$ of $\T$ 
with a common substructure 
$A$ which is a model of $\T$, 
there exists a further model $M$ of $\T$ endowed with embeddings 
$\mu_i : M_i \rightarrow M$, 
$i = 1,2$ whose restrictions to $A$ coincide.

A theory $\T$ has the {\em strong amalgamation property} if the 
preceding embeddings 
$\mu_1, \mu_2$ and the preceding model $M$ can be chosen so as to 
satisfy the following additional condition: if for some $m_1, m_2$ 
we have $\mu_1(m_1) = \mu_2(m_2)$, then there exists an element 
$a \in A$ such that $m_1 = m_2 = a$. 
\end{defi}
\begin{thm}[\cite{Bacsich75}]
Every universal theory with the amalgamation property has the ground 
interpolation property. 
\label{bacsich75}
\end{thm} 

\noindent 
Theorem~\ref{bacsich75} can be used to show that 
equational classes such as (abelian) groups, 
partially-ordered sets, lattices, semilattices, distributive lattices and 
Boolean algebras have ground interpolation. 

\medskip
\noindent 
If $\T$ is not a universal theory, the amalgamation property does
not necessarily imply the ground interpolation property. We present two
possible solutions in this situation: 
\begin{description} 
\item[Solution 1] Apply Theorem~\ref{bacsich75} to the 
universal fragment $\T_{\forall}$ to check whether $\T_{\forall}$ has
ground interpolation and note that a theory $\T$ 
has ground interpolation iff its universal fragment $\T_{\forall}$
does. 
\item[Solution 2] 
Extend the amalgamation property. 
\end{description}
We discuss and compare these two approaches in what follows. 

\

\noindent {\bf Solution 1: Regard $\T_{\forall}$ instead of $\T$}. 
 We relate existence of ground interpolants in $\T$ and 
$\T_{\forall}$, and use Theorem~\ref{bacsich75} to check whether 
$\T_{\forall}$ has ground interpolation. 

\begin{lem} 
Let $\T$ be a logical theory. $\T$ has ground interpolation iff
$\T_{\forall}$ has ground interpolation. 
\label{t-forall-t-ground-int}
\end{lem}

\begin{proof}
($\Rightarrow$) Assume first that $\T$ has ground interpolation. We
show that $\T_{\forall}$ has ground interpolation.
Let $A({\overline a}, {\overline c})$ and $B({\overline b}, {\overline c})$ 
be ground formulae in the signature of $\T$ possibly 
containing new constants ${\overline a}, {\overline b}, {\overline c}$
such that $ A({\overline a}, {\overline c}) \wedge B({\overline b},
{\overline c}) \models_{\T_{\forall}} \perp$. As all formulae in 
$\T_{\forall}$ are consequences of $\T$, 
$ A({\overline a}, {\overline c}) \wedge 
B({\overline b}, {\overline c}) \models_{\T} \perp$.
As $\T$ has ground interpolation, there exists a ground formula 
$I({\overline c})$ 
in the signature of $\T$ containing only additional constants 
occurring in both $A$ and $B$, such that $A({\overline a}, {\overline c}) \models_{\T} I({\overline c})$ and 
$I({\overline c}) \wedge B({\overline b}, {\overline c}) \models_{\T} \perp$. We argue that in this case 
$A({\overline a}, {\overline c}) \models_{\T_{\forall}} I({\overline c})$ and 
$I({\overline c}) \wedge B({\overline b}, {\overline c})
\models_{\T_{\forall}} \perp$, i.e.\ $I({\overline c})$ is a ground interpolant
of $A({\overline a}, {\overline c})$ and $B({\overline b}, {\overline c})$ also w.r.t.\ $\T_{\forall}$. 
 Indeed, the following are equivalent: 
\begin{enumerate}
\item $A({\overline a}, {\overline c}) \models_{\T} I({\overline c})$ 
\item $A({\overline a}, {\overline c}) \wedge \neg I({\overline c}) \models_{\T} \perp$
\item $\exists {\overline a} \exists {\overline c} \, (A({\overline a},
{\overline c}) \wedge \neg I({\overline c})) \models_{\T} \perp$ 
\item 
$\T \models \forall  {\overline a} \forall {\overline c} \, \neg(A({\overline a},
{\overline c}) \wedge \neg I({\overline c}))$. 
\end{enumerate}
(To simplify notation, we regarded the
additional constants in $A$ and $I$ as existentially quantified variables; after negation they became universally quantified.) 

\noindent Thus, $A({\overline a}, {\overline c}) \models_{\T} I({\overline c})$ iff 
$\forall  {\overline a} \forall {\overline c} \, \neg(A({\overline a},
{\overline c}) \wedge \neg I({\overline c}) )\in \T_{\forall}$ iff 
$\T_{\forall} \models \forall  {\overline a} \forall {\overline c} \, \neg(A({\overline a},
{\overline c}) \wedge \neg I({\overline c}))$. We can now use the chain 
of equivalences established before to conclude that 
$A({\overline a}, {\overline c}) \models_{\T} I({\overline c})$ iff $A({\overline a}, {\overline c}) \models_{\T_{\forall}} I({\overline c})$. 
Similarly we can show that $I({\overline c}) \wedge B({\overline b}, {\overline c}) \models_{\T} \perp$ iff 
$I({\overline c}) \wedge B({\overline b}, {\overline c}) \models_{\T_{\forall}} \perp$. 

\medskip
\noindent $(\Leftarrow)$ Assume now that $\T_{\forall}$ has ground
 interpolation. Let $A({\overline a}, {\overline c}), B({\overline b}, {\overline c})$ 
be ground formulae in the signature of $\T$ possibly 
containing new constants ${\overline a}, {\overline b}, {\overline c}$
such that 
$ A({\overline a}, {\overline c}) \wedge B({\overline b},
{\overline c}) \models_{\T} \perp$.
Then 
$\exists {\overline a}, {\overline b}, {\overline c} \, 
A({\overline a}, {\overline c}) \wedge B({\overline b},
{\overline c}) \models_{\T} \perp$. (For the sake of simplicity we 
again regarded the
additional constants in $A$ and $I$ as constants, which we quantified 
existentially since we talk about satisfiability; after negation they became universally quantified.) 
Hence, $\models_{\T} \forall {\overline a}, {\overline b}, {\overline c} \, 
\neg(A({\overline a}, {\overline c}) \wedge B({\overline b},
{\overline c}))$, i.e.\ $\forall {\overline a}, {\overline b},
{\overline c} \, \neg(A({\overline a}, {\overline c}) \wedge B({\overline b},
{\overline c})) \in \T_{\forall}$. 
Then $\exists {\overline a}, {\overline b}, {\overline c} \, 
A({\overline a}, {\overline c}) \wedge B({\overline b},
{\overline c}) \models_{\T_{\forall}} \perp$, so 
 $ A({\overline a}, {\overline c}) \wedge B({\overline b},
{\overline c}) \models_{\T_{\forall}} \perp$.
%
%
From the fact that $\T_{\forall}$ has 
the ground interpolation property it follows that there exists a 
ground interpolant $I({\overline c})$ such that 
\[ A({\overline a}, {\overline c}) \models_{\T_{\forall}} I({\overline
  c}) \text{ and } B({\overline b}, {\overline c})  
\wedge I({\overline c})  \models_{\T_{\forall}} \perp.\]  
Since $\T \models \T_{\forall}$ we then know that $A({\overline a}, {\overline c}) \models_{\T} I({\overline c})$ and $B({\overline b}, {\overline c}) 
\wedge I({\overline c})  \models_{\T} \perp$, so $I$ is an interpolant of $A \wedge 
B$ also w.r.t. $\T$.
\end{proof}

\begin{cor}
Let $\T$ be a logical theory. Assume that $\T_{\forall}$ has the 
amalgamation property. Then both  $\T_{\forall}$ and $\T$ have ground 
interpolation. 
\label{cor-tforall}
\end{cor} 

\begin{proof}
Since $\T_{\forall}$ is a universal theory,
by Theorem~\ref{bacsich75} if $\T_{\forall}$ has the 
amalgamation property then it has ground interpolation. 
By Lemma~\ref{t-forall-t-ground-int} it follows that $\T$ has ground interpolation: 
\end{proof}


\medskip
\noindent {\bf Solution 2: Extend the amalgamation property.} 
In \cite{Ghilardi-2014}
Theorem~\ref{bacsich75} is extended to theories which are not necessarily 
universal. If a theory ${\mathcal T}$ is not necessarily universal its
class of models is not closed under substructures.  
In order to extend Theorem~\ref{bacsich75}  to this case it was 
necessary to define a variant of the amalgamation property 
(called the sub-amalgamation property), in which it is not required 
that the common substructure $A$ of $M_1$ and $M_2$ is a model 
of the theory ${\mathcal T}$. 
\begin{defi}[Sub-Amalgamation Property \cite{Ghilardi-2014}] 
A theory $\T$ has the {\em sub-amalgamation property} iff 
whenever 
we are given models $M_1$ and $M_2$ of $\T$ with a common substructure 
$A$, 
there exists a further model $M$ of $\T$ endowed with embeddings 
$\mu_i : M_i \rightarrow M$, 
$i = 1,2$ whose restrictions to $A$ coincide. 

A theory $\T$ has the {\em strong sub-amalgamation property} if the 
preceding embeddings 
$\mu_1, \mu_2$ and the preceding model $M$ can be chosen so as to 
satisfy the following additional condition: if for some $m_1, m_2$ 
we have $\mu_1(m_1) = \mu_2(m_2)$, then there exists an element 
$a \in A$ such that $m_1 = m_2 = a$. 
\label{def:strong-subamalgamation}
\end{defi}
Clearly, for universal theories the amalgamation property and the 
sub-amalgamation property coincide. 
\begin{defi}[Equality Interpolating Theories \cite{Ghilardi-2014}]
A theory $\T$ is equality interpolating iff it has the
ground interpolation property and has the property that 
for all tuples  ${\overline x} = x_1, \dots, x_n$, 
${\overline y}^1 = y^1_1, \dots, y^1_{n_1}$, ${\overline z}^1 = z^1_1, \dots, z^1_{m_1}$, 
${\overline y}^2 = y^2_1, \dots, y^2_{n_2}$, ${\overline z}^2 = z^2_1, \dots, z^2_{m_2}$ of constants, 
and for every pair of
ground formulae $A({\overline x}, {\overline z}^1, {\overline y}^1)$
and $B({\overline x}, {\overline z}^2, {\overline y}^2)$  
such that
$ \displaystyle{A({\overline x}, {\overline z}^1, {\overline y}^1)
  \wedge B({\overline x}, {\overline z}^2, {\overline y}^2)
  \models_{\mathcal T} \bigvee_{i = 1}^{n_1} \bigvee_{j = 1}^{n_2} y^1_i
\approx y^2_j}$ 
there exists a tuple of terms containing only
the constants in ${\overline x}$, $v({\overline x}) = v_1, \dots, v_k$ 
such that

$\displaystyle{A({\overline x}, {\overline z}^1, {\overline y}^1)  
\wedge B({\overline x}, {\overline z}^2, {\overline y}^2) \models_{\mathcal T} \bigvee_{i = 1}^{n_1}
\bigvee_{u = 1}^{k} y^1_i \approx v_u \vee \bigvee_{j = 1}^{n_2} \bigvee_{u = 1}^{k} 
v_u \approx y^2_j.}$ 
\end{defi}
\begin{thm}[\cite{Ghilardi-2014}]
The following hold:
\begin{enumerate-} 
\item A theory $\T$ has the sub-amalgamation property iff
  it has ground interpolation. 
\item $\T$ is strongly sub-amalgamable iff it has 
  general ground interpolation. 
\item If $\T$ has ground interpolation,  then $\T$ is
  strongly sub-amalgamable
iff it is equality interpolating. 
\item If $\T$ is universal and has quantifier elimination, $\T$  is equality interpolating. 
\end{enumerate-}
\label{thm:ghilardi-amalg-interp}
\end{thm}
\begin{thm}[\cite{chang-keisler}] 
If ${\T^*}$ is a model companion of ${\T}$ the following are 
equivalent: 
\begin{enumerate-}
\item $\T^*$ is a model 
completion of $\T$. 
\item $\T$ has the amalgamation property. 
\end{enumerate-}
If, additionally, $\T$ has  universal axiomatization, either of the conditions 
(1) or (2) above is equivalent to (3) $\T^*$ allows quantifier elimination. 
\label{thm:criteria-model-compl} 
\end{thm}
\begin{thm}[\cite{hodges}, p.390]
If ${\T^*}$ is a model companion of ${\T}$ the following are 
equivalent: 
\begin{enumerate-}
\item $\T^*$ allows quantifier elimination. 
\item $\T_{\forall}$ has the amalgamation property. 
\end{enumerate-}
\label{mc-qe-iff-univ-fragm-amalg}
\end{thm}

\noindent We now show that for every theory $\T$ which has a model
companion $\T^*$ Solutions 1 and 2 are equivalent. 
\begin{thm}
Let $\T$ be a theory and let $\T^*$ be a model companion of $\T$.  
Then $\T_{\forall}$ has the amalgamation
property iff $\T$ has the sub-amalgamation property. 
%
\end{thm}

\begin{proof}($\Rightarrow$) Assume that 
$\T_{\forall}$ has the amalgamation property. Then, by
Corollary~\ref{cor-tforall}, 
$\T$ has ground interpolation, hence, 
by Theorem~\ref{thm:ghilardi-amalg-interp}
(1), $\T$ has the sub-amalgamation property. 

($\Leftarrow$) 
Assume now that 
$\T$ has the sub-amalgamation property.  By
Theorem~\ref{thm:ghilardi-amalg-interp} (1), $\T$ has
ground interpolation, so by Theorem~\ref{t-forall-t-ground-int} 
$\T_{\forall}$ has ground interpolation. Then, again by 
Theorem~\ref{thm:ghilardi-amalg-interp} (1), 
$\T_{\forall}$ has the sub-amalgamation property.  
Since $\T_{\forall}$ is a universal theory, it immediately follows
that $\T_{\forall}$ has the
amalgamation property.
\end{proof}

\subsection{Quantifier Elimination and Ground Interpolation}
Clearly, if a theory $\T$ allows quantifier elimination then it has
ground interpolation: Assume $A \wedge B \models_{\mathcal T} \perp$. 
We can simply use quantifier elimination to eliminate the non-shared 
constants  from $A$ w.r.t.\ $\T$ and obtain an interpolant.  
The converse is not true (the theory of uninterpreted function symbols 
over a signature $\Sigma$ has ground interpolation but does not allow
quantifier elimination). 
\begin{thm}
If $\T$ is a universal theory which allows quantifier elimination then
$\T$ has {\em general
ground interpolation}. 
\label{thm:QE-GGI}
\end{thm}

\begin{proof}Clearly, if $\T$ allows quantifier elimination then it has
ground interpolation. By Theorem~\ref{thm:ghilardi-amalg-interp}(4) 
we know that if a theory
$\T$ is universal and allows quantifier elimination then it is equality
interpolating. 
By Theorem~\ref{thm:ghilardi-amalg-interp}(3), if a theory $\T$ has ground interpolation and 
is equality interpolating then it  has the strong
  sub-amalgamation property, hence, by Theorem~\ref{thm:ghilardi-amalg-interp}(2),  it has 
  general ground interpolation.
\end{proof}

\begin{exa}
\begin{enumerate-}
\item All theories in Example~\ref{examples-qe} allow quantifier
elimination, hence have ground interpolation.
\item The theory of pure equality has the strong (sub-)amalgamation
  property \cite{Ghilardi-2014},
hence by Theorem~\ref{thm:ghilardi-amalg-interp} it allows general
ground interpolation. 
\item The theory of absolutely-free data structures \cite{malcev}
is universal and has quantifier elimination, hence by
Theorem~\ref{thm:QE-GGI} 
it has general ground interpolation. 
\end{enumerate-}
\end{exa}

\subsection{Model Companions and Ground Interpolation}
 In what follows we establish links between ground interpolation 
resp.\ quantifier elimination in a theory and in its model companions 
(if they exist).

\begin{thm}
If $\T$ is a universal theory which has ground interpolation, 
and $\T^*$ is a model companion of $\T$
then $\T^*$ allows quantifier elimination (and it is a model
completion of $\T$).
\label{thn:univ-GI-QE}
\end{thm}

\begin{proof}Assume that $\T$ is a universal theory which has ground 
interpolation. Then, by \cite{Bacsich75}, $\T$ has the amalgamation
property. 
By Theorem~\ref{thm:criteria-model-compl}, 
$\T^*$ is a model completion of $\T$ and it allows quantifier
elimination.
\end{proof}

\smallskip
\noindent 
We now analyze situations when $\T$ is not necessarily a universal
theory. 
\begin{thm}
Let $\T$ be a theory. Assume that $\T$ has a model companion $\T^*$. 
If $\T^*$ has ground interpolation then so does $\T$; the ground
interpolants computed w.r.t.\ $\T^*$ are also interpolants w.r.t.\
$\T$. 
%
\label{thm:mc-th}
\end{thm}

\begin{proof} If $\T^*$ is the model companion of $\T$ they are
co-theories, so $\T_{\forall} = \T^*_{\forall}$, cf.\ Remark~\ref{cor:co-theories1}.
 Assume that $\T^*$ has ground interpolation.
Let $A$, $B$ be two sets of ground clauses such that 
$\T \cup A \cup B \models \perp$. As $\T_{\forall} = \T^*_{\forall}$,
by Lemma~\ref{lem:co-theories-sat}, 
$\T^* \cup A \cup B \models \perp$.
As 
$\T^*$ has ground interpolation, there exists a ground formula $I$ 
containing only constants occurring in both $A$ and $B$ such that
$\T^* \cup A \cup \neg I$ and $\T^* \cup B \cup I$ are unsatisfiable. 
Then, again by Lemma~\ref{lem:co-theories-sat},  $\T \cup A \cup \neg I$ and $\T \cup B \cup I$ are unsatisfiable,
i.e.\ $I$ is an interpolant w.r.t.\ $\T$.
\end{proof}
\begin{cor}
Let $\T$ be a universal theory. Assume that $\T$ has a model companion
$\T^*$. 
Then $\T$ has ground interpolation iff $\T^*$ has ground
interpolation.
\label{cor:gi}
\end{cor}

 \begin{proof}If $\T$ is a universal theory and has ground
interpolation then, by Theorem~\ref{thn:univ-GI-QE}, $\T^*$ allows quantifier
elimination hence has ground
interpolation. The converse follows from Theorem~\ref{thm:mc-th}.
\end{proof}

\begin{cor}
Let $\T$ be a theory. Assume that $\T$ has a model companion $\T^*$. 
If $\T^*$ allows quantifier elimination then $\T$ has ground
interpolation.
\end{cor}

\begin{exa} The following theories have ground 
interpolation: 
\begin{enumerate-}
\item The pure theory of equality (its model completion is the theory 
of an infinite set, which allows quantifier elimination, cf.\ e.g.\ \cite{ghilardi:model-theoretic-methods}).
\item The theory of total orderings (its model completion is the theory
  of dense total orders without endpoints, which allows quantifier
  elimination, cf.\ e.g.\ \cite{ghilardi:model-theoretic-methods}). 
\item The theory of Boolean algebras (its model completion is the theory
  of atomless Boolean algebras, which allows quantifier elimination,
  cf.\ \cite{chang-keisler}). 
\item The theory of fields (its model completion is the theory of 
  algebraically closed fields, which allows quantifier elimination,
  cf.\ \cite{chang-keisler}). 
\end{enumerate-}
\end{exa}

\

\noindent For a theory $\T$ which is not universal the amalgamation property does
not necessarily imply the ground interpolation property. Instead, we
can check whether the universal fragment $\T_{\forall}$ of $\T$ has
the amalgamation property as explained in Corollary~\ref{cor-tforall}. 

\begin{thm}
Let $\T$ be a logical theory such that $\T_{\forall}$ has the 
amalgamation property. 
If $\T$ has a model companion $\T^*$ then $\T^*$ allows 
quantifier elimination (so it is a model completion of $\T$) 
hence interpolants in $\T$ can be computed by quantifier 
elimination in $\T^*$. 
\end{thm}

\begin{proof}
 If ${\T^*}$ is a model companion of ${\T}$  
and $\T_{\forall}$ has the amalgamation property 
then by Theorem~\ref{mc-qe-iff-univ-fragm-amalg} 
$\T^*$ allows quantifier elimination, so it has ground interpolation 
and by Theorem~\ref{thm:mc-th}  so does $\T$; the ground 
interpolants computed w.r.t.\ $\T^*$ are also interpolants w.r.t.\
$\T$.
\end{proof}

\noindent A summary of the results obtained in this section is given,
in succinct form, in Section~\ref{summary-1}. 

Until now, we discussed possibilities for symbol elimination and ground
interpolation in arbitrary theories. 
However, often the theories we consider are extensions of a ``base''
theory with additional function symbols satisfying certain properties
axiomatized using clauses; we now analyze such theories.
 In Section~\ref{sec:local} we recall 
the main definitions and results related to (local) theory extensions. 
We use these results in Section~\ref{symb-elim} to study possibilities of 
symbol elimination in theory extensions and in
Section~\ref{interp-loc} to identify theory extensions with ground
interpolation.

\section{Local Theory Extensions} 
\label{sec:local}

\noindent 
Let $\Pi_0 {=} (\Sigma_0, {\sf Pred})$ be a signature, and ${\mathcal T}_0$ be a 
``base'' theory with signature $\Pi_0$. 
We consider 
extensions $\T := {\mathcal T}_0 \cup \K$
of ${\mathcal T}_0$ with new function symbols $\Sigma$
({\em extension functions}) whose properties are axiomatized using 
a set $\K$ of (universally closed) clauses 
in the extended signature $\Pi = (\Sigma_0 \cup \Sigma, {\sf Pred})$, 
which contain function symbols in $\Sigma$. 
Let $C$ be a fixed countable set of fresh constants. We denote by 
$\Pi^C$ the extension of $\Pi$ with constants in $C$. 
 
\subsection{Locality Conditions}
If $G$ is a finite set of ground $\Pi^C$-clauses and $\K$ a set of
$\Pi$-clauses, we denote by  ${\sf st}({\mathcal K}, G)$ the set of all 
ground terms which occur in $G$ or ${\mathcal K}$. We denote by 
${\sf est}({\mathcal K}, G)$ the set of all extension ground terms (i.e.\ 
terms starting with a function in $\Sigma$) which occur in $G$ or ${\mathcal K}$.

We regard every finite set $G$
of ground clauses as the ground formula $\bigwedge_{C \in G} C$. 
If $T$ is a set of ground terms in the signature  $\Pi^C$, 
we denote by $\K[T]$ the set of all instances of $\K$ in which the terms 
starting with a function symbol in $\Sigma$ are in $T$. 
Formally: 
\begin{align*} 
\K[T] := \{ \varphi\sigma \,|\; & \forall \bar{x}.\,
\varphi(\bar{x}) \in \K, 
\text{ where (i) if } f \in \Sigma \text{ and } t = f(t_1,...,t_n) \text{ occurs in } \varphi\sigma\\[-1ex]                                     & \text{ then } t \in T \text{; (ii) if } x
\text{ is a variable that does not appear below some } \\[-1ex]
                                     & \text{ $\Sigma$-function in } \varphi \text{ then } \sigma(x) = x \}.
\end{align*}
%
\begin{defi}[\cite{Ihlemann-Sofronie-2010}]
\label{def-psi-loc}
Let $\Psi$ be a map associating with 
every set $T$ of ground $\Pi^C$-terms a set $\Psi(T)$ of ground $\Pi^C$-terms. 
For any set $G$ of ground $\Pi^C$-clauses we write 
$\K[\Psi_{\mathcal K}(G)]$ for $\K[\Psi({\sf est}({\mathcal K}, G))]$.
%
Let $\T_0 \cup \K$ be an extension of $\T_0$ 
with clauses in ${\mathcal K}$. 
We define: 

\medskip
\begin{tabbing}
\= ${\sf (Loc}_f^\Psi)$~ \quad \quad  \= For every finite set $G$ of ground clauses in $\Pi^C$ it holds that\\
\>\> $\T_0 \cup {\mathcal K} \cup G \models \bot$ if and only if $\T_0
\cup \K[\Psi_{\mathcal K}(G)] \cup G$ is unsatisfiable. 
\end{tabbing}

\medskip
\noindent Extensions satisfying condition ${\sf (Loc}_f^\Psi)$ are called
{\em $\Psi$-local}. 
\end{defi}
%
\noindent If $\Psi$ is the identity we obtain the notion of {\em local theory  
extensions} \cite{sofronie-cade-05,sofronie-frocos-07}, which generalizes the notion of 
{\em local theories}  
\cite{GivanMcAllester92,McAllester93,McAllester-acm-tocl-02,Ganzinger-01-lics}. 

\subsection{Partial Structures} 

\noindent In \cite{sofronie-cade-05} we showed that 
local theory extensions can be 
recognized by showing that certain partial models embed into total
ones, and in \cite{Ihlemann-Sofronie-2010} we established similar
results for $\Psi$-local theory extensions and generalizations
thereof. 
We introduce the main  
definitions here.

Let $\Pi = (\Sigma, {\sf Pred})$ be a
first-order signature with set of function symbols $\Sigma$ 
and set of predicate symbols ${\sf Pred}$. 
A \emph{partial $\Pi$-structure} is a structure 
$\A = (A, \{f_\A\}_{f\in\Sigma}, \{P_\A\}_{P\in \Pred})$, 
where $A$ is a non-empty set, for every $n$-ary $f \in \Sigma$, 
$f_\A$ is a partial function from $A^n$ to $A$, and for every $n$-ary 
$P \in {\sf Pred}$, $P_\A \subseteq A^n$. We consider constants (0-ary functions) to be always 
defined. $\A$ is called a \emph{total structure} if the 
functions $f_\A$ are all total. 
Given a (total or partial) $\Pi$-structure $\A$ and $\Pi_0 \subseteq \Pi$ 
we denote the reduct of 
$\A$ to $\Pi_0$ by $\A{|_{\Pi_0}}$.

The notion of evaluating a term $t$ with variables $X$ w.r.t.\ 
an assignment $\beta : X \rightarrow A$
 for its variables in a partial structure
$\A$ is the same as for total algebras, except that the evaluation is
undefined if $t = f(t_1,\ldots,t_n)$ 
and at least one of
$\beta(t_i)$ is undefined, or else $(\beta(t_1),\ldots,\beta(t_n))$ is
not in the domain of $f_\A$.

\begin{defi}
A \emph{weak $\Pi$-embedding} between two partial $\Pi$-structures
$\A$ and $\B$, where $\A = (A, \{f_\A\}_{f\in \Sigma}, \{P_\A\}_{P \in \Pred})$
 and  $\B = (B, \{f_\B\}_{f\in \Sigma}, \{P_\B\}_{P \in \Pred})$
is a total map $\varphi : A \rightarrow B$ such that 
\begin{enumerate-}
\item[(i)] $\varphi$ is
an embedding w.r.t.\ ${\sf Pred} \cup \{ = \}$, i.e.\ 
 for every $P \in {\sf Pred}$ with arity $n$ and every 
$a_1, \dots, a_n \in \A$, 
$(a_1, \dots, a_n) \in P_\A$ if and only if $(\varphi(a_1), \dots, \varphi(a_n))\in P_\B$. 
\item[(ii)]  
whenever $f_\A(a_1, \dots, a_n)$ is defined (in $\A$),  then 
$f_\B(\varphi(a_1), \dots, \varphi(a_n))$ is defined (in $\B$) and 
$\varphi(f_\A(a_1, \dots, a_n)) = f_\B(\varphi(a_1), \dots, \varphi(a_n))$,
for all $f \in \Sigma$. 
\end{enumerate-}
\end{defi} 
%

\begin{defi}[Weak validity]
Let $\A$ be a partial $\Pi$-algebra
and $\beta : X {\rightarrow} A$ a valuation for its variables.
$(\A, \beta)$ {\em weakly satisfies a clause  $C$} (notation: 
$(\A, \beta) \models_w C$) if either some of the literals in 
$\beta(C)$ are not defined or otherwise all literals are defined and
for at least one literal $L$ in $C$, $L$ is true in $\A$ w.r.t.\ $\beta$. 
$\A$ is a {\em weak partial model} 
of a set of clauses ${\mathcal K}$ if $(\A, \beta)  \models_w C$ for every 
valuation 
$\beta$ and every clause $C$ in ${\mathcal K}$. 
\end{defi}

\subsection{Recognizing \texorpdfstring{$\Psi$}{Psi}-Local Theory Extensions} 
In \cite{sofronie-cade-05} we proved that if 
every weak partial model of an extension 
${\mathcal T}_0 \cup {\mathcal K}$ of a base theory ${\mathcal T}_0$ 
with total base functions can be embedded into a total
model of the extension, then the extension is local. 
In \cite{sofronie-ihlemann-jacobs-tacas08} we lifted these results to $\Psi$-locality. 

Let $\alg{A} = (A, \{ f_{\A} \}_{f \in \Sigma_0 \cup \Sigma} \cup C, \{ P_\A
\}_{P \in {\sf Pred}})$ be a partial $\Pi^C$-structure with total
$\Sigma_0$-functions. 
Let $\Pi^A$ be the extension of the signature $\Pi$ with constants 
from $A$. We denote by $T(\A)$ the following set of ground 
$\Pi^A$-terms: 
$$T(\A) := \{ f(a_1,...,a_n) \,|\; f \in \Sigma, a_i \in A, i=1,\dots,n, f_{\A}(a_1,...,a_n) \text{ is defined }  \}. $$
Let ${\sf PMod}_{w,f}^\Psi({\Sigma}, {\mathcal T})$ be the class of all
weak partial models $\A$ of ${\mathcal T}_0 \cup {\mathcal K}$, such that
$\A{|_{\Pi_0}}$ is a total model of $\T_0$, the
$\Sigma$-functions are possibly partial, $T(\A)$ is finite and 
all terms in $\Psi({\sf est}(\K, T(\A)))$ are
defined (in the extension $\A^A$ with
constants from $A$).
We consider the following embeddability property of partial algebras:

\medskip
\begin{tabbing}
\= $({\sf Emb}_{w,f}^\Psi)$ \quad \= Every $\alg{A} \in {\sf PMod}_{w,f}^\Psi({\Sigma},  \T)$ weakly embeds into a total model of $\T$.
\index{$\Psi$-!embeddability}
%
%
\end{tabbing}

\medskip
\noindent 
We also consider the properties $({\sf EEmb}_{w,f}^{\Psi})$,  
which additionally requires the embedding to be {\em elementary} and 
$({\sf Comp}_f)$  which requires that every  structure $\alg{A} \in {\sf
  PMod}_{w,f}^\Psi({\Sigma}, \T)$ embeds
into a total model of $\T$ {\em with the same support}. 

\

When establishing links between locality and embeddability we require 
that the clauses in $\K$
are \emph{flat} 
and \emph{linear} w.r.t.\ $\Sigma$-functions.
When defining these notions we distinguish between ground and 
non-ground clauses.

\begin{defi}
An {\em extension clause $D$ is flat} (resp. \emph{quasi-flat}) 
when all symbols 
below a $\Sigma$-function symbol in $D$ are variables. 
(resp. variables or ground $\Pi_0$-terms).
$D$ is \emph{linear}  if whenever a variable occurs in two terms of $D$
starting with $\Sigma$-functions, the terms are equal, and 
no term contains two occurrences of a variable.

A {\em ground clause $D$ is flat} if all symbols below a $\Sigma$-function 
in $D$ are constants.
A {\em ground clause $D$ is linear} if whenever a constant occurs in two terms in $D$ whose root symbol is in $\Sigma$, the two terms are identical, and
if no term which starts with a $\Sigma$-function contains two occurrences of the same constant.
\label{flat}
\end{defi} 
 \begin{defi}[\cite{Ihlemann-Sofronie-2010}]
With the above notations, let $\Psi$ be a map associating with 
$\K$ and a set of $\Pi^C$-ground terms $T$ 
a set $\pK(T)$ of $\Pi^C$-ground terms. 
We call $\pK$ a \emph{term closure operator} if the following
holds for all sets of ground terms $T, T'$: 
\begin{enumerate}
\item $\mathrm{est}(\K, T) \subseteq \pK(T)$,
\item $T \subseteq T' \Rightarrow \pK(T) \subseteq \pK(T')$,
\item $\pK(\pK(T)) \subseteq \pK(T)$,
\item for
  any map $h: C \rightarrow C$, $\bar{h}(\pK(T)) = \Psi_{\bar{h}\K}(\bar{h}(T))$,
  where $\bar{h}$ is the canonical extension of $h$ to extension
  ground terms.
\end{enumerate}
\end{defi}

\begin{thm}[\cite{sofronie-ihlemann-jacobs-tacas08,
    Ihlemann-Sofronie-2010}]
Let ${\mathcal T}_0$ be a first-order theory and $\K$ a set of universally closed flat clauses in the signature
$\Pi$. The following hold: 
\begin{enumerate-}
\item If all clauses in $\K$ are linear  
and $\Psi$ is a term closure operator with the property 
that for every flat set of ground terms $T$, $\Psi(T)$ is
flat then either of the conditions $({\sf Emb}_{w,f}^\Psi)$ and $({\sf EEmb}_{w,f}^\Psi)$ implies
$({\sf Loc}_f^{\Psi})$.
\item If  the extension ${\mathcal T}_0 \subseteq {\mathcal T} {=} {\mathcal T}_0 {\cup} {\mathcal K}$
satisfies $({\sf Loc}_f^{\Psi})$ then $({\sf Emb}_{w,f}^\Psi)$ holds. 
\end{enumerate-}
\label{check-loc}
\end{thm}

\subsection{Examples of Local Theory Extensions} 
\label{sect:examples-local-theories}

Using a variant of Theorem~\ref{check-loc}, 
in \cite{sofronie-cade-05} we gave several examples of local 
theory extensions: 

\begin{enumerate-}
\item any extension $\T_0 \cup {\sf UIF}_{\Sigma}$ of a theory $\T_0$
  with free functions in a set $\Sigma$; 
\item extensions of a theory $\T_0$ with signature $\Sigma_0$
  having an injective function (constructor) $c$ with arity $n$
with suitable selector functions $s_1, \dots, s_n$; 
\item extensions of ${\mathbb R}$  with one or several 
  functions satisfying conditions such as boundedness, or boundedness
  on the slope; 
\item extensions of partially ordered 
theories -- in a class ${\sf Ord}$ 
consisting of the theories of posets, 
(dense) totally-ordered sets, semilattices, 
(distributive) lattices, Boolean algebras, or 
${\mathbb R}$ --  with 
a monotone function $f$, i.e.\ satisfying: 

\smallskip
~~~~~~~~~~~~~$ ({\sf Mon}(f)) \quad \quad \displaystyle{\bigwedge_{i =
    1}^n} x_i \leq y_i \rightarrow f(x_1, \dots, x_n) \leq f(y_1,
\dots, y_n).$

\item Generalized monotonicity conditions -- combinations of 
monotonicity in some arguments 
and antitonicity in other arguments  -- as well as extensions with
functions defined by case distinction (over a disjoint set of
conditions) were studied in 
\cite{Sofronie-Ihlemann07}. 
\end{enumerate-} 
We now present some more examples which were studied in 
\cite{Sofronie-lmcs}. 

\begin{thm}[\cite{Sofronie-lmcs}]
We consider the following base theories ${\mathcal T}_0$: 
(1) ${\mathcal P}$ (posets), 
(2) ${\sf  TOrd}$ (totally-ordered sets), 
(3) ${\sf SLat}$ (semilattices), 
(4) ${\sf DLat}$ (distributive lattices), 
(5) ${\sf Bool}$ (Boolean algebras),  
(6) the theory ${\mathbb R}$ of reals resp.\ 
${\sf LI}(\mathbb R)$ (linear arithmetic over ${\mathbb R}$), or 
the theory ${\mathbb Q}$ of rationals resp.\  
${\sf LI}(\mathbb Q)$ (linear arithmetic over ${\mathbb Q}$),   
or (a subtheory of) the theory of integers (e.g.\ Presburger arithmetic). 
The following theory extensions are local:

\
\begin{enumerate-} 
\item[(a)] Extensions of any theory ${\mathcal T}_0$ 
for which $\leq$ is reflexive with functions satisfying boundedness 
$({\sf Bound}^t(f))$ or guarded boundedness $({\sf GBound}^t(f))$ conditions 

\

$({\sf Bound}^t(f)) \quad  \quad \forall x_1, \dots, x_n (f(x_1, \dots, x_n) \leq t(x_1, \dots, x_n))$ 

\

$({\sf GBound}^t(f)) \quad \forall x_1, \dots, x_n (\phi(x_1, \dots, x_n) \rightarrow f(x_1, \dots, x_n) \leq t(x_1, \dots, x_n)),$

\

\noindent 
where $t(x_1, \dots, x_n)$ is a term in the base signature $\Pi_0$ and 
$\phi(x_1, \dots, x_n)$ a conjunction of literals in the signature $\Pi_0$, 
whose variables are in $\{ x_1, \dots, x_n \}$.

\medskip 
\item[(b)] Extensions of any theory  ${\mathcal T}_0$  in (1)--(6)  
with ${\sf Mon}(f) \wedge {\sf Bound}^t(f)$, if $t(x_1, \dots, x_n)$ is 
a term in the base signature $\Pi_0$ in the variables $x_1, \dots, x_n$ 
such that for every model of ${\mathcal T}_0$ the associated function is
monotone in the variables $x_1, \dots, x_n$.

\medskip
\item[(c)] Extensions of any theory $\T_0$ in  (1)--(6) 
with functions satisfying ${\sf Leq}(f,g) \wedge {\sf Mon}(f)$.

\

$({\sf Leq}(f,g)) \quad \forall x_1, \dots, x_n (\bigwedge_{i = 1}^n x_i \leq y_i \rightarrow f(x_1, \dots, x_n) \leq g(y_1, \dots, y_n))$

\

\item[(d)] Extensions of 
  any theory $\T_0$ which is one of the totally-ordered theories in (2) or (6)  (i.e.\ the theory ${\sf
    TOrd}$ of totally ordered sets or the theory of real numbers)
with functions satisfying 
${\sf SGc}(f,g_1, \dots, g_n) \wedge {\sf Mon}(f, g_1, \dots, g_n)$.

\

$({\sf SGc}(f,g_1, \dots, g_n)) \quad \forall x_1,\dots, x_n, x ( \bigwedge_{i = 1}^n x_i  \leq g_i(x) \rightarrow f(x_1, \dots, x_n) \leq x)$

\

\item[(e)] Extensions of  any theory $\T_0$ in (1)--(3) 
with functions satisfying ${\sf SGc}(f,g_1) \wedge {\sf Mon}(f, g_1)$.
\end{enumerate-}

\medskip
\noindent All the extensions above satisfy condition $({\sf Loc}_f)$. 
\label{examples-local}
\end{thm}

\subsection{Locality Transfer Results}
In \cite{Ihlemann-Sofronie-2010} we analyzed the way locality results 
can be transferred. 
Property ${\sf (EEmb_{w,f})}$, for instance, is preserved if we enrich 
the base theory $\T_0$:  

\begin{thm}[${\sf (EEmb)}$ Transfer,
  \cite{Ihlemann-Sofronie-2010}] 
\label{cor:eemb-transfer}
Let $\Pi_0 = (\Sigma_0, {\sf Pred})$ be a signature, 
$\T_0$ a theory in $\Pi_0$, 
$\Sigma_1$ and $\Sigma_2$ two disjoint sets of new function symbols,
$\Pi_i :=  (\Sigma_0 \cup \Sigma_i, {\sf Pred})$, $i = 1,2$. 
Assume that $\T_2$ is a $\Pi_2$-theory with $\T_0 \subseteq \T_2$, 
and $\K$ is a set of universally closed $\Pi_1$-clauses. 
If the extension $\T_0 \subseteq \T_0 \cup \K$ enjoys $({\sf EEmb}_{w,f})$ 
then so does the extension $\T_2 \subseteq \T_2 \cup \K.$

In particular, if $\K$ is flat and linear then the extension 
$\T_2 \subseteq \T_2 \cup \K$ satisfies condition $({\sf Loc}_f)$. 
If all the variables in clauses in ${\mathcal K}$ occur below $\Sigma_1$-functions, 
and ground satisfiability is decidable in ${\mathcal T}_2$, then 
ground satisfiability is decidable in ${\mathcal T}_2 {\cup} {\mathcal K}$. 
\label{cor-ext}
\end{thm}
The result extends in a natural way to the case of  
$({\sf EEmb}^{\Psi}_{w,f})$ and $\Psi$-locality. 
Theorem~\ref{cor-ext} is a very useful result, 
which allows us to identify a large number 
of local extensions. Below we include an
example from \cite{Ihlemann-Sofronie-2010}. 
\begin{exa}[\cite{Ihlemann-Sofronie-2010}]
Let ${\sf Lat}$ be the theory of lattices and 
${\sf Mon}_f = \{ \forall x, y ~ (x {\leq} y \rightarrow f(x) {\leq} f(y)) \}$ 
be the axiom expressing monotonicity of a new function symbol $f$. 
We can prove that the 
extension ${\sf Lat} \subseteq {\sf Lat} {\cup} {\sf Mon}_f$ satisfies 
condition $({\sf Comp}_{w,f})$ hence also $({\sf EEmb}_{w,f})$. 
By Theorem~\ref{cor-ext},    
${\T} \subseteq {\T} {\cup} {\sf Mon}_f$ satisfies 
condition $({\sf EEmb}_{w,f})$, hence 
$({\sf Loc}_f)$ for any extension $\T$ of the theory of lattices
(i.e.\ for the theory of distributive lattices, Heyting algebras, 
Boolean algebras, any theory with a total order -- e.g.\ 
the (ordered) theory of integers or of reals, etc.).
\end{exa}

\begin{thm}
\label{cor-ext-2}
Let $\Pi_0 = (\Sigma_0, {\sf Pred})$ be a signature, 
$\T_0$ a theory in $\Pi_0$, 
$\Sigma_P$ and $\Sigma$ two disjoint sets of new function symbols,
$\Pi_P :=  (\Sigma_0 \cup \Sigma_P, {\sf Pred})$, and 
$\Pi :=  (\Sigma_0 \cup \Sigma_P \cup \Sigma, {\sf Pred})$. 
Let $\Gamma$ be a universal $\Pi_P$-formula, and 
$\K$ be a set of $\Pi$-clauses. 
If the extension $\T_0 \subseteq \T_0 \cup \K$ enjoys $({\sf Comp}_{f})$ 
then so does the extension $\T_0 \cup \Gamma \subseteq \T_0 \cup \Gamma \cup \K.$

In particular, if $\K$ is flat and linear then the extension 
$\T_0 \cup \Gamma \subseteq \T_0 \cup \Gamma \cup \K$ satisfies condition $({\sf Loc}_f)$. 
\end{thm}

\begin{proof}Let $\A$ be a weak partial model of $\T \cup 
\Gamma \cup \K$ in which all $\Pi_P$-functions are total. 
Then $\A$ is a weak partial model of $\T \cup \K$, hence it weakly 
embeds into a total model $\B$ of $\T \cup \K$ such that 
$\A$ and $\B$ have the same support. 

Let $h : \A \rightarrow \B$ be the weak embedding.  
As (i)  all 
$\Pi_P$-functions are totally defined in $\A$; 
(ii) $\A$ and $\B$ have 
the same support, 
(iii) $\Gamma$ is a universal $\Pi_P$-formula which 
holds in $\A$,
it follows that $\Gamma$ holds also in $\B$. Thus, $\B$ is a 
total model of $\T_0 \cup \Gamma \cup \K$.
\end{proof}

\begin{thm}[\cite{Ihlemann-Sofronie-2010}]
\label{th-transfer12}
Let ${\mathcal T}_0$ be a theory. Assume that ${\mathcal T}_0$ has a model completion 
${\mathcal T}_0^*$ such that ${\mathcal T}_0 \subseteq {\mathcal T}_0^*$. 
Let ${\mathcal T} = {\mathcal T}_0 \cup {\mathcal K}$ be an extension of ${\mathcal T}_0$
with new function symbols $\Sigma$ whose properties are axiomatized by a 
set of flat and linear clauses ${\mathcal K}$ 
(all of which contain symbols in $\Sigma$). 
\begin{enumerate-}
\item [(1)] Assume that: 
\begin{enumerate-} 
\item[(i)] Every model of 
${\mathcal T}_0 \cup {\mathcal K}$ embeds 
into a model of ${\mathcal T}_0^* \cup {\mathcal K}$. 
\item[(ii)] ${\mathcal T}_0 \cup {\mathcal K}$ is a local extension of ${\mathcal T}_0$. 
\end{enumerate-}

\

\noindent 
Then ${\mathcal T}_0^* \subseteq {\mathcal T}_0^* \cup {\mathcal K}$ satisfies 
condition $({\sf EEmb}_{w,f})$, 
hence if ${\mathcal K}$ is a set of flat and linear clauses 
then ${\mathcal T}_0^*\subseteq {\mathcal T}_0^* \cup {\mathcal K}$ is a local extension. 
\item[(2)] If all variables in ${\mathcal K}$ occur below an extension 
function and 
${\mathcal T}_0^* \cup {\mathcal K}$ is a local extension of ${\mathcal T}_0^*$ 
then ${\mathcal T}_0 \cup {\mathcal K}$ is a local extension of ${\mathcal T}_0$. 
\end{enumerate-}
\end{thm}
The result extends in a natural way to $\Psi$-locality. 
These results were used in \cite{Ihlemann-Sofronie-2010} 
to give further examples of local theory
extensions: 

\begin{exa}[\cite{Ihlemann-Sofronie-2010}] The following hold: 
\begin{enumerate}
\item[(1)] The extension of the theory ${\sf TOrd}$ of total 
orderings with a strictly monotone function, i.e.\ a function 
$f$ satisfying the axiom: 
\[ {\sf SMon}(f) \quad  \quad   \forall x, y ( x < y \rightarrow f(x) < f(y)) \]
satisfies condition $({\sf Loc}_f)$. 

To show this, we used the fact  that the 
model completion ${\sf TOrd}^*$ of ${\sf TOrd}$ 
is the theory of dense total orderings without endpoints, 
and showed that the extension ${\sf TOrd}^* \subseteq 
{\sf TOrd}^* \cup {\sf SMon}(f)$ satisfies  condition $({\sf
  EEmb}_{w,f})$, hence it satisfies condition $({\sf Loc}_f)$. 

\item[(2)] The extension of the pure theory of 
equality with a function $f$ satisfying 
\[ {\sf Inj}(f) \quad \quad   \forall x, y \, ( x \not\approx y \rightarrow f(x) \not\approx
f(y)) \] 
is local.  (This can be proved in a similar way, using the fact that
the 
model completion of the pure theory of equality is the theory of
infinite sets.)
\end{enumerate}
\end{exa}

\subsection{Hierarchical Reasoning in Local Theory Extensions}
Consider a $\Psi$-local theory extension 
${\mathcal T}_0 \subseteq {\mathcal T}_0 \cup {\mathcal K}$. 
Condition $({\sf Loc}_f^{\Psi})$ requires that for every finite set $G$ of ground 
$\Pi^C$ clauses: 
\[ {\mathcal T}_0 \cup {\mathcal K} \cup G \models \perp \text{ if and 
  only if  } 
{\mathcal T}_0 \cup {\mathcal K}[\Psi_{\mathcal K}(G)] \cup G \models 
\perp.\] 

\

\noindent 
In all clauses in ${\mathcal K}[\Psi_{\mathcal K}(G)] \cup G$ the function 
symbols in $\Sigma$ only have ground terms as arguments, so  
${\mathcal K}[\Psi_{\mathcal K}(G)] {\cup} G$ can be flattened 
and purified\footnote{ i.e.\ the function symbols in $\Sigma$ 
are separated from the other symbols.}  
by introducing, in a bottom-up manner, new  
constants $c_t \in C$ for subterms $t {=} f(c_1, \dots, c_n)$ where $f {\in}
\Sigma$ and $c_i$ are constants, together with 
definitions $c_t {\approx} f(c_1, \dots, c_n)$ which are all included
in a set ${\sf Def}$. 
%
We thus obtain a set of clauses ${\mathcal K}_0 {\cup} G_0 {\cup} {\sf Def}$, 
where ${\mathcal K}_0$ and $G_0$ do 
not contain $\Sigma$-function symbols and ${\sf Def}$ contains clauses of the form 
$c {\approx} f(c_1, \dots, c_n)$, where $f {\in} \Sigma$, $c, c_1, \dots,
c_n$ are constants. 
\begin{thm}[\cite{sofronie-cade-05,sofronie-frocos-07,sofronie-ihlemann-jacobs-tacas08}]
Let ${\mathcal K}$ be a set of clauses. 
Assume that 
${\mathcal T}_0 \subseteq {\mathcal T}_1 = {\mathcal T}_0 \cup {\mathcal K}$ is a 
$\Psi$-local theory extension. 
For any finite set $G$ of ground clauses, 
let ${\mathcal K}_0 \cup G_0 \cup {\sf Def}$ 
be obtained from ${\mathcal K}[\Psi_{\mathcal K}(G)] \cup G$ by flattening and purification, 
as explained above. 
Then the following are equivalent to ${\mathcal T}_1 \cup G \models \perp$: 

\medskip
\begin{enumerate-}
\item ${\mathcal T}_0 {\cup} {\mathcal K}[\Psi_{\mathcal K}(G)] {\cup} G \models \perp.$ 
\item ${\mathcal T}_0 \cup {\mathcal K}_0 \cup G_0 \cup {\sf Con}_0 \models \perp,$ where 
$\displaystyle{{\footnotesize {\sf Con}_0  {=} \{ \bigwedge_{i = 1}^n c_i 
    {\approx} d_i \rightarrow c {\approx} d \, {\mid}
\begin{array}{l}
f(c_1, \dots, c_n) {\approx} c {\in} {\sf Def}\\
f(d_1, \dots, d_n) {\approx} d {\in} {\sf Def} 
\end{array} \}}}.$
\end{enumerate-}
\label{lemma-rel-transl}
\end{thm} 

\noindent We illustrate the ideas on an example first presented in
\cite{Sofronie-lmcs}. We chose this example because in Section~\ref{interp-loc}
we will use it to compare the method of computing interpolants in 
\cite{Sofronie-lmcs} with the method presented in this paper. 

\begin{exa}[\cite{Sofronie-lmcs}]
\label{example-hierarchic}
Let ${\mathcal T}_1 = \T_0 \cup {\sf SGc}(f, g) \cup {\sf Mon}(f, g)$
be the extension of the theory $\T_0 = {\sf SLat}$ of semilattices with two monotone functions
$f, g$ satisfying the semi-Galois condition 
\[ {\sf SGc}(f, g): \quad \forall x, y (x \leq g(y) \rightarrow f(x)
\leq y).\] 
Consider the following ground formulae $A$, $B$ in the signature of 
${\mathcal T}_1$: 

\bigskip
\noindent 
~~~~~~~~~~~~~$A:~~ d \leq g(a) ~\wedge~ a \leq c \quad \quad \quad \quad 
B:~~  b \leq d ~\wedge~ f(b) \not\leq c$

\bigskip
\noindent 
where $c$ and $d$ are shared constants. Let $G = A \wedge B$. 
By Theorem~\ref{examples-local}(e), 
${\mathcal T}_1$ is a local extension of the theory of semilattices. To prove that 
$G \models_{{\mathcal T}_1} \perp$ we proceed as follows:

\bigskip
\noindent
{\bf Step 1:} {\em Use locality.} By the locality condition, 
$G$ is unsatisfiable with respect to 
${\sf SLat} \wedge {\sf SGc}(f, g) \wedge {\sf Mon}(f, g)$ iff 
${\sf SLat} \wedge {\sf SGc}(f, g)[G] \wedge 
{\sf Mon}(f, g)[G] 
\wedge G$ has no weak partial model in which all terms in $G$
are defined. The extension terms occurring in $G$ are $f(b)$ and 
$g(a)$, hence:
\begin{eqnarray*}
{\sf Mon}(f, g)[G] & = & \{ a \leq a \rightarrow g(a) \leq g(a),~~ b \leq b \rightarrow f(b) \leq f(b) \} \\
{\sf SGc}(f, g)[G] & = & \{ b \leq g(a) \rightarrow f(b) \leq a \}  
\end{eqnarray*}

\bigskip
\noindent 
{\bf Step 2:} {\em Flattening and purification.}
We purify and flatten the formula ${\sf SGc}(f, g) \wedge {\sf Mon}(f, g)$ by 
replacing the ground terms starting with $f$ and $g$ with new constants. 
We obtain a set of definitions ${\sf Def} = \{ a_1 \approx g(a), b_1
\approx f(b) \}$, and a 
conjunction of formulae in the base signature, $G_0 \wedge {\sf SGc}_0
\wedge {\sf Mon}_0$
(where $G_0 = A_0 \wedge B_0$ is the purified form of $G = A \wedge B$).

\bigskip
\noindent {\bf Step 3:} 
{\em  Reduction to testing satisfiability in ${\mathcal T}_0$.}
As the extension ${\sf SLat} \subseteq {\mathcal T}_1$ is local, by 
Theorem~\ref{lemma-rel-transl} we know that
$$G \models_{{\mathcal T}_1} \perp  \text{ iff } 
G_0 \wedge  {\sf SGc}_0   \wedge {\sf Mon}_0 \wedge {\sf Con}_0
\text{ is unsatisfiable with respect to } {\sf SLat},$$ where
${\sf Con}_0 = {\sf Con}[G]_0$ consists of the flattened form of 
those instances 
of the congruence axioms containing only $f$- and $g$-terms which 
occur in ${\sf Def}$. 

$$\begin{array}{l|ll}
\hline 
{\sf Extension} & ~~~~~~{\sf Base} \\
{\sf Def} & ~G_0  \wedge {\sf SGc}_0 \wedge {\sf
  Mon}_0 & \wedge  {\sf Con}_0 ~~~~~~~~~~~ \\
\hline 
D_A =  a_1 \approx g(a) ~   & ~A_0 = d \leq a_1  \wedge a \leq c & {\sf SGc}_0   = b \leq a_1 \rightarrow b_1 \leq a  \\
D_B =  b_1 \approx f(b)  & ~B_0 = b \leq d \wedge  b_1 \not\leq c &
{\sf Con}_0 \wedge {\sf Mon}_0: \\
& & \quad  {\sf Con}_A \wedge {\sf Mon}_A = a \lhd a \rightarrow a_1 \lhd a_1, \lhd \in \{ \approx, \leq \}\\
& &  \quad {\sf Con}_B \wedge {\sf Mon}_B =  b \lhd b \rightarrow b_1 \lhd b_1,~ \lhd \in \{ \approx, \leq \} \\
\hline 
\end{array}$$

\

\noindent 
It is easy to see that $G_0 \wedge {\sf SGc}_0 \wedge {\sf Mon}_0 \wedge {\sf Con}_0 $  
is unsatisfiable with respect to ${\mathcal T}_0$: 
$G_0= A_0 \wedge B_0$ entails $b \leq a_1$, together with ${\sf SGc}_0$ this 
yields $b_1 \leq a$, which together with $a \leq c$ and $b_1 \not\leq c$ 
leads to a contradiction.
\end{exa}

\subsection{Chains of Theory Extensions} 
We can also consider chains of theory extensions: 
\begin{eqnarray*}
{\mathcal T}_0 & \subseteq &  {\mathcal T}_1  =  {\mathcal T}_0 \cup
{\mathcal K}_1  ~~~\subseteq~~~  {\mathcal T}_2 = {\mathcal T}_0  \cup
{\mathcal K}_1 \cup {\mathcal K}_2 ~~~\subseteq~~~ \dots ~~~ \subseteq~~~  {\mathcal
  T}_n = {\mathcal T}_0 \cup {\mathcal K}_1 \cup...\cup {\mathcal K}_n
\end{eqnarray*}
in which each theory is a local extension of the preceding one. 

\medskip
\noindent 
For a chain of local extensions  a satisfiability check w.r.t.\ the last extension can 
be reduced (in $n$ steps) to a satisfiability check w.r.t.\  
${\mathcal T}_0$. The only restriction we need to impose in order to
ensure that such a reduction is possible is that at each step the
clauses reduced so far need to be ground. 
Groundness is assured if each variable in a clause appears at least
once under 
an extension function.
This iterated instantiation procedure for chains of local theory
extensions has been implemented in H-PILoT \cite{hpilot}.\footnote{H-PILoT  ~allows the user to specify a chain of extensions by 
tagging the extension functions with their place in the chain 
(e.g., if $f$ occurs in ${\mathcal K}_3$ but not in ${\mathcal K}_1
\cup {\mathcal K}_2$ it is declared as level 3). }

\begin{exa}
\label{ex:chains}
Let $\T_0$ be the theory of dense total orderings without
endpoints. Consider the extension of $\T_0$ with functions $\Sigma_1 =
\{ f, g, h, c \}$
whose properties are axiomatized by:

\[ 
\K := \{ ~~\forall x (x \leq c \rightarrow g(x) \approx f(x)), ~~~ \forall x  ( c < x \rightarrow g(x) \approx h(x)) ~~\}.
 \] 

\medskip

\noindent The extension $\T_0 \subseteq \T_0 \cup \K$ can be ``refined'' to the 
following chain of theory extensions:

\[ \T_0 \subseteq \T_0 \cup {\sf UIF}_{\{ f, h \}} \subseteq (\T_0 \cup {\sf
  UIF}_{\{ f, h \}}) \cup \K = \T_0 \cup \K. \]  

\medskip
\begin{itemize}
\item The theory $\T_0 \cup {\sf UIF}_{\{ f, h \}}$ is a local extension of $\T_0$
because extensions with free function symbols are local 
(\cite{sofronie-cade-05}, see also the comments in 
Section~\ref{sect:examples-local-theories}). 

\bigskip
\item $\T_0 \cup \K$ is the extension of $\T_0 \cup {\sf UIF}_{\{ f, h \}}$ with
the function $g$, defined by case distinction (this is described in
the axioms $\K$). By the results in \cite{Sofronie-Ihlemann07} such 
extensions are also local. 
\end{itemize}
\noindent In fact, both extensions satisfy 
condition ${\sf Comp}_f$. 

\bigskip

\noindent 
Let $G$ be a set of ground clauses over the signature $\Sigma_0 \cup
\Sigma_1$. Then the following are equivalent:  
\begin{enumerate}
\bigskip
\item $\T_0 \cup \K \cup G \models \perp$; 

\bigskip
\item $(\T_0 \cup {\sf UIF}_{\{ f, h \}}) \cup \K \cup G \models
  \perp$; 

\bigskip
\item $(\T_0 \cup {\sf UIF}_{\{ f, h \}}) \cup \K[G] \cup G \models
  \perp$, ~~~
where $\K[G]$ is the set of all instances of $\K$ in which the terms 
starting with the function $g$ are ground terms occurring in $\K$ or
$G$; 

\bigskip
\item $(\T_0 \cup {\sf UIF}_{\{ f, h\}}) \cup G^1_0 \cup {\sf Def}^g
  \models \perp$, ~~~
where  $G^1 = \K[G] \cup G$ is a ground formula, $G^1_0$ is obtained
from $G^1$ by purification (replacing all ground terms starting with
$g$ with a new constant) and ${\sf Def}^g$ is the corresponding set of 
definitions, as explained in 
Theorem~\ref{lemma-rel-transl}; 

\bigskip
\item $\T_0 \cup G_0 \cup {\sf Def}^g\cup {\sf Def}^{f,h} \models
  \perp$, ~~~
where
$G_0$ is obtained from $G^1_0$ after one more round of purification in
which all ground terms starting with
$f$ and $h$ are replaced with new constants and 
  ${\sf Def}^{f,h}$ is the corresponding set of definitions, as explained
  in Theorem~\ref{lemma-rel-transl}; 

\bigskip
\item $\K_0 \cup G_0 \cup {\sf Con}_0$ is unsatisfiable w.r.t.\
  $\T_0$, ~~~ 
where ${\sf Con}_0$ is the set of congruence axioms
  corresponding to the set ${\sf Def} = {\sf Def}^g\cup {\sf
    Def}^{f,h} $ 
of definitions as explained in Theorem~\ref{lemma-rel-transl}. 
\end{enumerate}

\end{exa} 

\section{Symbol Elimination in Theory Extensions}
\label{symb-elim}

Let $\Pi_0 = (\Sigma_0, {\sf Pred})$. 
Let ${\mathcal T}_0$ be a $\Pi_0$-theory.  
We consider theory extensions $\T_0
\subseteq \T = \T_0 \cup \K$, in which among the extension functions
we identify a set of {\em parameters} $\Sigma_P$ (function and constant symbols). 
Let $\Sigma$ be a signature consisting of extension symbols which are not
parameters (i.e.\ such that $\Sigma \cap (\Sigma_0 \cup \Sigma_P) =
\emptyset$).
%
We assume that ${\mathcal K}$ is a set of clauses 
in the signature $\Pi_0 {\cup} \Sigma_P {\cup} \Sigma$ in which all 
variables occur also below functions in $\Sigma_1 = \Sigma_P \cup
\Sigma$.

\medskip

\noindent 
We identify situations in which we can
generate, 
for every ground formula $G$,   
a (universal) formula $\Gamma$ representing a family of
constraints on the parameters of $G$ 
such that $\T \cup \Gamma \cup G \models \perp$. 
We consider base theories $\T_0$ such that $\T_0$ or its model
completion $\T^*_0$ allows quantifier elimination, and use quantifier 
elimination to generate the formula $\Gamma$.  
Thus, we assume that one of the following conditions holds: 

\medskip

\begin{description}
\item[(C1)] $\T_0$ allows quantifier elimination, or 
\item[(C2)] $\T_0$ has a model
completion $\T^*_0$ which allows quantifier elimination. 
\end{description}

\medskip

\noindent 
Let $G$ be a finite set of ground clauses, and $T$ a finite set
of ground terms over the signature $\Pi_0 {\cup} \Sigma_P {\cup} \Sigma {\cup}
C$, where $C$ is a set of additional constants. 
We construct a universal 
formula $\forall y_1 \dots y_n \Gamma_T(y_1, \dots, y_n)$ over the
signature $\Pi_0 {\cup} \Sigma_P$ by following
the Steps 1--5 below: 

\medskip

\begin{description}
\item[Step 1] Let $\K_0 \cup G_0 \cup {\sf Con}_0$ be the set of
  $\Pi_0^C$ clauses obtained from $\K[T] \cup G$  after the purification step
  described in  Theorem~\ref{lemma-rel-transl} (with set of extension
  symbols $\Sigma_1$). 

\medskip

\item[Step 2] Let $G_1 = {\mathcal K}_0 \cup G_0\cup {\sf Con}_0$. 
Among the constants in $G_1$, we identify 
\begin{enumerate}
\item[(i)] the constants
${\overline c_f}$, $f \in \Sigma_P$, where either $c_f = f \in
\Sigma_P$ is a constant
parameter, or $c_f$ is 
introduced by a definition $c_f := f(c_1, \dots, c_k)$ in the hierarchical
reasoning method, 
\item[(ii)] all constants  ${\overline c_p}$ 
occurring as arguments of functions in $\Sigma_P$ in such definitions. 
\end{enumerate}
Let ${\overline  c}$ be the remaining constants. We replace the
constants in ${\overline  c}$
with existentially quantified variables ${\overline x}$ in $G_1$,
i.e.\ 
replace $G_1({\overline c_p}, {\overline c_f}, {\overline c})$ 
with $G_1({\overline c_p}, {\overline c_f}, {\overline x})$, and
consider the formula
$\exists {\overline x} G_1({\overline c_p},{\overline c_f}, {\overline x})$.

\medskip

\item[Step 3] Using the method for quantifier elimination in 
${\mathcal T}_0$ (if Condition (C1) holds) or in $\T_0^*$ (if Condition (C2)
holds) we can construct 
a formula  $\Gamma_1({\overline c_p}, {\overline c_f})$ equivalent to 
$\exists {\overline x} G_1({\overline c_p}, {\overline c_f},{\overline
  x})$
w.r.t.\ $\T_0$ (resp. $\T_0^*$). 

\medskip

\item[Step 4] Let $\Gamma_2({\overline c_p})$ be the formula 
obtained by replacing back in $\Gamma_1({\overline c_p}, {\overline c_f})$ 
the constants $c_f$ introduced by definitions $c_f := f(c_1, \dots,
c_k)$ with the terms $f(c_1, \dots,c_k)$. We replace ${\overline c_p}$ with existentially quantified variables ${\overline y}$. 

\medskip

\item[Step 5] Let $\forall {\overline y} \Gamma_T({\overline y})$ be
  $\forall {\overline y} \neg \Gamma_2({\overline y})$. 
\end{description}

\medskip
 
\noindent A similar approach is used in \cite{sofronie-ijcar10} 
for generating constraints on parameters 
which guarantee safety of parametric systems. 
We show that $\forall {\overline
  y} \Gamma_T({\overline y})$ guarantees unsatisfiability of 
$G$ and further study the properties of these formulae. 
At the end of Section~\ref{interp-loc}  we briefly indicate how this 
can be used for interpolant generation.

\subsection{Case 1: \texorpdfstring{$\T_0$}{T\_0} allows quantifier elimination} 
We first
analyze the case in which $\T_0$ allows quantifier elimination.
\begin{thm}
Assume that ${\mathcal T}_0$ allows quantifier elimination.  
For every finite set of ground clauses $G$, and every finite set $T$ of terms over the 
signature $\Pi_0 \cup \Sigma \cup \Sigma_P \cup C$ with 
${\sf est}(\K, G)
\subseteq T$ we can construct a universally quantified 
$\Pi_0 \cup \Sigma_P$-formula 
$\forall {\overline y} \Gamma_T({\overline y})$ 
with the following properties: 
\begin{enumerate}
\item[(1)] For every
structure ${\mathcal A}$ with signature 
$\Pi_0 \cup \Sigma \cup \Sigma_P \cup C$ 
which is a model of ${\mathcal T}_0 \cup {\mathcal K}$, if 
${\mathcal A} \models \forall {\overline y}
\Gamma_T({\overline y})$ then ${\mathcal A} \models \neg G$. 

\item[(2)] ${\mathcal T}_0 \cup \forall {\overline y}
\Gamma_T({\overline y}) \cup {\mathcal K} \cup
  G$ is unsatisfiable. 
\end{enumerate}
\label{inv-trans-qe}
\end{thm}
\begin{proof} Let $\forall {\overline y} \Gamma_T({\overline y})$ be the formula
obtained in Steps 1--5. 

\noindent (1) Let ${\mathcal A}$ be a $\Pi_0 \cup \Sigma \cup \Sigma_P \cup C$-structure such that 
${\mathcal A} \models {\mathcal T}_0 \cup {\mathcal K} \cup G$. 
Then ${\mathcal A} \models {\mathcal T}_0 \cup {\mathcal K}[T] \cup G$.  
Let ${\mathcal K}_0 \cup G_0
\cup {\sf Con}_0 \cup {\sf Def}$ be the formulae obtained from  ${\mathcal K}[T] \cup G$
after purification as explained in Theorem~\ref{lemma-rel-transl}.  
Clearly, ${\mathcal A} \models {\mathcal T}_0 \cup {\mathcal K}_0 \cup G_0
\cup {\sf Con}_0 \cup {\sf Def}$.\footnote{For simplicity, we here use
  the same symbol for $\A$ and its   expansion with new constants defined as in ${\sf Def}$.}

\smallskip
\noindent Let $G_1 = {\mathcal K}_0 \cup G_0\cup {\sf Con}_0$. 
Since  ${\mathcal  A} \models {\mathcal T}_0 \cup G_1 \cup  {\sf Def}$, we 
know that 
%
${\mathcal A} \models {\mathcal T}_0 \cup \exists {\overline x} G_1({\overline c_p}, {\overline c_f},
{\overline x}) \cup {\sf Def}.$
By quantifier elimination in ${\mathcal T}_0$ we can construct 
a formula  $\Gamma_1({\overline c_p}, {\overline c_f})$ equivalent to 
$\exists {\overline x} G_1({\overline c_p}, {\overline c_f},
{\overline x})$ w.r.t.\ $\T_0$.
Hence, ${\mathcal A} \models {\mathcal T}_0 \cup \Gamma_1({\overline c_p}, {\overline c_f}).$
Since ${\mathcal A}$ is also a model for ${\sf Def}$ we can replace
in $\Gamma_1$  the constants $c_f$ introduced by definitions 
$c_f := f(c_1, \dots, c_k)$ with the terms $f(c_1, \dots, c_k)$ they
replaced, thus obtaining the formula 
$\Gamma_2({\overline c_p})$, and ${\mathcal A} \models {\mathcal T}_0 \cup
\Gamma_2({\overline c_p})$. 
If $\Gamma_2({\overline y})$ is obtained
from $\Gamma_2({\overline c_p})$ by replacing the constants in
${\overline c_p}$ with new variables in ${\overline y}$, 
it follows that ${\mathcal A} \models {\mathcal T}_0 \cup
\exists y \Gamma_2({\overline y})$, hence 
(as $\Gamma_T = \neg \Gamma_2$),  
${\mathcal A} \models \exists {\overline y} \neg \Gamma_T({\overline
  y})$, i.e.\ ${\mathcal A} \models \neg \forall {\overline y} \Gamma_T({\overline
  y})$. 

We showed that if ${\mathcal A} \models {\mathcal T}_0 \cup {\mathcal K} \cup G$ then 
${\mathcal A} \models \neg \forall {\overline y} \Gamma_T({\overline
  y})$. 
Hence, if ${\mathcal A} \models {\mathcal T}_0 \cup \K \cup \forall  {\overline y} \Gamma_T({\overline
  y})$ then ${\mathcal A} \not\models {\mathcal T}_0 \cup {\mathcal K} \cup G$, hence 
$G$ is false in ${\mathcal A}$. 

\medskip
\noindent (2) The unsatisfiability of 
$\T_0 \cup \forall {\overline y}
\Gamma_T({\overline y})  \cup \K \cup G$ 
follows immediately from (1). 
\end{proof}
\noindent If we analyze the proof of Theorem~\ref{inv-trans-qe} we can make the following 
observations.
\begin{lem}
With the notation used in Steps 1--5 we can show that
the formulae $\Gamma_2({\overline c_p})$ and $\exists {\overline x} G_1({\overline c_p}, {\overline c_f},
{\overline x}) \wedge {\sf Def}$ are equivalent
modulo $\T_0 \cup {\sf UIF}_{\Sigma_P}$. 
\label{lemma:symb-elim}
\end{lem}
\begin{proof} We show that for every $\Pi_0 \cup \Sigma_P \cup \Sigma \cup C$-structure
${\mathcal A}$ which
is a model of $\T_0$, ${\mathcal A} \models \exists {\overline x} G_1({\overline c_p}, {\overline c_f},
{\overline x}) \wedge {\sf Def}$ if and only if 
${\mathcal A} \models \Gamma_2({\overline c_p})$. 
Assume that ${\mathcal A}$ is a model of $\T_0$ and of 
$\exists {\overline x} G_1({\overline c_p}, {\overline c_f},
{\overline x}) \wedge {\sf Def}$. As  $\exists {\overline x} G_1({\overline c_p}, {\overline c_f},
{\overline x}) $ and $\Gamma_1({\overline c_p}, {\overline c_f})$  are
equivalent w.r.t. $\T_0$ (the second is obtained from the first by
quantifier elimination)
it follows that ${\mathcal A}$ is a model of 
$\Gamma_1({\overline c_p}, {\overline c_f}) \wedge {\sf Def}$, hence it
is a model of $\Gamma_2({\overline c_p})$. 

Assume now that ${\mathcal A} \models 
\Gamma_2({\overline c_p})$. We can purify $\Gamma_2$ by introducing
new constants renaming the terms $f(c_1, \dots, c_n)$ according to the
definitions in ${\sf Def}$. 
The formula obtained this
way is 
$\Gamma_1({\overline c_p},
{\overline c_f}) \wedge {\sf Def}$. 
As  $\exists {\overline x} G_1({\overline c_p}, {\overline c_f},
{\overline x}) $ and $\Gamma_1({\overline c_p}, {\overline c_f})$  are
equivalent w.r.t. $\T_0$, it follows\footnote{For simplicity,
  we use the same symbol for $\A$ and its   
  expansion with new constants defined as in ${\sf Def}$.}  that 
${\mathcal A} \models \exists {\overline x} G_1({\overline c_p}, {\overline c_f},
{\overline x}) \wedge {\sf Def}$.
\end{proof}
\begin{thm}
If $T_1 \subseteq T_2$ then $\forall {\overline y} \Gamma_{T_1}({\overline y})$ entails 
  $\forall {\overline y} \Gamma_{T_2}({\overline y})$ (modulo $\T_0$). 
\label{mon}
\end{thm}
\begin{proof}Let $T_1, T_2$ be two finite sets of terms.
 If $T_1 \subseteq T_2$ then $\K[T_1] \subseteq \K[T_2]$. We denote by
$\K_1$ the purified form of $\K[T_1]$ and by $\K_2$ the purified form of
$\K[T_2]$, and let ${\sf Con}_i$ be the set of axioms corresponding to
the terms in $\K[T_i] \cup G$, $i = 1,2$. Then $\K_1 \cup {\sf Con}_1
\subseteq \K_2 \cup {\sf Con}_2$, hence 
$\K_2 \wedge G_0 \wedge {\sf Con}_2 \models \K_1 \wedge G_0 \wedge {\sf
  Con}_1$. 
Then every model of $\T_0$ which is a model of 
$\K_2 \wedge G_0 \wedge {\sf Con}_2$ is also a model of  
$\K_1 \wedge G_0 \wedge {\sf  Con}_1$. 

Let ${\overline c}$ denote the sequence
consisting of all constants in $\K_2 \wedge G_0 \wedge {\sf
  Con}_2$ and not in $\Sigma$ (a superset of the constants occurring
in $\K_1 \wedge G_0 \wedge {\sf Con}_1$). 
We regard the elements in ${\overline c}$ as
variables.\footnote{Instead of renaming the constants ${\overline c}$
  with  new variables ${\overline x}$, we here keep their names, but
  treat them as variables.} 

\medskip
\noindent 
We first show that $\exists {\overline c} (\K_2 \wedge G \wedge {\sf Con}_2) \models \exists {\overline
  c} (\K_1 \wedge G \wedge {\sf Con}_1)$. 

Indeed, let ${\mathcal A}$ be a model of $\exists {\overline c} (\K_2 \wedge G
\wedge {\sf Con}_2)$. Then there exists a valuation $\beta$ which assigns
values in $A$ to the variables in  ${\overline c}$ such that 
${\mathcal A}, \beta \models \K_2 \wedge G
\wedge {\sf Con}_2$. Then ${\mathcal A}^{\overline c}  \models \K_2 \wedge G
\wedge {\sf Con}_2$ (where ${\mathcal A}^{\overline c}$ is the expansion
of $\A$ with new constants ${\overline c}$, interpreted as specified by $\beta$). 
Then ${\mathcal A}^{\overline c}   \models \K_1 \wedge G_0\wedge {\sf
  Con}_1$, 
hence ${\mathcal A} \models \exists {\overline
  c} (\K_1 \wedge G_0 \wedge {\sf Con}_1)$. 
This shows that $\exists {\overline c} (\K_2 \wedge G \wedge {\sf Con}_2) \models \exists {\overline
  c} (\K_1 \wedge G \wedge {\sf Con}_1)$.

\medskip
\noindent We show that $\forall {\overline y} \Gamma_{T_1}({\overline y}) \models
\forall {\overline y} \Gamma_{T_2}({\overline y})$ modulo $\T_0$, i.e.\ that every model of $\T_0 \cup
\forall {\overline y} \Gamma_{T_1}({\overline y})$ is also a model of $\T_0 \cup
\forall {\overline y} \Gamma_{T_2}({\overline y})$. 

 Let $\A$ be a model of $\T_0$. 
Assume that ${\mathcal A} \not\models \forall {\overline y} \Gamma_{T_2}({\overline y})$. Then ${\mathcal A}
\models \exists {\overline y} \neg \Gamma_{T_2}({\overline y})$, hence
(using the chain of arguments used in the previous proofs and the fact
that $\A$ is a model of $\T_0$) ${\mathcal A}
\models \exists {\overline c} (\K_2 \wedge G_0 \wedge {\sf Con}_2)$. 
But then ${\mathcal A}
\models \exists {\overline c} (\K_1 \wedge G_0\wedge {\sf Con}_1)$, 
hence (again using the chain of arguments used in the previous proofs and the fact
that $\A$ is a model of $\T_0$) ${\mathcal A} \models \exists {\overline
  y} \neg \Gamma_{T_1}({\overline y})$, so ${\mathcal A} \not\models
\forall {\overline y} \Gamma_{T_1}({\overline y})$. 
\end{proof}

\noindent We denote by $\forall {\overline y} \Gamma_G({\overline y})$
the formula 
obtained when $T = {\sf est}(\K, G)$. 
\begin{thm} 
   If the extension ${\mathcal T}_0 \subseteq {\mathcal T}_0 \cup
  {\mathcal K}$ satisfies condition $({\sf Comp}_{f})$ and $\K$ is flat
  and linear then 
$\forall y \Gamma_G(y)$ is entailed by every
  conjunction $\Gamma$ of 
  clauses with the property that ${\mathcal T}_0 \cup \Gamma \cup
  {\mathcal K} \cup G$ is unsatisfiable (i.e.\ it is the weakest such constraint). 
\label{symb-elim-weakest}
\end{thm}

\begin{proof}We show that for every set $\Gamma$ of constraints on the parameters, 
if ${\mathcal T}_0 \cup \Gamma \cup {\mathcal K} \cup G$ is unsatisfiable then 
every model of $\T_0  \cup \Gamma$ is a model of  $\T_0 \cup \forall
y \Gamma_G(y)$. 

We know, by Theorem~\ref{cor-ext-2}, that if the extension 
${\mathcal T}_0 \subseteq {\mathcal T}_0 \cup
  {\mathcal K}$ satisfies condition $({\sf Comp}_{f})$ then also the
  extension 
${\mathcal  T}_0 \cup \Gamma \subseteq {\mathcal T}_0 \cup \Gamma \cup   {\mathcal K}$
  satisfies condition $({\sf Comp}_{f})$. If $\K$ is flat and linear then the
  extension is local. 
Let $T = {\sf est}(\K, G)$. 
By locality, ${\mathcal T}_0 \cup \Gamma \cup {\mathcal K} \cup G$ is
unsatisfiable
if and only if ${\mathcal T}_0 \cup \Gamma \cup {\mathcal K}[T] \cup G$
is unsatisfiable, if and only if  (with the notations in Steps 1--5) 
${\mathcal T}_0 \cup \Gamma \cup {\mathcal K}_0 \cup G_0
\cup {\sf Con}_0 \cup {\sf Def}$ is unsatisfiable. 
Let ${\mathcal A}$ be a model of $\T_0  \cup \Gamma$. Then ${\mathcal A}$ 
cannot be a model of ${\mathcal K}_0 \cup G_0
\cup {\sf Con}_0 \cup {\sf Def}$, so (with the notation used when
describing Steps 1--5) 
${\mathcal A} \not\models \Gamma_2({\overline
  c_p})$, i.e.\ ${\mathcal A} \not\models \exists {\overline y}
\Gamma_2({\overline  y})$. 
It follows that ${\mathcal A} \models \forall  {\overline y} \Gamma_G({\overline
  y})$. 
\end{proof}

\noindent A similar result holds for chains of local theory
extensions. 

\begin{thm} 
 Assume that we have the following chain of theory extensions:
\[ {\mathcal T}_0 ~~~\subseteq~~~ {\mathcal T}_0 \cup {\mathcal K}_1
~~~\subseteq~~~  
{\mathcal T}_0 \cup \K_1 \cup \K_2 ~~~\subseteq~~~ \dots ~~~\subseteq~~~ 
{\mathcal  T}_0 \cup \K_1 \cup \K_2 \cup \dots \cup \K_n \] 
where every extension in the chain satisfies condition 
$({\sf Comp}_{f})$, $\K_i$ are all flat
  and linear, and in all $\K_i$ all variables occur below the 
extension terms on level $i$. 

Let $G$ be a set of ground clauses, and let $G_1$ be 
the result of the hierarchical reduction of satisfiability 
of $G$ to a satisfiability test w.r.t.\ $\T_0$. Let $T(G)$ be the set
of all instances used in the chain of hierarchical reductions and let 
$\forall y \Gamma_{T(G)}(y)$ be the formula obtained by applying 
Steps 2--5 to $G_1$. 

Then $\forall y \Gamma_{T(G)}(y)$ 
is entailed by every
  conjunction $\Gamma$ of 
  clauses with the property that ${\mathcal T}_0 \cup \Gamma \cup
  {\mathcal K}_1 \cup \dots \cup \K_n  \cup G$ is unsatisfiable (i.e.\ it is the weakest such constraint). 
\label{symb-elim-weakest-chains}
\end{thm}

\begin{proof}We show that for every set $\Gamma$ of constraints on the parameters, 
if ${\mathcal T}_0 \cup \Gamma \cup {\mathcal K}_1 \cup \dots \cup \K_n \cup
G$ is unsatisfiable then 
every model of $\T_0  \cup \Gamma$ is a model of  $\T_0 \cup \forall
y \Gamma_{T(G)}(y)$. 

We know, by Theorem~\ref{cor-ext-2}, that if the extension 
\[ {\mathcal T}_0 \cup \K_1 \cup \dots \K_{i-1} \subseteq {\mathcal T}_0 \cup
\K_1 \cup \dots \K_{i-1} \cup {\mathcal K}_i \]
satisfies condition $({\sf Comp}_{f})$ then also the extension 
\[ {\mathcal   T}_0 \cup \Gamma \cup \K_1 \cup \dots \K_{i-1} \subseteq {\mathcal T}_0 \cup \Gamma \cup  \K_1 \cup \dots \K_{i-1} \cup {\mathcal K}_i \]
  satisfies condition $({\sf Comp}_{f})$. If $\K_i$ is flat and linear then the
  extension is local. 
By locality, the following are equivalent: 
\begin{itemize}
\item ${\mathcal T}_0 \cup \Gamma \cup \K_1 \cup \dots \cup \K_{n-1}
\cup \K_n \cup G \models \perp$;
\item ${\mathcal T}_0 \cup \Gamma \cup \K_1 \cup \dots \cup \K_{n-1}\cup
  {\mathcal K}_n[T_n] \cup G^n \models \perp$ ~~~
 where $G^n = G$ and $T_n = {\sf est}(\K_n, G^n)$; 
\item ${\mathcal T}_0 \cup \Gamma \cup \K_1 \cup \dots \cup
  \K_{n-1}[T_{n-1}] \cup G^{n-1} \models \perp$ ~~~ 
where $G^{n-1} =  {\mathcal K}_n[T_n]_0 \cup G^n_0 \cup {\sf
  Con}^n_0$ is the set of ground clauses obtained from  $ {\mathcal K}_n[T_n] \cup G^n$
after purification and adding the corresponding instances of
congruence axioms and $T_{n-1} = {\sf est}(\K_{n-1}, G^{n-1})$; \\
 ...
\item ${\mathcal T}_0 \cup \Gamma \cup \K_1[T_1] \cup G^1 \models
  \perp$ ~~~
where $G^1 =  {\mathcal K}_2[T_2]_0 \cup {G^2}_0 \cup {\sf
  Con}^2_0$ is the set of ground clauses obtained from $\K_2[T_2] \cup G^2$
after purification and adding the corresponding instances of
congruence axioms and $T_1 = {\sf est}(\K_1, G^1)$; 
\item ${\mathcal T}_0 \cup \Gamma \cup  {\mathcal K}_0 \cup G_0
\cup {\sf Con}_0 \cup {\sf Def} \models \perp$  ~~~
where $\K_0 =   {\mathcal K}_1[T_1]_0$, and $G_0 = G^1_0$ are the sets of ground
clauses obtained from $ {\mathcal K}_1[T_1]$ resp. $G^1$ after purification
and ${\sf Con}_0$ is the corresponding set of instances of congruence
axioms. 
\end{itemize}
Let ${\mathcal A}$ be a model of $\T_0  \cup \Gamma$. Then ${\mathcal A}$ 
cannot be a model of ${\mathcal K}_0 \cup G_0
\cup {\sf Con}_0 \cup {\sf Def}$, as $\T_0  \cup \Gamma \cup {\mathcal K}_0 \cup G_0
\cup {\sf Con}_0 \cup {\sf Def}$ has no models. 
Therefore  (with the notation used when
describing Steps 1--5),  
${\mathcal A} \not\models \Gamma_2({\overline
  c_p})$, i.e.\ ${\mathcal A} \not\models \exists {\overline y}
\Gamma_2({\overline  y})$. 
It follows that ${\mathcal A} \models \forall  {\overline y} \Gamma_{T(G)}({\overline
  y})$. 
\end{proof}

\begin{exa}
\label{ex1-symb-elim}
Let $\T_0$ be the theory of dense total orderings without
endpoints. Consider the extension of $\T_0$ with functions $\Sigma_1 =
\{ f, g, h, c \}$
whose properties are axiomatized by 

\[ \begin{array}{ll} 
\K := \{ & \forall x (x \leq c \rightarrow g(x) \approx f(x)), \\
& \forall x  ( c < x \rightarrow g(x) \approx h(x)) \quad \}.
\end{array} \] 
Assume $\Sigma_P = \{ f, h, c \}$ and $\Sigma = \{ g \}$. 
We are interested in generating a set of constraints on the functions $f$ and
$h$ which ensure that $g$ is monotone, e.g. satisfies
\[ {\sf Mon}(g): ~~ \forall x, y (x \leq y \rightarrow g(x) \leq
g(y)), \] 
i.e.\ a set $\Gamma$ of $\Sigma_0 \cup \Sigma_P$-constraints such that 
\[ \T_0 \cup \Gamma \cup \K \cup \{ c_1 \leq c_2, g(c_1) > g(c_2) \}
\text{ is unsatisfiable,} \] 
where 
$G = \{ c_1 \leq c_2,  g(c_1) > g(c_2) \}$ is the negation of ${\sf  Mon}(g)$.  

\

\noindent As explained in Example~\ref{ex:chains} 
we have the follwing chain of local theory extensions: 

\[ \T_0 \subseteq \T_0 \cup {\sf UIF}_{\{ f, h \}} \subseteq (\T_0 \cup 
{\sf  UIF}_{\{f, g\}}) \cup \K = \T_0 \cup \K. \]  
Both extensions satisfy the 
condition ${\sf Comp}_f$, and 
$\T_0 \cup \K \cup G$ is satisfiable iff
$\T_0  \cup \K[G] \cup G$ is satisfiable, where $\K[G]$ contains all
instances of $\K$ in which the terms starting with the extension
symbol $g$ are ground terms in $G$: 
\[ \begin{array}{lll} 
\K[G] := \{ & c_1 \leq c \rightarrow g(c_1) \approx f(c_1), & c_2 \leq c \rightarrow g(c_2) \approx f(c_2), \\
& c < c_1 \rightarrow g(c_1) \approx h(c_1), & c < c_2 \rightarrow
g(c_2) \approx h(c_2)) \quad \}.
\end{array} \]

\

\noindent We construct $\Gamma$ as follows: 
\begin{description}
\item[Step 1] We compute $\T_0 \cup \K[G] \cup G$ as described above, 
then purify it
in two steps, because we have a chain of two local extensions: First we 
introduce new constants $g_1, g_2$ for the terms $g(c_1), g(c_2)$, then, in the
next step, we introduce new constants $f_1, f_2, h_1, h_2$  for the terms $f(c_1), f(c_2),
h(c_1),$ and $h(c_2)$. 
We obtain:  

\

\noindent ${\sf Def} = \{ g_1 {\approx} g(c_1),  g_2 {\approx} g(c_2), f_1 {=} f(c_1),
f_2 {\approx} f(c_2), h_1 {\approx} h(c_1), h_2 {\approx} h(c_2) \}$ and 

\

\noindent {$\begin{array}{ll} 
\K_0  \cup {\sf Con}_0 \cup G_0:= \{ & 
 c_1 \leq c \rightarrow g_1 \approx f_1 , \quad c_2 \leq c \rightarrow g_2
 \approx f_2,  \\
& c < c_1 \rightarrow g_1 \approx h_1,  \quad c < c_2 \rightarrow g_2
\approx h_2, \\
& c_1 \approx c_2 \rightarrow g_1 \approx g_2,  \quad c_1 \approx c_2
\rightarrow f_1 \approx f_2, \quad  c_1 \approx c_2
\rightarrow h_1 \approx h_2, \\
& c_1 \leq c_2, \quad g_1 > g_2 \quad \} 
\end{array}$ }

\
 
\item[Step 2] The parameters are contained in the set $\Sigma_P = \{
  f, h, c \}$.  
We want to eliminate the function symbol $g$, so
  we replace $g_1, g_2$ with existentially
  quantified variables $z_1, z_2$. 

\noindent 
We obtain the existentially quantified formula $\exists z_1,
  z_2 \, G_1(c_1, c_2, c, f_1, f_2, h_1, h_2, z_1, z_2)$:

\

\noindent 
$\begin{array}{l@{}l@{}l@{}l@{}l} 
\exists z_1, z_2 (& c_1 \leq c \rightarrow z_1 \approx f_1 \quad  \wedge \quad & 
c_2 \leq c \rightarrow z_2 \approx  f_2  \quad \wedge \quad  & c_1 \approx c_2
\rightarrow f_1 \approx f_2 \quad  \wedge \quad \\
& c < c_1 \rightarrow z_1 \approx h_1 \quad \wedge \quad  & 
c < c_2 \rightarrow z_2 \approx h_2 \quad \wedge \quad &  c_1 \approx c_2 \rightarrow h_1 \approx h_2 \quad  \wedge \quad \\
&  c_1 \approx c_2 \rightarrow z_1 \approx z_2  \quad \wedge \quad  &  c_1 \leq c_2  \, \wedge \,  z_1 > z_2) 
\end{array}$ 

\

\noindent We can simplify the formula $G_1(c_1, c_2, c, f_1, f_2, h_1, h_2, z_1, z_2)$
taking into account that in the theory of (dense)
total orderings the following equivalences hold: 

\

\begin{enumerate}
\item $(c_1 \approx c_2 \rightarrow z_1 \approx z_2) \wedge z_1 > z_2 \quad \equiv
  \quad c_1 \not\approx
c_2 \wedge z_1 > z_2$ 

\item $(c_1\approx  c_2 \rightarrow f_1 \approx  f_2) \wedge (c_1 \approx c_2 \rightarrow
  h_1 \approx h_2) \wedge c_1 \not\approx c_2 \quad \equiv \quad c_1 \not\approx c_2$
\end{enumerate} 

\

\noindent We obtain the formula $\exists z_1,
  z_2 \, G'_1(c_1, c_2, c, f_1, f_2, h_1, h_2, z_1, z_2)$:  

\

\noindent 
{ $\begin{array}{l@{}l@{}l@{}l@{}l} 
\exists z_1, z_2 [& (c_1 \leq c \rightarrow z_1 \approx f_1) ~ \wedge ~ & 
(c < c_1 \rightarrow z_1 \approx h_1) ~ \wedge ~  & \\
& (c_2 \leq c \rightarrow z_2 \approx f_2)  ~ \wedge ~  & 
(c < c_2 \rightarrow z_2 \approx h_2) ~ \wedge ~ & c_1 \not\approx c_2 ~ \wedge ~   c_1 \leq c_2  ~ \wedge ~ &  z_1 > z_2]
\end{array}$ } 

\

\item[Step 3] For quantifier elimination we can use a system such as Mathematica,
Redlog or QEPCAD. 
For convenience, we illustrate how the computations can be done by 
hand. 
Note that in
  propositional logic we have:
\[ (P \rightarrow Q) \wedge (\neg P \rightarrow
  Q') \equiv (P \wedge Q) \vee (\neg P \wedge Q'). \] 
 From this it
follows that:

\

$(P \rightarrow Q) \wedge (\neg P \rightarrow 
  Q') \wedge (R \rightarrow S) \wedge (\neg R \rightarrow S') \wedge 
  W$ 

\noindent $\begin{array}{@{}lll} 

\quad & \equiv & [(P \wedge Q) \vee (\neg P \wedge Q')] \wedge  
[(R \wedge S) \vee (\neg R \wedge S')] \wedge W \\
& \equiv & 
(P \wedge Q \wedge R \wedge S \wedge W) \vee 
 (P \wedge Q \wedge \neg
  R \wedge S' \wedge W) \vee \\
& & (\neg P \wedge Q'  \wedge R \wedge S \wedge W) \vee 
 (\neg
  P \wedge Q'  \wedge \neg R \wedge S' \wedge W).
\end{array}$ 

\

\noindent Therefore, 
 the formula above is equivalent to: 

\bigskip
\noindent $\begin{array}{ll} 
\exists z_1, z_2 (& (c_1 \leq c \wedge z_1 \approx f_1  \wedge c_2 \leq c
\wedge z_2 \approx f_2 \wedge c_1 \leq c_2 \wedge z_1 > z_2 \wedge c_1 \not\approx c_2)  ~~ \vee \\
& (c_1 \leq c \wedge z_1 \approx f_1  \wedge c < c_2
\wedge z_2 \approx h_2 \wedge c_1 \leq c_2 \wedge z_1 > z_2 \wedge c_1 \not\approx
c_2)  ~~ \vee \\
&(c < c_1 \wedge z_1 \approx h_1  \wedge c_2 \leq c
\wedge z_2 \approx f_2 \wedge c_1 \leq c_2 \wedge z_1 > z_2 \wedge c_1 \not\approx
c_2)  ~~ \vee \\
&(c < c_1 \wedge z_1 \approx h_1  \wedge c < c_2
\wedge z_2 \approx h_2 \wedge c_1 \leq c_2 \wedge z_1 > z_2 \wedge c_1 \not\approx
c_2) )\\
\end{array}$ 


\

\noindent 
Using the method for quantifier elimination for dense total orderings
  without endpoints for eliminating the existentially quantified
  variables $z_1, z_2$ in this last formula we obtain the formula $\Gamma_1(c_1, c_2, c, f_1,
  f_2, h_1, h_2)$:  

\

{ $ \begin{array}{ll} 
( & (c_1 \leq c ~~\wedge~~ c_2 \leq c ~~\wedge~~ c_1 \leq c_2 ~~\wedge~~ f_1 > f_2 ~~\wedge~~ c_1 \not\approx c_2)  ~~ \vee \\[-0.2ex]
& (c_1 \leq c ~~\wedge~~ c < c_2 ~~\wedge~~ c_1 \leq c_2 ~~\wedge~~ f_1 > h_2 ~~\wedge~~ c_1 \not\approx
c_2)  ~~ \vee \\[-0.2ex]
&(c < c_1 ~~\wedge~~  c_2 \leq c
~~\wedge~~ c_1 \leq c_2 ~~\wedge~~ h_1 > f_2 ~~\wedge~~ c_1 \not\approx
c_2)  ~~ \vee \\[-0.2ex]
&(c < c_1 ~~\wedge~~ c < c_2
~~\wedge~~ c_1 \leq c_2 ~~\wedge~~ h_1 > h_2 ~~\wedge~~ c_1 \not\approx
c_2) ~~~~~~)
\end{array}$ 

 } 

\

\

\item[Step 4] We construct the formula $\Gamma_2(c_1, c_2, c)$ 
obtained from $\Gamma_1$ 
by replacing $f_i$ by $f(c_i)$ and $h_i$ by
  $h(c_i)$, $i = 1,2$. We obtain (after further minor simplification
  and rearrangement for facilitating reading): 
\end{description}

\

{
$((c_1 < c_2 \leq c \wedge f(c_1){>} f(c_2))  \vee (c_1 {\leq} c < c_2 \wedge f(c_1) {>} h(c_2))  \vee (c {<} c_1 {<} c_2 \wedge h(c_1) {>} h(c_2) ))
$}

\

\begin{description}
\item[Step 5] Then we obtain the constraint on the parameters 
$\forall z_1, z_2 \Gamma_T(z_1, z_2)$, i.e.:  

\

$\begin{array}{ll}
\forall z_1, z_2 [ & (z_1 < z_2 \leq c \rightarrow f(z_1) \leq f(z_2))
\wedge \\
& (z_1 \leq c < z_2  \rightarrow f(z_1) \leq h(z_2))
\wedge \\
& (c < z_1 < z_2  \rightarrow h(z_1) \leq h(z_2)) \quad ]
\end{array}$

\

\noindent which guarantees that $g$ is monotone. 
\end{description}
\end{exa}

\subsection{Case 2: \texorpdfstring{$\T_0$}{T\_0} does not allow quantifier elimination, but its model completion
  does}

We now  analyze the case in which $\T_0$ does not necessarily allow quantifier
elimination, but has a model completion which allows  quantifier
elimination.
\begin{thm}
Let ${\mathcal T}_0$ be a theory. Assume that ${\mathcal T}_0$ has a model completion 
${\mathcal T}_0^*$ such that ${\mathcal T}_0 \subseteq {\mathcal T}_0^*$. 
Let ${\mathcal T} = {\mathcal T}_0 \cup {\mathcal K}$ be an extension of ${\mathcal T}_0$
with new function symbols $\Sigma_1 = \Sigma_P \cup \Sigma$  
whose properties are axiomatized by a 
set of 
clauses 
${\mathcal K}$ (all of which contain symbols in $\Sigma$)
in which all variables occur also below extension functions in $\Sigma_1$.  
Assume that: 
\begin{enumerate}
\item[(i)] every model of 
${\mathcal T}_0 \cup {\mathcal K}$ embeds
into a model of ${\mathcal T}_0^* \cup {\mathcal K}$, and 
\item[(ii)] $\T_0^*$ allows quantifier elimination. 
\end{enumerate}
Then, for every finite set of ground clauses $G$ and every finite set
$T$ of ground terms over the 
signature $\Pi^C = \Pi_0 \cup \Sigma \cup \Sigma_P \cup C$ with ${\sf
  est}(\K, G)
\subseteq T$ we can construct a universally quantified 
$\Pi_0 \cup \Sigma_P$-formula 
$\forall {\overline x} \Gamma_T({\overline x})$ 
such that: 
\begin{enumerate}
\item[(1)] For every
structure ${\mathcal A}$ with signature 
$\Pi_0 \cup \Sigma \cup \Sigma_P \cup C$ 
which is a model of ${\mathcal T}_0 \cup {\mathcal K}$, if 
${\mathcal A} \models \forall {\overline x}
\Gamma_T({\overline x})$ then ${\mathcal A} \models \neg G$. 

\item[(2)] ${\mathcal T}_0 \cup \forall y \Gamma_T(y) \cup {\mathcal K} \cup
  G$ is unsatisfiable. 
\end{enumerate}
\label{QE-mc}
\end{thm}

\begin{proof}Let $G$ be a finite set of $\Pi^C$-clauses and
$T$ be a finite set of ground $\Pi^C$-terms
containing ${\sf est}({\mathcal K}, G)$.  
Since, by assumption (ii), $\T_0^*$ has quantifier elimination, 
by Theorem~\ref{inv-trans-qe} we know
that we can construct a universally quantified 
$\Pi_0 \cup \Sigma_P$-formula 
$\forall {\overline x} \Gamma_T({\overline x})$ (containing 
some parameters in $\Sigma_P$) with the following properties: 
\begin{itemize}
\item For every
structure ${\mathcal A}$ with signature 
$\Pi_0 \cup \Sigma \cup \Sigma_P \cup C$ 
which is a model of ${\mathcal T}^*_0 \cup {\mathcal K}$, if 
${\mathcal A} \models \forall {\overline x}
\Gamma_T({\overline x})$ then ${\mathcal A} \models \neg G$; 

\item ${\mathcal T}^*_0 \cup \forall y \Gamma_T(y) \cup {\mathcal K} \cup
  G$ is unsatisfiable; 
\end{itemize}
and that $\forall {\overline x} \Gamma_T({\overline x})$ is
constructed using Steps 1--5. 
We show that (1) and (2) hold. 

(1) We prove the contrapositive. Let ${\mathcal A}$ be a structure with signature  
$\Pi_0 \cup \Sigma \cup \Sigma_P \cup C$ 
which is a model of ${\mathcal T}_0 \cup {\mathcal K} \cup G$. 
As  ${\mathcal A}$ is a model of ${\mathcal T}_0 \cup {\mathcal K}$, by Assumption
(i),  ${\mathcal A}$
embeds into a model ${\mathcal B}$ of ${\mathcal T}^*_0 \cup {\mathcal K}$. 
Since $G$ is a set of ground clauses which are true in ${\mathcal A}$
and ${\mathcal A}$ embeds into ${\mathcal B}$, $G$ is also true in 
${\mathcal  B}$. Thus, ${\mathcal B}$ is a model of ${\mathcal T}^*_0 \cup {\mathcal K} \cup
G$. 
By the proof of Theorem~\ref{inv-trans-qe} and with the notation used there 
it follows that ${\mathcal B} \models \T^*_0 \cup \Gamma_2({\overline
  c}_p)$. Again, since ${\mathcal A}$ embeds into ${\mathcal B}$, and 
since $\Gamma_2({\overline c}_p)$ is a ground formula in the signature
of ${\mathcal A}$, 
${\mathcal A} \models \Gamma_2({\overline c}_p)$. It follows (as in 
the proof of Theorem~\ref{inv-trans-qe}) that ${\mathcal A} \models \exists {\overline
  y} \Gamma_2({\overline y})$. 

We showed that if ${\mathcal A} \models {\mathcal T}_0 \cup {\mathcal K} \cup G$ then 
${\mathcal A}$ is a model of $\exists x \neg \Gamma_T(x)$. 
Hence, if ${\mathcal A} \models {\mathcal T}_0 \cup \K \cup \forall  {\overline y} \Gamma_T({\overline
  y})$ then ${\mathcal A} \not\models {\mathcal T}_0 \cup {\mathcal K} \cup G$, hence 
$G$ is false in ${\mathcal A}$. \\
(2) follows directly from (1).
\end{proof}
\begin{exa} 
Consider the problem in Example~\ref{ex1-symb-elim} when the base theory $\T_0$ is 
the theory of total orderings.  We first show that conditions (i) and
(ii)  in Theorem~\ref{QE-mc} hold: 

\smallskip
\noindent $\T_0^*$ is the theory of dense total orderings without
endpoints, which allows quantifier elimination, so 
(ii) holds. 

\smallskip 
\noindent Let $\A$ be a model of $\T_0 \cup \K$, where 
\[ \begin{array}{ll} 
\K := \{ & \forall x (x \leq c \rightarrow g(x) \approx f(x)), \\
& \forall x  ( c < x \rightarrow g(x) \approx h(x)) \quad \}.
\end{array} \] 
Then $\A$ is a totally ordered set, which clearly embeds into a 
model $\B$ of $\T_0^*$ (a dense, totally ordered set without
endpoints). We can use the definitions of the functions of $f, g, h$ 
in $\A$ to define a partial $\Sigma_P \cup \Sigma$-structure on $\B$. 
Due to the form of $\K$ it is easy to see that we can extend this 
partial structure to a total model of $\T_0^* \cup \K$:
\begin{itemize}
\item the functions 
$f, h$ can be defined arbitrarily wherever they are not defined; 
\item $g$ is then defined by case distinction, such that for every $b \in
\B$, if $b \leq c$ then $g_{\B}(b) = f_{\B}(b)$ and if $b > c$ then
$g_{\B}(b) = h_{\B}(b)$. 
\end{itemize}
By Theorem~\ref{QE-mc}, the formula $\forall z_1, z_2 \Gamma_T(z_1,
z_2)$ constructed in Example~\ref{ex1-symb-elim} ensures that $g$ is monotone also in this case. 
\label{ex-dense}
\end{exa}
%

\

\noindent Unfortunately, under the assumptions of Theorem~\ref{QE-mc},
in the case of local theory extensions satisfying the conditions in 
Theorem~\ref{symb-elim-weakest} 
we cannot guarantee that the formula $\forall y \Gamma_G(y)$ is the
weakest among all universal formulae $\Gamma$ with ${\mathcal T}_0 \cup \Gamma \cup
  {\mathcal K} \cup G \models \perp$, as is illustrated by the following
  example. 
%
\begin{exa}
Let $\T_0$ be the theory of total orderings and 
\[ G:= \{ a < g(a), ~g(a) < h(a)\}.\]  
We apply Steps 1--5 for $\T_0^*$, $\K = \emptyset$ and $G$, with $T = {\sf st}(G) = \{
a, g(a), h(a) \}$, where $\Sigma_P = \{ h \}$: 

\medskip
\begin{description}
\item[Step 1] We compute $\T_0 \cup \K[G] \cup G$, then purify it. 
We obtain:  

\medskip
${\sf Def} = \{ g_1 \approx g(a), h_1 \approx h(a) \} $ 
\quad \quad \quad \quad 
$\K_0 \cup {\sf Con}_a \cup G_0 = \{ a < g_1, g_1 < h_1 \}$. 

\medskip
\item[Step 2] $\Sigma_P = \{ h \}$.  We want to eliminate $g$, so
  we replace $g_1$ with the existentially
  quantified variable $z_1$. We obtain the existentially quantified
  formula 
$\exists z_1 (a < z_1 \wedge z_1 < h_1)$. 

\medskip
\item[Step 3] Using a method for quantifier elimination for the theory
  of dense total orderings
  without endpoints for eliminating the existentially quantified
  variable $z_1$ in this formula we obtain the formula 
$a < h_1$. 

\medskip
\item[Step 4] We construct the formula $a < h(a)$ from this formula by replacing 
$h_1$ back with $h(a)$. 

\medskip
\item[Step 5] By replacing $a$ with an existentially quantified
  variable and negating we obtain 
the formula:  $\forall y \Gamma_G(y) = \forall y (h(y) \leq y).$ 
\end{description}

\medskip
\noindent We argue that this last formula is not the most general universal formula that
entails $\neg G$ (w.r.t.\ $\T_0$). 
%
%
Let $\Gamma :=  \forall x, y, z (x < y \rightarrow y \geq z)$. 
Then $\Gamma \wedge G$ is unsatisfiable w.r.t.\ $\T_0$: 
Indeed, assume that $\Gamma \wedge G$ has a model $\A$. 
Then in $\A$, $a < g(a)$ and $g(a) < h(a)$. 
\begin{itemize} 
\item As $a < g(a)$, $g(a) \geq a'$ for every $a' \in A$, so $g(a)$ is a
maximal element of $\A$. 
\item But then $g(a) \geq h(a)$. This contradicts the fact that, in
  $\A$, $g(a)  < h(a)$. 
\end{itemize}
This shows that $\Gamma \wedge G$ is unsatisfiable w.r.t.\ $\T_0$.

\noindent However, there exists a structure $\A_1$ with two elements
$a_1, a_2$ where $a_1 < a_2$ 
such that $h_{\A_1}(a_1) = a_2$ which satisfies $\Gamma$ but not $\Gamma_G$:  
\begin{enumerate}
\item This structure clearly satisfies $\Gamma$: for every valuation $\beta
: \{ x, y, z \} \rightarrow \{ a_1,  a_2 \}$ we have the following
situations: 
\begin{itemize}
\item $\beta(x) \geq \beta(y)$: Then $\A_1, \beta \models (x <
  y \rightarrow y \geq z)$ since the premise is false. 
\item $\beta(x)  < \beta(y)$: Then $\beta(x) = a_1, \beta(y) = a_2$,
  so $\beta(y) \geq \beta(z)$ no matter what the value of $\beta(z)$
  is. Thus, also in this case $\A_1, \beta \models (x <
  y \rightarrow y \geq z)$. 
\end{itemize}
\item The structure does not satisfy $\forall y \Gamma_G(y) := \forall y
  (h(y) \leq y)$: for the valuation with $\beta(y) = a_1$ we have  
$a_2 = h_{\A_1}(a_1)  > a_1$, so $\A_1, \beta \not\models \forall y
  (h(y) \leq y)$. 
\end{enumerate}

\noindent 
Note that this situation cannot occur when $\T_0$ has quantifier
elimination: Then the formula $\exists {\overline x} G_1({\overline x})$ 
is either true or false in $\T_0$. If it is true then
to achieve unsatisfiability we have to add $\Gamma = \perp$, 
which entails any other constraint. 
If it is false then we do not need to add any constraints to achieve 
unsatisfiability, so $\Gamma = \top$, which is entailed by any other
constraint. 
\end{exa}

\section{Ground Interpolation in Theory Extensions}
\label{interp-loc}

In this section we present criteria for recognizing whether a 
theory extension $\T = \T_0 \cup \K$ has ground interpolation provided 
that $\T_0$ has (general) ground interpolation. 
 In particular, we are interested in giving criteria for 
 checking whether a theory $\T$ (resp.\ a theory extension 
 $\T = \T_0 \cup \K$) has a special form of the ground 
 interpolation property, in which 
 for every pair of ground formulae $A, B$ with $A \wedge B \models_{\T} \perp$ 
 there exists an interpolant $I$ of $A$ and $B$  such that the terms
 (resp.\ the extension terms) 
occurring in $I$ are in a set $W(A, B)$ which can be 
 constructed from the set of ground terms of $A$ and $B$. 
 
\subsection{Previous Work} 
In \cite{Sofronie-lmcs} we identified classes of local extensions in 
which ground interpolants can be computed hierarchically
(for this, we had to find ways of  separating 
the instances of axioms in $\K$ and of the congruence axioms).
We present the ideas below. 

Let $\T_0 \subseteq \T = \T_0 \cup \K$ be a local theory extension
with function symbols in a set $\Sigma_1$ and let
$A({\overline a}, {\overline c}), B({\overline b}, {\overline c})$ 
be sets of ground clauses over the signature of $\T$, possibly
containing additional constants in a set $C$, such that 
$A \wedge B \models_{\T_0  \cup \K} \perp$. 
From Theorem~\ref{lemma-rel-transl} we know that in such extensions 
hierarchical reasoning is possible: 
if $A$ and $B$ are sets of ground clauses in a signature $\Pi^C$, and 
$A_0 \wedge D_A$ (resp.\ $B_0 \wedge D_B$) are 
obtained from $A$ (resp.\ $B$) by purification and flattening -- where
${\sf Def} = D_A \cup D_B$ the union of the set $D_A$ 
containing those equalities $c_t \approx t$, where $t$ is an extension
term in $A$ and the set $D_B$ containing those equalities $c_t \approx
t$, 
where $t$ is an extension
term in $B$ --  then the following are equivalent: 
\begin{itemize}
\item $A \wedge B \models_{{\T}_1} \perp$; 
\item ${\T}_0 \wedge \K[A \wedge B] \wedge (A \wedge B) \models \perp$; 
\item ${\T}_0 \wedge \K[A \wedge B] \wedge (A_0 \wedge D_A) \wedge (B_0 \wedge D_B) \models \perp$;  
\item $\K_0 \wedge 
A_0 \wedge B_0 \wedge {\sf Con}[D_A \wedge D_B]_0 \models_{{\T}_0} \perp,$   
\end{itemize}
where ${\K}_0$ is obtained from 
${\K}[A \wedge B]$ by replacing the $\Sigma_1$-terms with the 
corresponding constants contained in the definitions $D_A$ and  $D_B$ and 
 
\noindent $ {\sf Con}[D_A \wedge D_B]_0 = 
{\sf Con}_0 = \displaystyle{ \bigwedge  
\{ \bigwedge_{i = 1}^n} c_i \approx d_i   \rightarrow c \approx d 
\mid \begin{array}{l} 
f(c_1, \dots, c_n) \approx c \in {\sf Def} = D_A \cup D_B \\
f(d_1, \dots, d_n) \approx d \in {\sf Def} = D_A \cup D_B 
\end{array} \}.$

\noindent In general the method for hierarchical reasoning in local theory 
extensions is not sufficient for computing hierarchically interpolants 
because:  
\begin{enumerate}
\item[(i)] ${\mathcal K}[A \wedge B]$ may contain free variables. 

\item[(ii)] If some clause in ${\mathcal K}$ contains two or 
more different extension functions, 
these extension functions cannot always be ``separated''. 

\item[(iii)]  The clauses ${\mathcal K}_0 \wedge {\sf Con}[D_A \wedge D_B]_0$
may contain combinations of constants and extension functions from $A$ and $B$. 
\end{enumerate}
In order to avoid problem (i), in \cite{Sofronie-lmcs} we considered only
extensions with sets of clauses ${\mathcal K}$ 
of clauses in which all variables occur below some extension term.
To address (ii), we defined an equivalence relation 
$\sim$ between extension functions, where 
$f \sim g$ if $f$ and $g$ occur in the same clause in ${\mathcal K}$, 
and considered that a function $f \in \Sigma_1$ is 
common to  $A$ and $B$ if there exist $g, h \in \Sigma_1$ 
such that $f \sim g$, $f \sim h$, $g$ occurs in $A$ and $h$ occurs in $B$.

\medskip
\noindent In order to address (iii), we identified situations in which it is
possible to separate mixed instances of 
axioms in ${\mathcal K}_0$, or of congruence axioms in 
${\sf Con}[D_A \wedge D_B]_0$, into 
an $A$-part and a $B$-part. 
We illustrate this on an example discussed in \cite{Sofronie-lmcs}. 

\begin{exa}[\cite{Sofronie-lmcs}] 
\label{ex-idea-interp}
Consider the conjunction $A_0 \wedge D_A \wedge B_0 \wedge D_B \wedge 
{\sf Con}[D_A \wedge D_B]_0 \wedge {\sf Mon}_0 \wedge {\sf SGc}_0$  in 
Example~\ref{example-hierarchic}, where $\T_0 = {\sf SLat}$. The $A$ and $B$-part share the constants 
$c$ and $d$, and no function symbols. However, as $f$ and $g$ occur together 
in ${\sf SGc}$, $f \sim g$, 
so they are considered to be all shared. (Thus, the interpolant 
is allowed to contain both $f$ and $g$.)  
We can obtain a separation for the clause 
$b \leq a_1 \rightarrow b_1 \leq a$ of ${\sf SGc}_0$ as follows:
\begin{enumerate}
\item[(i)] We note that $A_0 \wedge B_0 \models b \leq a_1$. 
\item[(ii)] We can find an ${\sf SLat}$-term $t$ containing only shared 
constants of 
$A_0$ and $B_0$ such that $A_0 \wedge B_0 \models b \leq t \wedge t \leq a_1$.
(Indeed, such a term is $t = d$.) 
\item[(iii)] We show that, instead of the axiom 
$b \leq g(a) \rightarrow f(b) \leq a$, 
whose flattened form is in ${\sf SGc}_0$, we can use, without loss 
of unsatisfiability:

\medskip

\begin{enumerate}
\item[(1)] an instance of the monotonicity axiom for $f$: 
$b \leq d \rightarrow f(b) \leq f(d)$, 

\item[(2)] another instance of ${\sf SGc}$, namely: 
$d \leq g(a) \rightarrow f(d) \leq a$. 
\end{enumerate}

\medskip

\noindent For this, we introduce a new constant $c_{f(d)}$ for $f(d)$
(its definition, $c_{f(d)} \approx f(d)$, is stored in a set $D_T$), 
and 
the corresponding instances ${\mathcal H}_{\sf sep} = {\mathcal H}^{A}_{\sf sep} 
\wedge {\mathcal H}^{B}_{\sf sep}$ 
of the congru\-ence, monotonicity and 
${\sf SGc}(f, g)$-axioms, which are now  
separated into an $A$-part 
(${\mathcal H}^{A}_{\sf sep}: d \leq a_1 \rightarrow c_{f(d)} \leq a$) and a 
$B$-part (${\mathcal H}^{B}_{\sf sep}: b \leq d \rightarrow b_1 \leq c_{f(d)}$).
We thus obtain a separated conjunction
${\overline A}_0 \wedge {\overline B}_0$  (where 
${\overline A}_0 =  {\mathcal H}^{A}_{\sf sep} \wedge A_0$ and  
${\overline B}_0 = {\mathcal H}^{B}_{\sf sep} \wedge B_0$),  
which can be proved to be unsatisfiable in 
${\mathcal T}_0 = {\sf SLat}$. 
\item[(iv)] To compute an interpolant in ${\sf SLat}$ for 
${\overline A}_0 \wedge {\overline B}_0$ 
note that 
${\overline A}_0$ is logically equivalent to the conjunction of unit 
literals 
\/ $d \leq a_1 ~\wedge~ a \leq c ~\wedge~ c_{f(d)} \leq a$
and ${\overline B}_0$ is logically equivalent to 
\/ $b \leq d ~\wedge~ b_1 {\not\leq} c ~\wedge~ b_1 \leq c_{f(d)}$. 
An interpolant  is 
$I_0 = c_{f(d)} \leq c$. 
\item[(v)] By replacing the new constants with the 
terms they denote we obtain the interpolant 
$I = f(d) \leq c$ for $A \wedge B$. 
\end{enumerate}

\noindent The same ideas can be used when $\T_0 = {\sf TOrd}$. 
\end{exa}

\noindent 
Criteria linking hierarchical ground interpolation to a notion of
``separability'' and to an amalgamability 
property for partial algebras were given in \cite{Wies, Wies-journal}. 
We present the ideas in \cite{Wies, Wies-journal} and then extend some of the
results presented there. 

\subsection{W-Separability}

\begin{defi}[\cite{Wies}] 
An amalgamation closure for a theory 
extension $\T = \T_0 \cup  \K$ is a function $W$ 
associating with finite 
sets of ground terms $T_A$ and $T_B$, a finite set 
$W(T_A, T_B)$ of ground terms such that 
\begin{enumerate-}
\item all ground subterms in $\K$ and $T_A$
are in $W(T_A, T_B)$; 
\item $W$ is monotone, i.e., for all $T_A \subseteq T'_A$,
$T_B \subseteq T'_B$, $W(T_A, T_B) \subseteq W(T'_A, T'_B)$; 
\item $W$ is a closure,
i.e., $W(W(T_A, T_B),W(T_B, T_A)) \subseteq  W(T_A, T_B)$; 
\item
$W$ is compatible with any map $h$ between constants satisfying 
$h(c_1) \neq h(c_2)$, for all constants $c_1 \in {\sf st}(T_A), c_2 
\in {\sf st}(T_B)$ that 
are not shared between $T_A$ and $T_B$, i.e., for any such $h$ we require 
$W(h(T_A), h(T_B)) = h(W(T_A, T_B))$; and 
\item $W(T_A, T_B)$  only contains $T_A$-pure terms (i.e.\ terms
  containing only constants in $C$ which occur in $T_A$). 
\end{enumerate-}
\end{defi}
For sets of ground clauses $A, B$ we write 
$W(A, B)$ for $W({\sf st}(A), {\sf st}(B))$. 
In what follows,  when we use a binary function $W$ 
we always refer to an amalgamation closure.

\begin{defi}[\cite{Wies}] 
A theory extension $\T = \T_0 \cup  \K$  is 
$W$-separable if for all sets of ground clauses $A$ and $B$, 
\[ T_0 \cup \K \cup A \cup B \models \perp \text{ iff }  T_0 \cup  \K[W(A,B)] \cup  A \cup  \K[W(B,A)] \cup B \models 
\perp. \] 
\end{defi}

\begin{exa}
Let $\T_0$  be the theory ${\sf TOrd}$ of total orderings\footnote{We
 chose here $\T_0$ to be the theory of total orderings in order to simplify the
  example: The signature of ${\sf TOrd}$ does not contain function
  symbols, so the amalgamation closure $W$ is easier to describe.}. We consider
the extension of $\T_0$ with function symbols $f, g$ satisfying the
axioms $\K = \{ {\sf SGC}(f, g), {\sf Mon}(f, g) \}$ discussed 
in Examples~\ref{example-hierarchic} and
\ref{ex-idea-interp} (cf.\ also \cite{Sofronie-lmcs}), where: 

\begin{itemize}
\item ${\sf {\sf SGC}}(f, g): \forall x, y (x \leq g(y) \rightarrow f(x) \leq y)$; 
\item ${\sf Mon}(f, g):  \forall x, y (x \leq y \rightarrow f(x) \leq
  f(y)) \wedge \forall x, y (x \leq y \rightarrow  g(x) \leq
  g(y))$. 
\end{itemize}

\noindent The theory of total orderings is $\leq$-interpolating
(for details cf.\ \cite{Sofronie-lmcs}): If $A_0$ and $B_0$ are sets of ground clauses 
in the signature of ${\sf TOrd}$ and $A_0 \wedge B_0 \models_{\sf
  TOrd} a \leq b$, where $a$ is a constant in $A_0$ and $b$ a constant
in $B_0$ then there exists a constant $d$ (common to $A_0$ and $B_0$) 
such that $A_0 \wedge B_0 \models_{\sf
  TOrd} a \leq d \wedge d \leq b$. 
Thus, the terms needed for $\leq$-interpolation are the common
constants of $A$ and $B$. 

\medskip
\noindent If $C_A$ ($C_B$) are the constants in $A$
($B$) then,   
from the form of the clauses in $\K$ and the results in
\cite{Sofronie-lmcs}
we can show that $\T_0 \cup \K$ is $W$-separable where $W(A, B) = {\sf
  st}(A) \cup \{ f(c), g(c) 
\mid c \in C_A \cap C_B \}$ and $W(B, A) = {\sf st}(B) \cup \{ f(c),
g(c) 
\mid c \in C_A \cap C_B \}$. 

In fact, the results in \cite{Sofronie-lmcs} show that if $A \wedge B
\models \perp$, $W$ can be defined more precisely as 
$W(A, B) = {\sf
  st}(A) \cup \{ f(c), g(c) 
\mid c \in D_{AB} \}$ and $W(B, A) = {\sf st}(B) \cup \{ f(c), g(c) 
\mid c \in D_{AB} \}$, where $D_{AB}$ is the set of constants
common to $A$ and $B$ which can be used for  $\leq$-interpolation.

\label{ex1-interp}
\end{exa}

\subsection{\texorpdfstring{$W$}{W}-Separability and Partial \texorpdfstring{$W$}{W}-Amalgamation} 
In \cite{Wies} it is shown that 
if $\T = \T_0 \cup  \K$ is $W$-separable, and $\K$ 
is flat and linear, then the extension $\T_0 \subseteq \T = \T_0 \cup \K$ is $\Psi$-local 
where $\Psi(T) = W(T, T)$ 
for all sets of ground 
terms $T$. Then, a notion of partial $W$-amalgamability is defined as
follows: 
 \begin{defi}[\cite{Wies}]
A theory extension $\T_0 \subseteq \T = \T_0 \cup \K$ is said to
have the partial amalgamation property
with respect to amalgamation closure $W$ (for short: the partial
$W$-amalgamation property) 
if whenever 
$M_A, M_B, M_C \in {\sf PMod}_{w,f}(\Sigma, \T)$ are such that: 
\begin{enumerate-}
\item $M_C$ is a substructure of $M_A$ and of $M_B$, i.e.\
the universe $|M_C|$ of $M_C$ is included in the universes of $M_A$ and $M_B$
and the inclusions into $M_A, M_B$ are weak embeddings; 
\item $|M_C| = |M_A| \cap |M_B|$; 
\item the sets $T_{M_A} = \{ f(a_1, \dots, a_n) \mid a_1, \dots, a_n \in
  M_A, f_{M_A}(a_1, \dots, a_n) \text{ defined} \}$ and $T_{M_B} = \{ f(a_1, \dots, a_n) \mid a_1, \dots, a_n \in
  M_B, f_{M_B}(a_1, \dots, a_n) \text{ defined}  \}$ of terms which are defined in $M_A$ resp.\ $M_B$ are closed
  under the operator $W$, i.e.\ $W(T_{M_A}, T_{M_B}) \subseteq T_{M_A}$ and
  $W(T_{M_B}, T_{M_A}) \subseteq T_{M_B}$; 
\item $T(M_A) \cap T(M_B) \subseteq T(M_C)$; 
\end{enumerate-}
there exists a model $M_D$ of $\T$, and weak embeddings $h_A : M_A
\rightarrow M_D$ 
and $h_B : M_B \rightarrow M_D$,
such that ${h_A}_{|_{|M_C|}} = {h_B}_{|_{|M_C|}}$.
\label{def:part-W-amalg}
\end{defi}

\noindent It is shown that if  
$\T_0 \subseteq \T = \T_0 \cup  \K$ is a theory extension with $\K$ 
flat and linear and $\T_1$ has the 
partial $W$-amalgamation property  w.r.t.\ $W$, 
then $\T_1$ is $W$-separable. 

\medskip

\noindent 
We  make this last result more precise by showing that: 
\begin{itemize}
\item In order to obtain a criterion for $W$-separability we only need a
  weak version of partial $W$-amalgamability, namely partial
  $W$-amalgamability for partial structures with the same $\Pi_0$-reduct
  (Definition~\ref{def:part-amalg}). 
\item   We then prove that also the converse holds, i.e.\ that for
  extensions $\T = \T_0 \cup \K$ of a first-order theory $\T_0$ if 
  (i) the extension is W-separable and 
  (ii) $\T_0$ allows general ground interpolation 
  then $\T$ has the partial
  $W$-amalgamability property 
  for partial structures with the same $\Pi_0$-reduct. 
  This implication was not studied in \cite{Wies}.
\end{itemize}
We will then show that in general  partial $W$-amalgamability
implies partial  $W$-amalgamability 
for partial structures with the same $\Pi_0$-reduct, and that 
the converse implication holds under the additional assumption 
that $\T_0$ allows general ground interpolation.


\begin{defi} 
Let $W$ be an amalgamation closure operator. 
A theory extension $\T = \T_0 \cup \K$ has the {\em partial $W$-amalgamation
property for models with the same $\Pi_0$-reduct} if whenever 
$M_A, M_B \in {\sf PMod}_{w,f}(\Sigma, \T)$ are such that: 
\begin{enumerate-}
\item $M_A, M_B$  have the same reduct $M$ to $\Pi_0$; 
\item For all $m_1, \dots, m_n \in |M| = |M_A| = |M_B|$ if 
$f_{M_A}(m_1, \dots, m_n)$ is defined and is equal to $m$ and  
$f_{M_B}(m_1, \dots, m_n)$ is defined and is equal to $m'$ then $m =
m'$; 
\item The sets $T_{M_A} = \{ f(a_1, \dots, a_n) \mid a_1, \dots, a_n \in
  M_A, f_{M_A}(a_1, \dots, a_n) \text{ defined} \}$ and $T_{M_B} = \{ f(a_1, \dots, a_n) \mid a_1, \dots, a_n \in
  M_B, f_{M_B}(a_1, \dots, a_n) \text{ defined}  \}$ of terms which are defined in $M_A$ resp.\ $M_B$ are closed
  under the operator $W$, i.e.\ $W(T_{M_A}, T_{M_B}) \subseteq T_{M_A}$ and
  $W(T_{M_B}, T_{M_A}) \subseteq T_{M_B}$; 
\end{enumerate-}
there exists a model $M_D$ of $\T_0 \cup \K$ and weak embeddings $h_A:
M_A \rightarrow M_D, h_B : M_B \rightarrow M_D$ which agree on
$M$ (and thus coincide as functions).\footnote{The partial $W$-amalgamation
property for models with the same $\Pi_0$-reduct can also be regarded
as an embeddability property for partial algebras in which the set of
terms which are defined can be seen as the union of two sets $T_{M_A}$
and $T_{M_B}$ which are closed under the application of $W$.  
We decided to use the notion ``partial $W$-amalgamation
property for models with the same $\Pi_0$-reduct'' in this paper 
for the sake of consistency with the terminology introduced in
\cite{sofronie-ijcar-2016}.} 
\label{def:part-amalg} 
\end{defi}
\begin{defi}[\cite{Wies}]
Let $M$ be a model of $\T_0 \cup \K$ and 
let $T$ be a finite set of ground terms such that ${\sf st}(\K)
\subseteq T$. 
We assume that the terms in $T$ are flat (cf.\ Definition~\ref{flat})
or quasi-flat (i.e.\ for all terms of the form $f(t_1, \dots, t_n)$  of $T$, 
where $f$ is an extension function, $t_1, \dots, t_n$ are constants or 
ground terms over $\Pi_0^C$). 

We denote by $M_{|T}$  the partial structure which has the same
support as $M$, and in which all symbols in $\Pi_0$ are defined as in
$M$, but in which the interpretation
of extension symbols $f \in \Sigma$ is restricted as follows: 
For all elements $a_1, \dots , a_n  \in |M|$   if there
exist terms $t_1, \dots , t_n$ such that the interpretation $(t_i)_M$ of $t_i$
in $M$ is $a_i$ for all $i \in \{ 1, \dots, n \}$, and 
$f (t_1, \dots , t_n) \in T$, then
$f_{M_{|T}}(a_1, \dots , a_n) = f_M(a_1, . . . , a_n)$; otherwise 
$f_{M_{|T}}(a_1, \dots , a_n)$ is undefined. 

%
\end{defi}
\begin{lem}
  \label{lem:lemma-68}
Let $M$ be a model of $\T = \T_0 \cup \K$ and $T$ be a finite set of
ground quasi-flat $\Pi$-terms (closed under subterms) with ${\sf
  st}({\mathcal K}) \subseteq T$.
Then the following hold: 
\begin{enumerate-}
\item $M_{|T} \in {\sf PMod}_{w,f}(\Sigma,\T)$. 
\item 
$M_{|T}$ is a partial model of ${\mathcal K}[T]$ in which all
  extension terms  
  in ${\mathcal K}[T]$ are defined. 
\end{enumerate-} 
\label{lem:restriction-T-tot}
\end{lem}
\begin{proof}
(1) The proof proceeds as in \cite{Wies}. 
Clearly, $M_{|T}$ is a total model of $\T_0$. Let $C \in \K$ and
$\beta : X \rightarrow M_{|T}$. 
If some of the terms in 
$\beta(C)$ are undefined in $M_{|T}$ then by definition $M_{|T}, \beta
\models_w C$. 
Assume now that all terms in $\beta(C)$ are defined in $M_{|T}$. 
By the definition of $M_{|T}$ they have the same values as in $M$. 
Therefore, as $M, \beta \models C$ it follows that $M_{|T}, \beta
\models_w C$
also in this case.

(2) The only extension terms occurring in $\K[T]$ are those in $T$, and these are defined
in $M_{|T}$. Let $D \in \K[T]$ and $\beta : X \rightarrow
M_{|T}$. Then $D = C\sigma$ where $C \in \K$ and $\sigma$ is a 
substitution such that for every term $t$ occurring in $C$ which starts with a function
symbol in $\Sigma$,  $t\sigma \in T$. 
Thus, in $D$ all terms starting with an extension function $f \in \Sigma$ are ground
terms $f(t_1, \dots t_n) \in T$. If $D$ contains variables, these do
not occur below extension functions. Thus, $\beta(D)$ is defined in
$M_{|T}$. 
Let $\gamma : X \rightarrow M$ be 
$\gamma(x) = \beta(\sigma(x))$ for every $x \in X$. 
Then $\beta(D) = \beta(\sigma(C)) = \gamma(C)$. In $\beta(D)$, the
$\Sigma$-terms are in $T$, hence they are defined in $M_{|T}$ and have
the same value as in $M$.  
Since $M \models \K$,  $M, \gamma \models C$, so there exists at least one literal in 
$C$ which is true in $M$ w.r.t.\ $\gamma$, thus there exists at least
one literal in 
$D$ which is true in $M_{|T}$ w.r.t.\ $\beta$.
\end{proof}

\begin{rem} If we impose, in addition, that all variables
in $\K$ occur below an extension function, then (2) is immediate:
$\K[T]$ is ground and contains only $\Sigma$-terms in $T$; those 
terms are all defined in $M_{|T}$ and their value is the same as in
$M$. (We do not need to consider variable assignments in that case.)
\end{rem}

\begin{lem}
Let $\T_0 \cup \K$ be an extension of $\T_0$ with a set $\K$ of flat
and linear clauses. 
Let $T$ be a finite set of
ground flat $\Pi$-terms (closed under subterms) with ${\sf
  st}({\mathcal K}) \subseteq T$, and let 
$M$ be a model of $\T_0 \cup \K[T]$. Then $M_{|T} \in {\sf PMod}_{w,f}(\Sigma,\T)$. 
\label{lem:restriction-T}
\end{lem}
\begin{proof}
The proof is similar to the beginning of the proof of Theorem~2 in
\cite{sofronie-cade-05} (cf.\ also \cite{Wies}). 
Clearly, $M_{|T}$ is a total model of $\T_0$. 
To show that $M_{|T} \models_w \K$ we use the fact that if $D$ is a
clause in $\K$ and 
$\beta : X \rightarrow M_{|T}$ 
is an assignment in which $\beta(t)$ is defined for every term $t$
occurring in $D$, then $D$ is true in $M_{|T}$ w.r.t.\ $\beta$. 

For this, note that if for 
every term $t = f(x_1, \dots, x_n)$ of $D$, $\beta(t) = f_{M_{|T}}(\beta(x_1),
\dots, \beta(x_n))$ is defined in $M_{|T}$ then there exist terms $t_1,
\dots, t_n$ such that $(t_i)_{M} = \beta(x_i)$ and $f(t_1, \dots, t_n)
\in T$. Let $\sigma$ be the  substitution with $\sigma(x_i) = t_i$
($\sigma$ is well-defined because $D$ is flat and linear). 
Then $\sigma(D) \in \K[T]$ and 
$\beta(\sigma(t)) = \beta(t)$. Since  $M$ be a model of $\T_0 \cup
\K[T]$, $\sigma(D)$ is true in $M$ w.r.t.\ $\beta$, thus (as
$\beta(\sigma(t)) = \beta(t)$ for every term in $D$) 
$D$ is true in $M_{|T}$ w.r.t.\ $\beta$.
\end{proof}
%


\begin{thm}
Let $W$ be an amalgamation closure operator with the additional
property that if $T_1$ and $T_2$ are sets of flat ground terms,
$W(T_1, T_2)$ is a set of flat ground terms\footnote{The result
  holds also if $W(T_1, T_2)$ is a set of quasi-flat ground terms
  whenever $T_1, T_2$ are sets of quasi-flat ground terms, but
  condition (4) in the definition of an amalgamation closure needs then
  to be adapted.}. 
Assume that $\T_0$ is a first-order theory and let $\K$ be a set of
flat and linear 
clauses over $\Pi_0 \cup \Sigma$. 
If $\T_0 \cup \K$ has the partial amalgamation property 
for models with the same $\Pi_0$-reduct then $\T_0 \cup \K$ is
$W$-separable. 
\label{thm:p-amalg-impl-Wsep} 
\end{thm}

\begin{proof}
The proof proceeds like the proof in \cite{Wies} (but we reformulated some hypotheses). 
Assume that $\T_0 \cup \K$ has the partial $W$-amalgamation property 
for models with the same $\Pi_0$-reduct. 
Let $A$ and $B$ be sets of ground clauses over $\Pi_0 \cup \Sigma \cup
C$. Without loss of generality we assume that 
$A$ and $B$ are flat sets of ground clauses, 
thus ${\sf st}(A)$ and ${\sf  st}(B)$ consist of flat terms only. 
We show that 
$\T_0 \cup \K \cup A \cup B$ is unsatisfiable 
iff $\T_0  \cup  (\K[W(A,B)] \cup A) \cup  (\K[W(B,A)] \cup B)$ 
is unsatisfiable.

\medskip

\noindent 
The converse implication is obvious. We prove the direct implication. 
Assume that $\T_0 \cup \K \cup A \cup B$ is unsatisfiable  but 
$\T_0  \cup  (\K[W(A,B)] \cup A) \cup  (\K[W(B,A)] \cup B)$ 
has a model $M$. 
We define $M_A := M_{|W(A,B)}$, $M_B := M_{|W(B,A)}$. 
As 
${\sf st}(\K) \subseteq W(A, B) \cap W(B, A)$, and $\K$ is flat and
linear we know, by Lemma~\ref{lem:lemma-68} and Lemma~\ref{lem:restriction-T}, that: 
\begin{enumerate}
\item[(i)] $M_A, M_B \in {\sf
  PMod}_{w,f}(\Sigma, \T)$. 
\item[(ii)] $M_A$ is a model of $\K[W(A,B)] \cup A$ and 
$M_B$ is a model of 
$\K[W(B,A)] \cup B$, and 
\item[(iii)] all terms in $W(A, B)$ and $A$ are defined in
$M_A$, and \\
all terms in $W(B, A)$ and $B$ are defined in $M_B$. 
\end{enumerate}

\medskip

\noindent The models $M_A$ and $M_B$ 
satisfy the conditions in Definition~\ref{def:part-amalg}:
\begin{enumerate}
\item they clearly have
the same reduct to $\Pi_0$ (namely $M_{|{\Pi_0}}$), 
\item $f_{M_A}(m_1, \dots, m_n)$ is defined and equal to $m$ iff
  there exists ground  $\Pi_0$-terms $t_1, \dots, t_n$ with 
  $f(t_1, \dots, t_n) \in W(A,B)$, $(t_i)_M = m_i$ for all $i = 1,
  \dots, n$, and $f_M(m_1, \dots, m_n) = m$; \\
$f_{M_B}(m_1, \dots, m_n)$ is defined and equal to $m'$ iff
  there exists ground  $\Pi_0$-terms $t'_1, \dots, t'_n$ with 
  $f(t'_1, \dots, t'_n) \in W(B,A)$, $(t'_i)_M = m_i$ for all $i = 1,
  \dots, n$ and $f_M(m_1, \dots, m_n) = m'$. 

Thus, if $f_{M_A}(m_1, \dots, m_n) = m$ and $f_{M_B}(m_1, \dots, m_n)
= m'$ then $f_M(m_1, \dots, m_n) = m = m'$, so $m = m'$.

\item $W(T_{M_A}, T_{M_B}) \subseteq T_{M_A}$ and $W(T_{M_B}, T_{M_A})
  \subseteq T_{M_B}$. Indeed, let $h$ be the map that maps all ground
  $\Pi_0$-terms occurring in $W(A, B)$ resp.\ $W(B, A)$ (which are
  constants if $W(A, B)$ and $W(B, A)$ are flat) to the value
  of these terms in $M$. Then:
\begin{itemize}
\item $W(T_{M_A}, T_{M_B}) = W(h(W(A, B)),
  h(W(B, A))) \subseteq h(W(A, B)) = T_{M_A}$ and 
\item $W(T_{M_B}, T_{M_A}) = W(h(W(B, A)),
  h(W(A, B))) \subseteq h(W(B, A)) = T_{M_B}$. 
\end{itemize}
\end{enumerate}
Since $\T_0 \cup \K$ has the partial amalgamation property 
for models with the same $\Pi_0$-reduct, there exists a model 
$M_D$ of $\T_0 \cup \K$ and  weak embeddings $h_A:
M_A \rightarrow M_D, h_B : M_B \rightarrow M_D$ which agree on
$M_{|\Pi_0}$. 

Clearly,  weak embeddings of a partial into a total algebra 
preserve the truth of ground formulae which are defined in the partial
algebra. As all terms in $A$ are defined in $M_A$ and $A$ is true in
$M_A$, and as all terms in $B$ are defined in $M_B$ and $B$ is true in $M_B$, 
it follows therefore that 
$M_D$ is a model of $A$ and $B$, hence a model of 
$\T_0 \cup \K \cup A \cup B$. Contradiction. 

It follows that $\T_0
\cup  (\K[W(A,B)] \cup A) \cup  (\K[W(B,A)] \cup B)$ is satisfiable
iff $T_0 \cup \K \cup A \cup B$ is satisfiable.
\end{proof}

\begin{exa}
In \cite{Wies} it was proved that the theory of arrays with difference 
function and the theory of linked lists with reachability 
have partial amalgamation, hence are 
$W$-separable (for suitable versions of $W$, described in \cite{Wies}). 
\end{exa}

\medskip
\noindent We now prove a converse of Theorem~\ref{thm:p-amalg-impl-Wsep}. 
\begin{thm}
Let $W$ be an amalgamation closure operator. 
Assume that $\T_0$ is a first-order theory which allows general ground
interpolation and has the property that for every set $\Sigma$ of
additional function symbols and ground $\Sigma _0\cup \Sigma$-formulae
$A, B$, the interpolant $I$ contains only $\Sigma$-terms in $W(A, B) \cap W(B, A)$. 
Let $\K$ be a set of
clauses over $\Pi_0 \cup
\Sigma$, such that all variables occur below an extension symbol. 

If $\T = \T_0 \cup \K$ is
$W$-separable then it has the partial $W$-amalgamation property 
for models with the same $\Pi_0$-reduct. 
\label{thm:Wsep+intW-impl-pamalg}
\end{thm}  
\begin{proof}Assume that $\T_0$ satisfies the assumptions of this
theorem and $\T = \T_0 \cup \K$ is $W$-separable. 
We show that $\T$ has the partial $W$-amalgamation property 
for models with the same $\Pi_0$-reduct.
Let $M_A, M_B \in {\sf PMod}_{w, f}(\Sigma, \T)$ 
be such that: 
\begin{enumerate-}
\item $M_A, M_B$ 
have the same reduct $M$ to $\Pi_0$. 
\item For all $m_1, \dots, m_n \in |M|$, if $f_{M_A}(m_1, \dots, m_n)$
  and $f_{M_B}(m_1, \dots, m_n)$ are defined and $f_{M_A}(m_1, \dots,
  m_n) = m$ and  $f_{M_B}(m_1, \dots, m_n) = m'$ then $m = m'$.
\item The sets $T_{M_A} = \{ f(a_1, \dots, a_n) \mid a_1, \dots, a_n \in
  M_A, f_{M_A}(a_1, \dots, a_n) \text{ defined} \}$ and $T_{M_B} = \{ f(a_1, \dots, a_n) \mid a_1, \dots, a_n \in
  M_B, f_{M_B}(a_1, \dots, a_n) \text{ defined}  \}$ of terms which are defined in $M_A$ resp.\ $M_B$ are closed
  under the closure operator $W$, i.e.\ $W(T_{M_A}, T_{M_B}) \subseteq T_{M_A}$ and
  $W(T_{M_B}, T_{M_A}) \subseteq T_{M_B}$. 
\end{enumerate-}
In order to show that there exists a model $M_D$ of $\T_0 \cup \K$ and weak embeddings $h_A:
M_A \rightarrow M_D, h_B : M_B \rightarrow M_D$ which agree on $M$ 
we show that $\T_0 \cup \K \cup \D_A \cup \D_B$
is satisfiable, where $\D_A$ is defined by: 
\begin{eqnarray*}
\D_A & = & \{ f(a_1, \dots a_n) \approx a \mid \text{ if } 
f_{M_A}(a_1, \dots, a_n) \text{ is defined and equal to } a \} \\
          & & \cup \{ f(a_1, \dots a_n) \not\approx a \mid \text{ if } 
f_{M_A}(a_1, \dots, a_n) \text{ is defined and not equal to } a \} \\
          & & \cup \{ P(a_1, \dots, a_n) \mid P \in {\sf Pred} \text{ and } (a_1, \dots, a_n) \in P_{M_A} \} \\
          & & \cup \{ \neg P(a_1, \dots, a_n) \mid P \in {\sf Pred} 
\text{ and } (a_1, \dots, a_n) \not\in P_{M_A} \} \\
          & & \cup \{ a \not\approx a' \mid  a, a' \in |M_A| \text{
            different } \} 
\end{eqnarray*}
and $\D_B$ is defined analogously (the elements in the universe
of $M$ are used as additional constants).

Assume $\T_0 \cup \K \cup \D_A \cup \D_B$ is
unsatisfiable. Then, by compactness, there exist finite subsets 
$A \subseteq \D_A, B \subseteq \D_B$ such that 
$\T_0 \cup \K \cup A \cup B$ is unsatisfiable. As $\T_1 = \T_0 \cup  \K$
is $W$-separable, it follows that 
$\T_0 \cup  (\K[W(A,B)] \cup A) \cup  (\K[W(B,A)] \cup B)$ 
is unsatisfiable. 

On the other hand, $A$ contains only terms in $T_{M_A}$ and $B$ only
terms in $T_{M_B}$. Thus, $W(A, B) \subseteq W(T_{M_A}, T_{M_B})
\subseteq T_{M_A}$ and $W(B, A) \subseteq W(T_{M_B}, T_{M_A})
\subseteq T_{M_B}$. 

Therefore, all extension terms in $\K[W(A,B)] \cup A$ are defined in
$M_A$ and all extension terms in $\K[W(B,A)] \cup B$ are defined in
$M_B$.

As $\T_0$ has general ground interpolation, $\T_0 \cup {\sf
  UIF}_{\Sigma}$ has ground interpolation. Thus, 
there exists an interpolant $I$ for  $(\K[W(A,B)] \cup A)$ and
$(\K[W(B,A)] \cup B)$ w.r.t.\ 
$\T_0 \cup {\sf  UIF}_{\Sigma}$ such that all $\Sigma$-terms in $I$ are in
$W(W(A,B), W(B, A)) \cap W(W(B,A), W(A, B))$, thus as $W$ is a
closure, in $W(A, B) \cap W(B, A)$.

As all terms in $\K[W(A,B)] \cup A$ and $I$ are defined in $M_A$ and 
$\K[W(A,B)] \cup A \models_{\T_0 \cup {\sf UIF}_{\Sigma}} I$, it follows that $I$
is true in $M_A$. 
As $I$ contains only terms that are defined in $M_B$ and the
definitions of the extension functions in $M_A$ and $M_B$ agree for
defined terms, $I$ is also true in $M_B$. 
On the other hand, we know that $I \cup \K[W(B,A)] \cup B \models_{\T_0 \cup {\sf UIF}_{\Sigma}} \perp$, 
hence $\K[W(B,A)] \cup B \models_{\T_0 \cup {\sf UIF}_{\Sigma}} \neg
I$. As all terms in $\K[W(B,A)] \cup B$ and of $I$ are defined in
$M_B$, and  $\K[W(B,A)] \cup B$ is true in $M_B$, $\neg I$ must be
true in $M_B$. Contradiction. 

It follows that $\T_0 \cup \K \cup \D_A \cup \D_B$ is satisfiable. 
Let $M_D$ be a model for $\T_0 \cup \K \cup \D_A \cup \D_B$. 
Define $h_A : M_A \rightarrow M_D$ and $h_B : M_B \rightarrow M_D$ by
$h_A(a) = h_B(a) = a_{M_D}$ for every $a \in M$ (where $a_{M_D}$ is
the interpretation of the constant $a$ in $M_D$). 
 
Then $M_D$ is a model of $\T_0 \cup \K$. 
$h_A$ and $h_B$ are clearly injective: Assume that $a, a' \in M$ and
$a \not\approx a'$. Then $a \not\approx a' \in \D_A$ so this literal is true in
$M_D$, hence $h_A(a) = h_B(a) = a_{M_D} \neq a'_{M_D} = h_B(a') =
h_A(a')$. 
Also, whenever 
$f_{M_A}(a_1, \dots, a_n)$ is defined and is equal to $a$, $f(a_1,
\dots, a_n) \in \D_A$, hence $f_{M_D}({a_1}_{M_D}, \dots, {a_n}_{M_D})
= a_{M_D}$. 
Moreover, $(a_1, \dots, a_n) \in P_{M_A}$ iff $P(a_1, \dots, a_n) \in
\D_A$ iff $({a_1}_{M_D}, \dots, {a_n}_{M_D}) \in P_{M_D}$. 
Thus, $h_A$ is a weak embedding; the proof that $h_B$ is a
weak embedding is similar. 
These embeddings clearly agree on $M$.
\end{proof}

We now show that in general 
partial $W$-amalgamation implies partial  $W$-amalgamation for models
with the same $\Pi_0$-reduct, and that if 
$\T_0$ allows general ground interpolation 
the two notions are equivalent.

\begin{prop}
Let $W$ be an amalgamation closure operator. 
Let $\K$ be a set of clauses over $\Pi_0 \cup \Sigma$. The following
hold: 
\begin{enumerate-}
\item[(1)] If $\T = \T_0 \cup \K$ has the partial $W$-amalgamation
  property then it has the partial $W$-amalgamation property 
for models with the same $\Pi_0$-reduct. 

\item[(2)] If $\T_0$ is a first-order theory which allows general ground
interpolation and $\T = \T_0 \cup \K$ has  the partial $W$-amalgamation property 
for models with the same $\Pi_0$-reduct then $\T = \T_0 \cup \K$ 
has the partial $W$-amalgamation property.  
\end{enumerate-}
\end{prop}
\begin{proof}(1) 
Assume that $\T = \T_0 \cup \K$ has the partial $W$-amalgamation
property. Let $M_A, M_B \in {\sf PMod}_{w,f}(\Sigma, \T)$ having 
the same reduct $M$ to $\Pi_0$ and such that: 
\begin{enumerate-}
\item[(i)] for all $m_1, \dots, m_n \in |M|$, if 
$f_{M_A}(m_1, \dots, m_n)$ and $f_{M_B}(m_1, \dots, m_n)$
are both defined and $f_{M_A}(m_1, \dots, m_n) = m$ and  
$f_{M_B}(m_1, \dots, m_n) = m'$ then $m =
m'$, and 
\item[(ii)] the sets $T_{M_A}$ and $T_{M_B}$ of terms 
defined in $M_A$ resp.\ $M_B$ are closed under $W$.
\end{enumerate-}
Let $M_C$ be the partial structure with the same reduct $M$ to $\Pi_0$ as
$M_A$ and $M_B$, and where for every $f \in \Sigma$ with arity $n$, 
and every $m_1, \dots, m_n \in |M|$, 
$f_{M_C}(m_1, \dots, m_n)$ is defined iff $f_{M_A}(m_1, \dots, m_n)$
and $f_{M_B}(m_1, \dots, m_n)$ are both defined, and if so 
$f_{M_C}(m_1, \dots, m_n) = f_{M_A}(m_1, \dots, m_n) = f_{M_B}(m_1,
\dots, m_n)$. Clearly, the conditions in
Definition~\ref{def:part-W-amalg}
are satisfied: 
\begin{itemize}
\item The universe $|M_C|$ of $M_C$ is included in the universes of $M_A$ and $M_B$
and the inclusions into $M_A, M_B$ are weak embeddings; 
\item $|M_C| = |M| = |M| \cap |M| = |M_A| \cap |M_B|$; 
\item The sets $T_{M_A}$ and $T_{M_B}$ of terms which are
defined in $M_A$ resp.\ $M_B$ are closed under $W$.
\item $T(M_A) \cap T(M_B) \subseteq T(M_C)$: If $f(a_1, \dots, a_n)
  \in T(M_A) \cap T(M_B)$ then $f_{M_A}(a_1, \dots, a_n)$ is defined and
  $f_{M_B}(a_1, \dots, a_n)$ is defined (and they are equal by (i)), so by the definition of $M_C$
  $f_{M_C}(a_1, \dots, a_n)$ is also defined. 
\end{itemize}
It follows that there exists a model $M_D$ of $\T$, and weak embeddings $h_A : M_A
\rightarrow M_D$ 
and $h_B : M_B \rightarrow M_D$,
such that ${h_A}_{|_{|M_C|}} = {h_B}_{|_{|M_C|}}$, i.e.\ $h_A$ and $h_B$
coincide on $M$.  

\medskip
\noindent (2) Assume that $\T_0$ allows general ground
interpolation and $\T = \T_0 \cup \K$ has  the partial $W$-amalgamation property 
for models with the same $\Pi_0$-reduct. 
Let $M_A, M_B, M_C \in {\sf PMod}_{w,f}(\Sigma, \T)$ satisfying the
conditions from Definition~\ref{def:part-W-amalg}. 
Then ${M_C}_{|_{\Pi_0}}$ is a substructure  of ${M_A}_{|_{\Pi_0}}$ and of
${M_B}_{|_{\Pi_0}}$. 

By Theorem~\ref{thm:ghilardi-amalg-interp} (cf.\ also
\cite{Ghilardi-2014}),  $\T_0$ allows general ground
interpolation iff $\T_0$ is strongly sub-amalgamable
(cf.\ Definition~\ref{def:strong-subamalgamation}). Therefore there
exists a further model $M$ of $\T_0$ and embeddings $\mu_1 :
{M_A}_{|_{\Pi_0}} \rightarrow M$ and $\mu_2 : {M_B}_{|_{\Pi_0}}\rightarrow M$
whose restrictions to ${M_C}_{|_{\Pi_0}}$ coincide, 
such that if $\mu_1(m_1) = \mu_2(m_2)$ then there exists $m \in
|{M_C}_{|_{\Pi_0}}|$ with $m = m_1 = m_2$. 

We use the embeddings $\mu_1$ and $\mu_2$ to construct two partial
algebras $M'_A$ and $M'_B$  as
follows: 
$M'_A$ has universe $|M|$, all $\Pi_0$ operations are defined as
  in $M$, and for every $f \in \Sigma$ with arity $m$ and every $m_1,
  \dots, m_n \in |M|$, $f_{M'_A}(m_1, \dots, m_n)$ is defined iff 
there exist ${\overline m}_1, \dots, {\overline m}_n \in |M_A|$ such
that $\mu_1({\overline m}_i) = m_i$ and $f_{M_A}({\overline m}_1,
\dots, {\overline m}_n)$ is defined. 
If this is the case, then $f_{M'_A}(m_1, \dots, m_n) = \mu_1(f_A({\overline m}_1, \dots,
{\overline m}_n))$. $M'_B$ is defined analogously.
We show that $M'_A, M'_B$ satisfy the conditions from
Definition~\ref{def:part-amalg}. 
\begin{itemize}
\item $M'_A, M'_B$  have the same reduct $M$ to $\Pi_0$, 
\item For $m_1, \dots, m_n \in |M|$, if $f_{M'_A}(m_1, \dots, m_n)$ is defined and equal to $m$ and  
$f_{M'_B}(m_1, \dots, m_n)$ is defined and equal to $m'$ then $m =
m'$. \\
Indeed, if $f_{M'_A}(m_1, \dots, m_n)$ is defined then there
exist ${\overline m}_1, \dots, {\overline m}_n \in |M_A|$ with
$\mu_1({\overline m}_i) = m_i$ such that 
$f_{M_A}({\overline m}_1, \dots,{\overline m}_n)$ is defined 
and $f_{M'_A}(m_1, \dots, m_n) = \mu_1(f_{M_A}({\overline m}_1, \dots,
{\overline m}_n))$. \\
If $f_{M'_B}(m_1, \dots, m_n)$ is defined then there
exist ${\overline {\overline m}}_1, \dots, {\overline {\overline m}}_n \in |M_B|$ with
$\mu_2({\overline {\overline m}}_i) = m_i$ such that 
$f_{M_B}({\overline {\overline m}}_1, \dots, {\overline {\overline m}}_n)$ is defined 
and $f_{M'_B}(m_1, \dots,m_n) = \mu_2(f_{M_B}({\overline {\overline m}}_1, \dots,
{\overline {\overline m}}_n))$. Since $\mu_1({\overline m}_i) =
\mu_2({\overline {\overline m}}_i) = m_i$, there exists $m'_i \in |M_C|$
such that $m'_i = {\overline m}_i = {\overline {\overline m}}_i$. 

As $f_{M_A}({\overline m}_1, \dots,
{\overline m}_n)$ and 
$f_{M_B}({\overline {\overline m}}_1, \dots, {\overline {\overline
    m}}_n)$ are defined, $f({\overline m}_1, \dots,
{\overline m}_n) \in T_{M_A}$ and 
$f({\overline {\overline m}}_1, \dots, {\overline {\overline m}}_n)
\in T_{M_B}$. 
Since for every $i$, $m'_i = {\overline m}_i = {\overline {\overline m}}_i$
and we assumed that $T(M_A) \cap T(M_B) \subseteq T(M_C)$, it follows that $f(m'_1,
\dots, m'_n) \in T(M_C)$, so $f_{M_C}(m'_1, \dots, m'_n)$ is defined.
Since $M_C$ is a weak substructure of $M_A, M_B$ it follows that 
$f_{M_A}(m'_1, \dots, m'_n)  = f_{M_C}(m'_1, \dots, m'_n) =
f_{M_B}(m'_1, \dots, m'_n)$, and therefore (since $\mu_1$ and $\mu_2$
agree on $M_C$) we have: 
$m = \mu_1(f_{M_A}(m'_1, \dots, m'_n)) = \mu_1(f_{M_C}(m'_1, \dots,
m'_n)) = \mu_2(f_{M_C}(m'_1, \dots,
m'_n)) = \mu_2(f_{M_B}(m'_1, \dots, m'_n)) = m'$. 
\item Up to renaming of constants $h$ defined using $\mu_1$ and $\mu_2$: $T_{M'_A} =
 h(T_{M_A})$, $T_{M'_B} = h(T_{M_B})$, 
so $T_{M'_A}$ and $T_{M'_B}$ are closed under  $W$. 
\end{itemize}
It follows that there exists a model $M_D$ of $\T_0 \cup \K$ and weak embeddings $h'_A:
M'_A \rightarrow M_D, h'_B : M'_B \rightarrow M_D$ which agree on $M$.
The maps $h_A = \mu_1 \circ h'_A : M_A \rightarrow M_D$ and 
$h_B = \mu_2 \circ h'_B : M_B \rightarrow M_D$ are then weak embeddings
such that ${h_A}_{|_{M_C}} = {h_B}_{|_{M_C}}$, thus $\T_0 \cup \K$ has the partial
$W$-amalgamation property.
\end{proof}

\subsection{Separability, Locality and Interpolant Computation}
If the extension ${\mathcal T}_0 \subseteq \T_0 \cup \K$ is
$W$-separable and ${\mathcal T}_0$ has ground interpolation, 
then we can hierarchically compute interpolants in ${\mathcal T}_0
\subseteq \T_0 \cup \K$ (cf.\ also \cite{Wies-journal}). 
%
%
%
\begin{thm}
Let $W$ be an amalgamation closure operator. 
Assume that the theory ${\mathcal T}_0$ has general ground interpolation,
and there is a method for effectively computing general ground
interpolants w.r.t.\ ${\mathcal T}_0$. 
Let ${\mathcal T}_0 \cup {\mathcal K}$ be a $W$-separable extension of 
${\mathcal  T}_0$ with a set of clauses ${\mathcal K}$ in which every variable
occurs below an extension function. 
Let $A$ and $B$ be two ground $\Sigma_0 \cup \Sigma$-formulae. 
Assume that $A \wedge B \models_{{\mathcal T}_0 \cup {\mathcal K}} \perp$. 
Then we can effectively compute a ground interpolant for $A$ and $B$,
by computing an interpolant of $\K[W(A, B)] \cup A$ and $\K[W(B, A)]
\cup B$. 
\label{hierarchical-gr-int}
\end{thm} 

\begin{proof}By $W$-separability 
${\mathcal T}_0 \cup {\mathcal K} \cup (A \wedge B) \models \perp$ iff 
${\mathcal T}_0 \cup {\mathcal K}[W(A, B)] \cup A \cup 
{\mathcal K}[W(A, B)] \cup B \models \perp$. 
As every variable occurs in $\K$ below an extension function, 
${\mathcal K}[W(A, B)] \cup A \cup 
{\mathcal K}[W(B, A)] \cup B$ is a set of ground formulae. 

We can use the method for computing general ground interpolants in
$\T_0$ for computing the interpolant  $I_0$, which is an interpolant
for $A$ and $B$ w.r.t.\ $\T_0 \cup \K$.
\end{proof}

\begin{cor}
Let $W$ be an amalgamation closure operator, and 
let ${\mathcal T}_0 \cup {\mathcal K}$ be a $W$-separable extension of ${\mathcal  
  T}_0$ with a set of clauses ${\mathcal K}$ in which every variable  
occurs below an extension function.  
Then $\T_0 \cup \K$ has ground interpolation in each of the
following cases: 
\begin{enumerate-}
\item $\T_0$ has ground interpolation and is equality
interpolating. 
\item $\T_0$ allows quantifier elimination and is equality
interpolating. 
\item $\T_0$ is universal and allows quantifier elimination. 
\end{enumerate-}
\label{cor:loc-th-int}
\end{cor}

\begin{proof}
(1) If $\T_0$ has ground interpolation and is equality
interpolating, then by Theorem~\ref{thm:ghilardi-amalg-interp}(3) and
(2), the extension of $\T_0$
with uninterpreted function symbols in $\Sigma$ has ground
interpolation.\\ 
(2) If $\T_0$ allows quantifier elimination then it has ground
interpolation, so (1) can be used. 
(3) follows from (2) and
Theorem~\ref{thm:ghilardi-amalg-interp}(4).
\end{proof}

\subsection{Ground Interpolation and Model Completions} 
It is sometimes difficult to check directly whether the theory $\T_0$
has ground interpolation. If $\T_0$ has a model completion with good 
properties, this becomes easier to check. In this case, we can use
quantifier elimination in the model completion to compute the
interpolant. 

\begin{thm}
Let $W$ be an amalgamation closure operator, and 
let ${\mathcal T}_0 \cup {\mathcal K}$ be a $W$-separable extension of ${\mathcal  
  T}_0$ with a set of clauses ${\mathcal K}$ in which every variable  
occurs below an extension function.  

Assume that $\T_0$ has a model companion $\T_0^*$ with the following 
properties: 
\begin{enumerate-}
\item $\T_0 \subseteq \T_0^*$; 
\item $\T^*_0$ has general ground interpolation. (This can
    happen for instance when $\T^*_0$ allows quantifier elimination and
    is equality interpolating.) 
\end{enumerate-}
Then $\T_0 \cup \K$ has ground interpolation.  
\label{thm-int-local-mc}
\end{thm}

\begin{proof}
Assume that ${\mathcal T}_0 \cup {\mathcal K} \cup (A \wedge B)
\models \perp$. 
By $W$-separability 
${\mathcal T}_0 \cup {\mathcal K} \cup (A \wedge B) \models \perp$ iff 
${\mathcal T}_0 \cup {\mathcal K}[W(A, B)] \cup A \cup 
{\mathcal K}[W(A, B)] \cup B \models \perp$. 
As every variable occurs in $\K$ below an extension function, 
${\mathcal K}[W(A, B)] \cup A \cup 
{\mathcal K}[W(A, B)] \cup B$ is a set of ground formulae. 

Condition (1) implies that every model of $\T_0^* \cup {\sf
  UIF}_{\Sigma}$ is a model of $\T_0\cup {\sf
  UIF}_{\Sigma}$.  

From the assumption that $\T_0^*$ is a model companion of $\T_0$ we
know that every model of $\T_0$ embeds into a model of $\T_0^*$.
This implies that every model of $\T_0 \cup {\sf
  UIF}_{\Sigma}$ embeds into a model of $\T_0^* \cup {\sf
  UIF}_{\Sigma}$. Indeed, let $\A$ be a model of 
$\T_0 \cup {\sf UIF}_{\Sigma}$. Then $\A_{|\Pi_0}$ is a model of $\T_0$ thus 
it embeds into a model $\B$ of $\T_0^*$.
We define a partial $\Sigma$-structure $P$, having as
$\Pi_0$-reduct $\B$, and such 
that for every $f \in \Sigma$ with arity $n$, $f_{P}(a_1, \dots,
a_n)$ is defined iff $a_1, \dots, a_n \in |\A|$, and if so 
$f_{P}(a_1, \dots,a_n) = f_{\A}(a_1, \dots, a_n)$. $P$ can be 
transformed into a total algebra ${\mathcal P}$ by selecting an element $c$ in its support
and setting the values of all function symbols in $\Sigma$ to $c$ 
if they are undefined in $P$. ${\mathcal P}$ has the same 
$\Pi_0$-reduct as $P$, namely $\B$, and is thus a model of $\T_0^*
\cup {\sf UIF}_{\Sigma}$.

It therefore follows that $\T_0 \cup {\sf
  UIF}_{\Sigma}$ and $\T_0^* \cup {\sf
  UIF}_{\Sigma}$ are co-theories, so, 
by Lemma~\ref{lem:co-theories-sat}, ${\mathcal T}^*_0 \cup {\mathcal K}[W(A, B)] \cup A \cup 
{\mathcal K}[W(A, B)] \cup B \models \perp$. 
As $\T_0^*$ has general ground interpolation, we know that there exists a
ground formula $I$ containing only the common constants and extension 
functions of ${\mathcal K}[W(A, B)] \cup A$ and
${\mathcal K}[W(A, B)] \cup B$ such that:
\begin{enumerate}
\item[(a)]  ${\mathcal K}[W(A, B)] \cup A \cup \neg I \models_{\T^*_0 \cup {\sf UIF}_{\Sigma}}
  \perp$, hence by Lemma~\ref{lem:co-theories-sat}, \\
${\mathcal K}[W(A, B)] \cup A \cup \neg I \models_{{\T}_0 \cup {\sf UIF}_{\Sigma}}
  \perp$; 
\item[(b)] ${\mathcal K}[W(B, A)] \cup B \cup I \models_{\T^*_0 \cup {\sf UIF}_{\Sigma}} 
  \perp$, hence by Lemma~\ref{lem:co-theories-sat}, \\
${\mathcal K}[W(B, A)]
  \cup B \cup I \models_{\T_0 \cup {\sf UIF}_{\Sigma}} \perp$. 
\end{enumerate}
Thus, $I$ is an interpolant w.r.t.\ $\T_0 \cup {\sf UIF}_{\Sigma}$, hence,
as can be seen from the proof of Theorem~\ref{hierarchical-gr-int}, it is an interpolant of $A$ and $B$ w.r.t.\ $\T_0
\cup \K$.
\end{proof}

\begin{exa}
Consider the theory extension $\T_0 \cup \K$ in Examples~\ref{example-hierarchic},
\ref{ex-idea-interp} and \ref{ex1-interp}, where $\T_0 = {\sf TOrd}$ is the theory
of total orderings and $\K = \{ {\sf SGC}(f, g), {\sf Mon}(f, g) \}$. 
Let $A$ and $B$ be as in Example~\ref{example-hierarchic}: 

\medskip
\noindent 
~~~~~~~~~~~~~$A:~~ d \leq g(a) ~\wedge~ a \leq c \quad \quad 
B:~~  b \leq d ~\wedge~ f(b) \not\leq c.$

\medskip
\noindent 
We already proved that $A \wedge B \models_{\T_0 \cup \K} \perp$, 
by using hierarchical reasoning in local theory extensions;  after 
instantiation and purification we obtained: 
$$\begin{array}{l|ll}
\hline 
{\sf Extension} & ~~~~~~{\sf Base} \\
D_A \wedge D_B & ~A_0 \wedge B_0  \wedge {\sf SGc}_0 \wedge {\sf
  Mon}_0 & \wedge {\sf Con}_0  \\
\hline 
a_1 \approx g(a) ~   & ~A_0 = d \leq a_1  \wedge a \leq c & {\sf SGc}_0   = b \leq a_1 \rightarrow b_1 \leq a  \\
b_1 \approx f(b)  & ~B_0 = b \leq d \wedge  c < b_1  & {\sf Con}_A \wedge {\sf Mon}_A = a \lhd a \rightarrow a_1 \lhd a_1, \lhd \in \{ \approx, \leq \}\\
& &  {\sf Con}_B \wedge {\sf Mon}_B =  b \lhd b \rightarrow b_1 \lhd b_1,~ \lhd \in \{ \approx, \leq \} \\
\hline 
\end{array}$$
As mentioned before, 
$A_0 \wedge B_0 \models b \leq a_1$ and, in fact: 
$A_0 \wedge
B_0 \models b \leq d \wedge d \leq a_1$ ($d$ is a shared
constant).
After separation and purification of the newly introduced
instances of the axioms ${\sf SGc}$ and ${\sf Mon}$ 
(using $d_1$ for $f(d)$) as explained in \cite{Sofronie-lmcs}
we obtain: 
\begin{itemize}
\item ${\overline A}_0 = d \leq a_1  \wedge a \leq c \wedge (d \leq
  a_1 \rightarrow d_1 \leq a)$ equiv. to $(d \leq a_1
  \wedge a \leq c \wedge d_1 \leq a)$
\item ${\overline B}_0 = b \leq d \wedge c <  b_1  \wedge (b
  \leq d \rightarrow b_1 \leq d_1)$ equiv. to $(b \leq
  d \wedge  b_1 \not\leq c \wedge b_1 \leq d_1)$
\end{itemize}
We can use a method for ground interpolation in the theory of total
orderings to obtain a ground interpolant $I_0$. However, it might be
more efficient to do so by using quantifier elimination in the model
completion of $\T_0$ (the theory of dense total orderings
without endpoints) to eliminate the constants $a, a_1$ from
${\overline A}_0$. 
We can eliminate quantifiers as follows: 
\[ \begin{array}{lll}
\exists a \exists a_1 \, ( d \leq a_1  \wedge a \leq c \wedge (d \leq
  a_1 \rightarrow d_1 \leq a)) & \equiv & \exists a \exists a_1 \, ( d
  \leq a_1  \wedge a \leq c \wedge d_1 \leq a)\\
& \equiv & \exists a \, ( a \leq c \wedge d_1 \leq a)\\
& \equiv & d_1 \leq c.
\end{array} \] 
We obtain the interpolant $I_0 = d_1 \leq c$. Since $d_1$ is used 
as an abbreviation for $f(d)$, we replace it back and obtain the
interpolant $I = f(d) \leq c$. 

\medskip
A similar result can also be obtained using $W$-separability with the second 
of the amalgamation operators defined in Example~\ref{ex1-interp}: 
If $D_{AB} = \{ d \}$ is the set of constants
common to $A$ and $B$ which can be used for  $\leq$-interpolation then 
an amalgamation closure which can be used is 
$W(A, B) = {\sf
  st}(A) \cup \{ f(e), g(e) \mid e \in D_{AB} \} = \{ a, c, d, g(a) \} \cup
\{ f(d), g(d) \}$ and $W(B, A) = {\sf st}(B) \cup \{ f(e), g(e)
\mid e \in D_{AB} \} = \{ b, c, d, f(b) \} \cup \{ f(d), g(d) \}$. 
The results in \cite{Sofronie-lmcs} show that the smaller sets 
$W'(A, B) = \{ a, c, d, g(a) \} \cup
\{ f(d) \}$ and $W'(B, A) = \{ b, c, d, f(b) \} \cup \{ f(d) \}$ are
sufficient for this example. 
After instantiation and purification we obtain the conjunctions 
${\overline A}_0$ and ${\overline B}_0$ of ground clauses we
considered above. 
\end{exa}

\subsection{Symbol Elimination and Interpolation} For $W$-separable theories
we can  use the method for symbol elimination  in Section~\ref{symb-elim}
for computing interpolants. If $\T_0 \cup \K[W(A, B)] \cup A \cup
\K[W(B, A)] \cup B \models \perp$, the formula $\Gamma_2$ obtained 
using Steps 1--4 in Section~\ref{symb-elim} for $\T_0 \cup \K[W(A, B)] \cup
A$ (with $\Sigma_P$ consisting of the common constants) is an
interpolant. 

\begin{thm}
If $\T_0 \cup \K[W(A, B)] \cup A \cup
\K[W(B, A)] \cup B \models \perp$, the formula $\Gamma_2$ obtained 
using Steps 1--4 in Sect.~\ref{symb-elim} for $\T_0 \cup \K[W(A, B)] \cup
A$ (with $\Sigma_P$ consisting of the common extension functions and constants) is an
interpolant of $A$ and $B$ w.r.t.\ $\T_0 \cup \K$. 
\label{app:th:se}
\end{thm} 

\begin{proof}By Lemma~\ref{lemma:symb-elim}, we know that
the formulae 
$\exists {\overline  x} G_1({\overline c}_p, {\overline c}_f,
{\overline x}) \wedge {\sf Def}$ and $\Gamma_2({\overline c}_p)$ are 
equivalent w.r.t.\ $\T_0 \cup {\sf UIF}_{\Sigma}$.
In other words, for every model $\A$ of $\T_0 \cup {\sf
  UIF}_{\Sigma}$, $\A \models \Gamma_2({\overline c}_p)$ if and only
if its extension $\A^{{\overline c}_f}$ with the new constants ${\overline c}_f$ 
with interpretations defined such that ${\sf Def}$ are true, is a 
model of $\exists {\overline  x} G_1({\overline c}_p, {\overline c}_f,
{\overline x})$. 

We first prove that $A \models_{\T_0 \cup \K} 
\Gamma_2({\overline c}_p)$. 
Let $\A$ be a model of $\T_0 \cup \K \cup A$. 
Then $\A$ is a model of $\T_0 \cup \K[W(A, B)] \cup A$, so 
its extension $\A^{{\overline c}_f}$ with the new constants ${\overline c}_f$ 
with interpretations defined such that ${\sf Def}$ are true is a model
of $\T_0 \cup \K[W(A, B)])_0 \cup
A_0 \cup {\sf Def}$. 
It follows therefore that  $\A^{{\overline c}_f}$  is a model of 
$\exists {\overline
  x} G_1({\overline c}_p, {\overline c}_f, {\overline x}) \cup {\sf  Def}$. 
By Lemma~\ref{lemma:symb-elim}, $\A$ is then a model of $\Gamma_2({\overline c}_p)$.

\medskip
We now show that $\Gamma_2({\overline c}_p) \cup B \models_{\T_0 \cup
  \K} \perp$. Assume that this is not the case, i.e.\ there exists a
structure $\A$ which is a model of $\T_0 \cup \K$, of $B$  and of 
$\Gamma_2({\overline  c}_p)$.
Then: 
\begin{itemize}
\item $\A \models \K[W(B, A)] \cup B$.
\item As $\A$ is a model of $\Gamma_2({\overline  c}_p)$, its
  extension $\A^{{\overline c}_f}$ with the new constants ${\overline c}_f$ 
with interpretations defined such that ${\sf Def}$ are true 
is a model of $\exists {\overline
  x} G_1({\overline c}_p, {\overline c}_f, {\overline x})$. 
Thus,  there exists $\beta : X \rightarrow |\A|$ such that 
$A, \beta 
\models G_1({\overline c}_p, {\overline c}_f, {\overline  x})$. 

Let $G_1({\overline c}_p, {\overline c}_f, {\overline  c})$ be the
ground formula obtained from 
$G_1({\overline c}_p, {\overline c}_f, {\overline  x})$ by replacing
every variable $x_i$ in ${\overline x} = x_1, \dots, x_n$ with the 
constant $c_i$ in ${\overline c} = c_1, \dots, c_n$. 
Let $\A^{{\overline c}_f, {\overline c}}$ be the structure
that coincides with $\A^{{\overline c}_f}$, except
for the values of the constants in ${\overline c}$ which are given 
by the values of the variables ${\overline x}$ w.r.t.\ $\beta$. 
Then $\A^{{\overline c}_f, {\overline c}}\models  G_1({\overline c}_p,
{\overline c}_f, {\overline  c})$.

From the definition of $G_1$, $\A^{{\overline c}_f, {\overline
    c}}\models  \T_0 \cup \K[W(A, B)])_0 \cup
A_0 \cup {\sf Def}$, hence

$\A^{{\overline c}}\models \T_0 \cup \K[W(A, B)] \cup A$. 

\item 
As the constants ${\overline c}$ do not occur in $B$ or $\K[W(B, A)]$, 
$\A^{{\overline c}} \models \K[W(B, A)] \cup B$.
\end{itemize}
It follows that $\K[W(A, B)] \cup A \cup
\K[W(B, A)] \cup B$ is satisfiable w.r.t.\ $\T_0 \cup {\sf UIF}_{\Sigma}$. Contradiction. 
This shows that $\Gamma_2({\overline
  c}_p) \wedge B \models_{\T_0 \cup \K} \perp$.
\end{proof}

\begin{exa}
Consider the theory $\T_0 \cup \K$ in 
Example ~\ref{ex1-symb-elim}: $\T_0$ is the theory of dense total orderings without
endpoints; we consider its extension with functions $\Sigma_1 =
\{ f, g, h, c \}$
whose properties are axiomatized by 

\[ \begin{array}{ll} 
\K := \{ & \forall x (x \leq c \rightarrow g(x) \approx f(x)), \quad \quad  \forall x  ( c < x \rightarrow g(x) \approx h(x)) \quad \}.
\end{array} \] 

\

\noindent Let $A$ and $B$ be the formulae: 
\begin{itemize}
\vspace{1mm}
\item $A := \{ c_1 \leq c_2, \quad g(c_1) \approx a_1, \quad  g(c_2) \approx a_2, a_1 > a_2  \}$
\vspace{1mm}
\item $B := \{ c_1 \leq c < c_2, \quad f(c_1) \approx b_1, \quad
  h(c_2) \approx b_2, \quad  b_1 \leq
b_2 \}$. 
\end{itemize}
\vspace{1mm}
\noindent It is easy to check that $\T_0 \cup \K \cup A \cup B \models
\perp$. 
The common symbols of $A$ and $B$ are $c_1$ and $c_2$, as well the
constant $c$ which is part of the signature of the theory $\T_0 \cup
\K$. All function symbols $f, g, h$ can be considered to be shared
because they are used together in the axioms in $\K$. 
An interpolant of $A$ and $B$ cannot contain the constants $a_1, a_2, b_1, b_2$. 

We can compute an interpolant by eliminating the symbols $a_1, a_2$ 
from $A$ with the method described in Section~\ref{symb-elim}. 
Note that the formula $A$ coincides with the set obtained from the
family $G$ of ground clauses considered in Example~\ref{ex1-symb-elim}
after introducing new constants $a_1, a_2$ for the extension terms $g(c_1),
g(c_2)$.   We apply Steps 1--4 to this formula. 
Let $\Gamma_2$ be the formula computed in Step 4 in
Example~\ref{ex1-symb-elim}, namely: 

\[ \begin{array}{ll} 
\Gamma_2 = ( & (c_1 < c_2 \leq c ~~\wedge~~ f(c_1)> f(c_2))  ~~\vee \\
& (c_1 \leq c < c_2 ~~\wedge~~ f(c_1) > h(c_2))  ~~\vee \\
& (c < c_1 < c_2 ~~\wedge~~ h(c_1) > h(c_2) ) \quad )
\end{array}
\] 

\

\noindent 
By Theorem~\ref{app:th:se}, this formula is an interpolant of $A$ and $B$. 
\end{exa}


\section{Conclusions, Summary of Results}
In this paper we studied several problems related to symbol
elimination and ground interpolation in theories and theory 
extensions. We here briefly summarize these results, then discuss 
some directions in which we would like to extend them.  

\subsection{Amalgamation, Ground Interpolation, Quantifier Elimination}
\label{summary-1}

 It is well-known that if a theory has quantifier 
elimination then this can be used for symbol elimination and 
also for computing ground interpolants of ground formulae. 
However, the great majority of logical theories do not have 
quantifier elimination. We showed that if a theory $\T$ has a 
model completion $\T^*$, then interpolants computed w.r.t.\
$\T^*$ are also interpolants w.r.t.\ $\T$. As there are many 
examples of model completions of theories $\T$ which allow 
quantifier elimination, this can be used for computing interpolants 
w.r.t.\ $\T$.

\medskip

\noindent 
The links between the different notions amalgamation, quantifier-free
interpolation, quantifier elimination and the quality of being
equality interpolating are summarized below:


\medskip

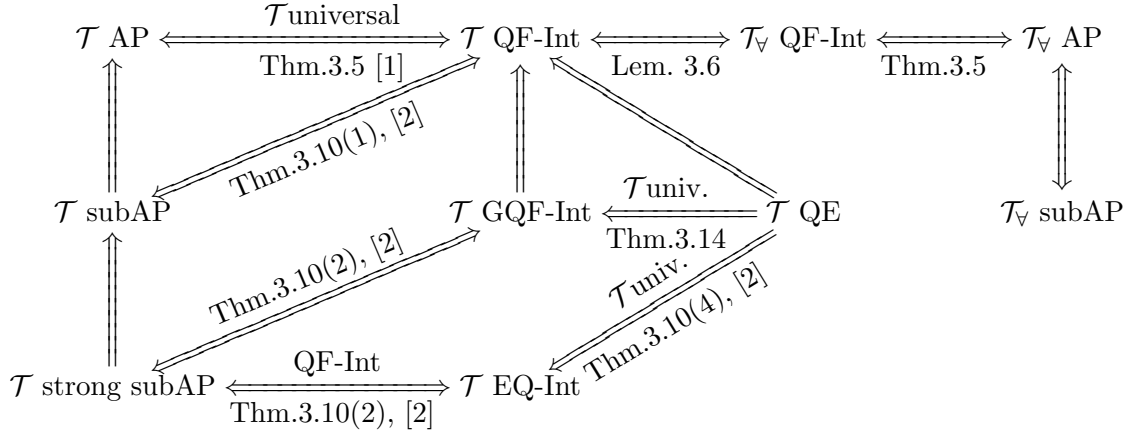
\begin{diagram}
\T \text{ AP} & \rEquivalent^{\T
  \text{universal}}_{\text{Thm.\ref{bacsich75} \cite{Bacsich75}}} &
\T \text{ QF-Int} &
\rEquivalent_{\text{Lem.~\ref{t-forall-t-ground-int}}}& \T_{\forall}
\text{ QF-Int}& \rEquivalent_{\text{Thm.\ref{bacsich75}}}&
\T_{\forall} \text{  AP} \\
\uImplies &
\ruEquivalent_{\text{Thm.\ref{thm:ghilardi-amalg-interp}(1),
    \cite{Ghilardi-2014}}} & \uImplies & \luImplies& & & \uEquivalent\\
\T \text{  subAP}& & \T \text{  GQF-Int} & \lImplies^{\T
  \text{univ.}}_{\text{Thm.}\ref{thm:QE-GGI}} & \T \text{ QE} & & \T_{\forall} \text{ subAP} \\
\uImplies &
\ruEquivalent^{\text{Thm.\ref{thm:ghilardi-amalg-interp}(2),
    \cite{Ghilardi-2014}}} & &  \ldImplies^{\T
  \text{univ.}}_{\text{Thm.\ref{thm:ghilardi-amalg-interp}(4), \cite{Ghilardi-2014}}}\\
\T \text{ strong subAP} & \rEquivalent^{\text{ QF-Int}}_{\text{Thm.\ref{thm:ghilardi-amalg-interp}(2),
    \cite{Ghilardi-2014}}} & \T \text{  EQ-Int}
\end{diagram}

\

\noindent We used the following abbreviations: 
\begin{itemize}
\item AP: amalganation property;  
\item (strong) subAP: (strong) sub-amalgamation property; 
\item (G)QF-Int: (general) quantifier free interpolation; 
\item EQ-Int: Equality Interpolating
\end{itemize}

\

\noindent 
Let $\T$ be a theory and $\T^*$ a model companion of $\T$. 
The links between the properties of $\T^*$ and 
amalgamation in $\T$ and the links between ground interpolation in
$\T$ and $\T^*$ are summarized below: 

\medskip

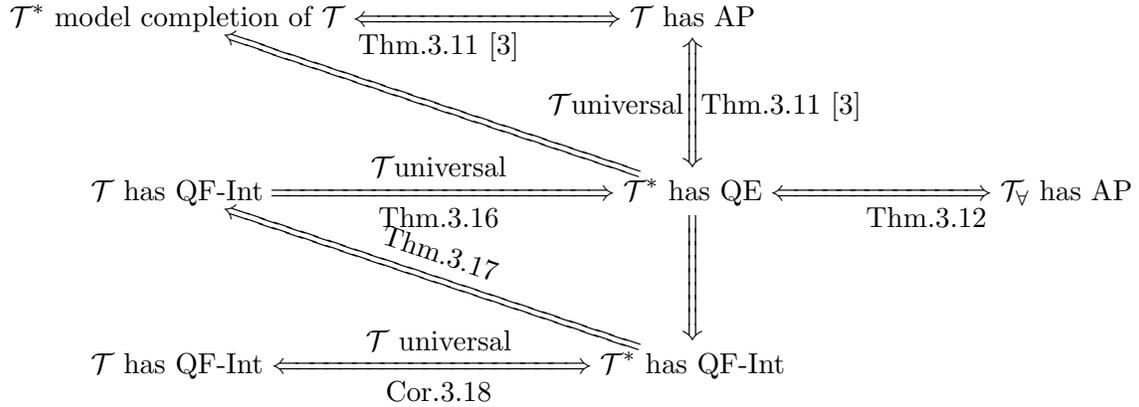
\begin{diagram}
\T^* \text{ model completion of } \T &
\rEquivalent_{\text{Thm.\ref{thm:criteria-model-compl}
    \cite{chang-keisler}}}  & \T \text{ has AP} \\
& \luImplies & \dEquivalent^{\T \text{universal}}_{\text{Thm.\ref{thm:criteria-model-compl}
    \cite{chang-keisler}}}   \\
\T \text{ has QF-Int}& \rImplies^{\T
  \text{universal}}_{\text{Thm.}\ref{thn:univ-GI-QE}} &    \T^* \text{ has
  QE} & \rEquivalent_{\text{Thm.}\ref{mc-qe-iff-univ-fragm-amalg}} & \T_{\forall} \text{ has AP} \\
& \luImplies^{\text{Thm.\ref{thm:mc-th}}} & \dImplies \\
\T \text{ has QF-Int} & \lEquivalent^{\T \text{ universal}}_{\text{Cor.\ref{cor:gi}}}& \T^* \text{ has
  QF-Int} \\
\end{diagram}

\subsection{Symbol Elimination in Theory Extensions}

 We analyzed how this approach can be lifted to 
{\em extensions} of a theory $\T$, by identifying situations in 
which we can use existing methods for symbol elimination in 
$\T$ for symbol elimination or for ground interpolation 
in the extension. If $\T$ has a model completion $\T^*$, 
we analyzed under which conditions we can use possibilities 
of symbol elimination in  $\T^*$ for such tasks. 

\

\noindent The results we obtained are
schematically presented below: 
Assume that $\T_0$ is a theory and $\T = \T_0 \cup \K$ is an extension
of $\T_0$ with additional function symbols, whose properties are 
axiomatized by a set $\K$ of clauses, and $\T_0^*$ a model completion
of $\T_0$ (when applicable). 

\

\noindent \begin{tabular}{@{}|l@{}l|l|}
\hline 
& Condition & Symbol Elimination \\
\hline 
\hline 
$\T_0$~ &~~  allows quantifier elimination   & For every set $G$ of
clauses and every set $T$ of \\
& &  terms  there exists a universal formula $\forall y  \Gamma_T(y)$ \\
& & s.t. ~~~ (*)~~~ $\T_0 \cup \K \cup \forall y \Gamma_T(y) \models \neg G$. \\[2ex]
\cline{2-3}
& ~+ $\T_0 \subseteq \T_0 \cup \K$ sat.  ${\sf Comp}_f$ & $\forall y  \Gamma_T(y)$ weakest universal constraint 
with (*)\\
& ~+ $\K$ flat \& linear& \\ 
\hline 
$\T_0$~ &~~ does not allow quantifier elimination & $\forall y  \Gamma_T(y)$ satisfying (*)
can be obtained \\
& + $\T^*_0$ allows quantifier elimination; & using quantifier elimination in $\T_0^*$ \\
& + every model of $\T_0 \cup \K$ & \\
& ~~ embeds into one of $\T^*_0 \cup
\K$ & \\
\hline 
\end{tabular}

\subsection{Separability, Amalgamation and Ground Interpolation}

In the study of ground interpolation in extensions $\T {\cup} \K$
of a theory $\T$ with a set of clauses $\K$ we 
followed an approach proposed in \cite{Wies, Wies-journal}, in which 
the terms needed to separate the instances of $\K$ are 
considered explicitly.  Our analysis extends
  both the results in \cite{Sofronie-lmcs} and those in \cite{Wies} 
  mainly by avoiding the restriction to convex base theories
  (in \cite{Wies-journal} the formulation is more general) and by
  identifying conditions under which $W$-separability implies
an amalgamation property. 
In addition, when formulating our theorems we explicitly 
pointed out all conditions needed for hierarchical interpolation 
which were missing or only implicit in \cite{Wies}.



\subsection{Future Work}
The results we established in this paper have direct applicability to
the verification of parametric systems. In the future we plan to
further analyze such situations. 
The results about the links between separability, amalgamation and 
ground interpolation we established in Theorem~\ref{thm:Wsep+intW-impl-pamalg} 
use the fact that we assume that the sets 
of terms which need to be used in the separations, for 
equality interpolation, and in the interpolants themselves 
can be described using a closure operator. 
We would like to obtain criteria that 
guarantee the existence of interpolants containing terms that can be 
described using such operators. 
In future work we would like to also extend the approach to
interpolation and symbol elimination described here such that it can
be used for the 
study of {\em uniform interpolation} in logical theories and 
theory extensions.

\section*{Acknowledgments}
I thank the reviewers for their helpful comments.


\end{document}